\documentclass[a4paper,12pt]{article} 
\usepackage{amsmath,amsthm,amssymb}
\newcommand{\bm}[1]{\mbox{\boldmath$#1$}}
\newcommand{\be}{\begin{equation}}
\newcommand{\ee}{\end{equation}}
\newcommand{\bea}{\begin{eqnarray}}
\newcommand{\eea}{\end{eqnarray}}
\newcommand{\non}{\nonumber}
\newtheorem{df}{Definition}

\newtheorem{lem}[df]{Lemma}

\newtheorem{prop}[df]{Proposition}

\newtheorem{lem2}{Lemma}[section]
\newtheorem{prop2}{Proposition}[section]
\newtheorem{cor2}{Corollary}[section]

\begin{document}

\title{Form factors of integrable higher-spin XXZ chains 
and the affine quantum-group symmetry 
}

\author{Tetsuo Deguchi$^{1}$\footnote{e-mail deguchi@phys.ocha.ac.jp} 
and Chihiro Matsui$^{2, 3}$
 \footnote{e-mail matsui@spin.phys.s.u-tokyo.ac.jp}}

%\date{}
\date{\today}

\maketitle
\begin{center}  $^{1}$ 
Department of Physics, Graduate School of Humanities and Sciences, 
Ochanomizu University \\
2-1-1 Ohtsuka, Bunkyo-ku, Tokyo 112-8610, Japan
\end{center} 
\begin{center} $^{2}$ 
Department of Physics, Graduate School of Science, 
University of Tokyo \\ 
7-3-1 Hongo, Bunkyo-ku, Tokyo 113-0033, Japan \\
%\end{center} 
%\begin{center} 
$^{3}$ CREST, JST, 4-1-8 Honcho Kawaguchi, Saitama, 332-0012, Japan
\end{center} 
\abstract{We derive exactly scalar products and form factors 
for integrable higher-spin XXZ chains 
through the algebraic Bethe-ansatz method.   
Here spin values are arbitrary and different spins can be mixed. 
We show the affine quantum-group symmetry, 
$U_q(\widehat{sl_2})$, for the monodromy matrix of the XXZ spin chain,   
and then obtain the exact expressions. 
Furthermore, through the quantum-group symmetry we explicitly derive   
the diagonalized forms of the $B$ and $C$ operators 
in the $F$-basis for the spin-1/2 XXZ spin chain, 
which was conjectured in the algebraic Bethe-ansatz 
calculation of the XXZ correlation functions. 
The results should be fundamental in studying 
form factors and correlation functions systematically 
for various solvable models associated with the integrable XXZ spin chains. 
}

%%%%%%%%%%%%%%
%
\newpage 

 \setcounter{equation}{0} 
 \renewcommand{\theequation}{1.\arabic{equation}}
\section{Introduction}

Correlation functions of the spin-1/2 XXZ spin chain have 
attracted much attention in mathematical physics for more than a decade 
\cite{Korepin,Jimbo-Miwa,Review}.   
The multiple-integral representations of XXZ correlation functions 
were first derived in terms of the $q$-vertex operators 
\cite{Nakayashiki}.  Based on the algebraic Bethe ansatz method,  
the determinant expressions \cite{Slavnov} 
of the scalar products and the norms of Bethe ansatz eigenstates 
were reconstructed in terms of the $F$-basis \cite{MS2000}, 
and then the XXZ correlation 
functions are derived under an external magnetic field 
\cite{KMT1999,KMT2000}. The multiple-integral representations 
at zero temperature have been extended into those at finite temperature 
\cite{Goehmann}.  
Furthermore, dynamical structure factors   
have been evaluated  by solving 
the Bethe ansatz equations numerically \cite{Karbach,Sato,Caux,DSF}. 

The Hamiltonian of the spin-1/2 XXZ chain under the periodic 
boundary conditions is given by  
\be 
{\cal H}_{\rm XXZ} =  
{\frac 1 2} \sum_{j=1}^{L} \left(\sigma_j^X \sigma_{j+1}^X +
 \sigma_j^Y \sigma_{j+1}^Y + \Delta \sigma_j^Z \sigma_{j+1}^Z  \right) \, . 
\label{hxxz}
\ee
Here $\sigma_j^{a}$ ($a=X, Y, Z$) are the Pauli matrices defined 
on the $j$th site, and $\Delta$ the XXZ coupling. 
By $\Delta= (q+q^{-1})/2$ we define parameter $q$, 
which plays a significant role in the paper.  
We note that $L$ denotes the number of the one-dimensional lattice sites. 
Solvable higher-spin generalizations of the XXX and XXZ chains have been 
constructed by the fusion method in several references 
\cite{FusionXXX,Babujan,Fateev-Zamolodchikov,SAA,KR,V-WDA}. 
For instance, the Hamiltonian of the 
solvable spin-$s$ XXX chain is given by the following \cite{Babujan}: 
\be 
H_s = J \sum_{j=1}^{L} Q_{2s}({\vec S}_j \cdot {\vec S}_{j+1}) \, , 
\ee
where ${\vec S}_j$ are operators of spin $s$ acting on the $j$th site and 
$Q_{2s}(x)$ is a polynomial of degree $2s$ 
\be 
Q_{2s}(x) = \sum_{p=1}^{2s} \left( \sum_{k=1}^{p} {\frac 1 k} \right)
\prod_{\ell=0, \ell \ne j}^{2s} {\frac {x - x_{\ell}} {x_j - x_{\ell}}} \, . 
\ee
Here $x_{\ell} = [\ell(\ell+1) -2s (s+1)]/2$.  
At $T=0$ in the critical regime, it is discussed that 
the low-excitation spectrum of the spin-$s$ XXX chain 
is described in terms of the level-$k$ $SU(2)$ WZWN model 
where $k=2s$ \cite{JSuzuki}.

In the present paper we derive exact expressions of scalar products and 
form factors for the integrable higher-spin XXZ spin chains. 
Here different spins can be mixed. 
We first show that the monodromy matrix of the XXZ spin chain has 
the symmetry of the affine quantum group $U_q({\widehat{sl_2}})$. 
Then, we derive the exact expressions 
taking advantage of the quantum group symmetry. 
By a similarity transformation \cite{AW}, 
we transform the symmetric $R$-matrix into an asymmetric one,  
which is directly connected to the quantum group $U_q(sl_2)$. 
We derive projection operators from the asymmetric $R$-matrices \cite{V-WDA},  
and construct integrable higher-spin   
XXZ spin chains by the fusion method similarly as the case of 
the XXX spin chain \cite{FusionXXX}.  
Here we make an extensive use of 
the $q$-analogues of Young's projection operators, which play a central role  
in the $q$-analogue of the Schur-Weyl reciprocity 
of the quantum group $U_q(sl_2)$ \cite{Jimbo-QG,Jimbo-Hecke}. 
Hereafter, we call transformations on the $R$-matrix 
 gauge transformations. 
After we construct integrable higher-spin models,  
applying the inverse transformation to them, 
we reduce the asymmetric monodromy matrices  
into those of the symmetric $L$-operators 
constructed from the symmetric $R$-matrices, 
and thus obtain scalar products and form factors 
for the integrable higher-spin models in the standard formulation.

Form factors and correlation functions 
have been discussed for integrable higher-spin XXX models 
in previous researches 
\cite{Terras1999,MT2000,Kitanine2001,Castro-Alvaredo}.  
In the approach of the so-called quantum inverse scattering problem 
for the integrable $N$-state models, 
one has to construct the $N$-by-$N$ monodromy matrix   
in order to express local operators in terms of the global operators.  
However, it seems to be technically 
nontrivial to construct the $N$-by-$N$ monodromy matrix 
for the integrable higher-spin systems (see also \cite{Martins}). 
By the approach of the present paper, the calculational task is much   
reduced into the minimal level. 
In fact, as a consequence of the $q$-analogue of the Schur-Weyl duality,  
the exact expressions of scalar products and form factors 
are derived from the formulas of the spin-1/2 XXZ chain 
by setting the inhomogeneous parameters in the form 
of ``complete $\ell$-strings''\cite{Odyssey}, and 
the most general results are  straightforwardly obtained.

The affine quantum-group invariance has another important consequence. 
In the XXZ case we can explicitly prove the pseudo-diagonalized forms 
of the $B$ and $C$ operators of the algebraic Bethe-ansatz method 
through the quantum-group symmetry. 
Here we remark that the XXZ spin chain has no spin SU(2) symmetry.
In the pseudo-diagonal basis, 
the $B$ and $C$ operators are expressed as sums of local 
spin operators $\sigma_i^{-}$ and $\sigma_i^{+}$ 
multiplied by diagonal matrices, respectively,   where each of the local 
spin operators $\sigma_i^{\pm}$ are defined on one lattice-site, i.e. 
on the $i$th site. 
The pseudo-diagonalized forms of the $B$ and $C$ operators were conjectured 
in the algebraic Bethe-ansatz derivation of the XXZ correlation functions 
\cite{MS2000,KMT1999,KMT2000}. 
In fact, the $B$ operators create the Bethe states where the $C$ operators 
are conjugate to them, and the pseudo-diagonalized forms 
play a central role in the calculation of the scalar products 
and the norms of Bethe states. 
They are also fundamental in the quantum inverse 
scattering problem, by which local operators are expressed in terms of 
the global operators such as the $B$ and $C$ operators. 
In the XXZ case, however, an explicit derivation 
of the pseudo-diagonalized forms of the $B$ and $C$ operators 
has not been shown previously, yet.  
Thus, the explicit derivation in the paper 
completes the algebraic Bethe-ansatz formulation  
of the form factors and correlation functions 
of the integrable XXZ spin chains. 
Here we remark that for the XXX case (i.e. the isotropic case), 
the diagonalized forms have been shown by Maillet and Sanchez de Santos 
in Ref. \cite{MS2000} by making an explicit 
use of the rotational SU(2)-symmetry. 
However,  the method for the XXX case does not hold 
for the XXZ spin chain which has no SU(2) symmetry.

The derivation of the affine quantum-group symmetry 
of the monodromy matrix should be 
not only theoretically interesting but also practically useful 
for calculation. 
Here we remark that the infinite-dimensional symmetry, 
$U_q({\widehat{sl_2}})$, was realized for the infinite XXZ spin chain 
with the $q$-vertex operators \cite{Jimbo-Miwa,Nakayashiki}. 
Thus, it should be an interesting 
open problem how the affine quantum-group symmetry 
of the finite XXZ spin chain 
can be related to that of the infinite XXZ spin chain. 
 Furthermore, there are several advantages in the present formulation   
of the affine quantum-group symmetry. 
We derive the symmetry  through gauge transformations.  
The transformed asymmetric $R$-matrix is directly related to 
the quantum group $U_q(sl_2)$ so that 
 we can  systematically construct higher-spin representations of the 
$R$-matrices through the $q$-analogue of the Young symmetrizers.     
Here we should note that the gauge transformation connects 
the $R$ matrix in the different gradings of $U_q(\widehat{sl_2})$. 
The symmetric and asymmetric $R$-matrices are equivalent to  
that of the principal and homogeneous gradings, respectively 
(see for instance \S 5.4 of \cite{Jimbo-Miwa}).     
Moreover, we thus avoid technical difficulties appearing 
when we directly derive the matrix representation of 
the universal $R$-matrix of the affine quantum group,  
which is given by a product of infinite series of generators 
\cite{Drinfeld}. 
Although one can construct matrix representations 
of the modified universal $R$-matrix with its derivation $d$ 
dropped \cite{Drinfeld,Jimbo-QG}(see A.2 of \cite{Jimbo-review}), 
it seems that the calculation is not quite straightforward 
when we construct higher-spin representations.

The results of the present paper should be useful 
for calculating exact expressions of correlation functions for 
various integrable models associated with higher-spin XXZ chains.  
For instance, the $\tau_2$ model in the $N$-dimensional 
nilpotent representation corresponds 
to the integrable spin-$(N-1)/2$ XXZ spin chain with $q$ being  
a primitive $N$th root of unity \cite{Nishino}. 
Here we remark that the $\tau_2$ model is closely 
related to the $N$-state superintegrable chiral Potts model.  
Furthermore, there are several possible physical applications, 
such as calculating form factors of  
quantum impurity models  through the result in the case of mixed spins. 
As an illustrative example, we have calculated    
exact expressions for the emptiness formation probability   
of the higher-spin XXZ spin chains, 
which we shall discuss in a subsequent paper.

The content of the paper consists of the following: 
In section 2, we introduce the symmetric $R$-matrix of the spin-1/2 
XXZ spin chain, and define the monodromy matrix \cite{TF}.  
We also define the action of 
the symmetric group on products of $R$-matrices.  
In section 3 we derive the symmetry of the quantum affine 
algebra for the monodromy matrix of the XXZ spin chain.  
We introduce the asymmetric $R$-matrix and then derive it 
from the symmetric one by a gauge transformation.  
We decompose the asymmetric $R$-matrix in terms of the 
generators of the Temperley-Lieb algebra, and  
show the affine quantum-group symmetry. 
We also derive it by a systematic method for expressing 
products of $R$-matrices  formulated in definition \ref{df:Rp}. 
In fact,  all the fundamental relations 
of the quantum inverse-scattering problem 
can be derived much more simply 
 without using the ${\hat R}$-matrix of Ref. \cite{MS2000},   
as shown in Appendices A and B. 
In section 4 we construct the $R$ matrices of  
integrable higher-spin XXZ spin chains 
with projection operators of $U_q(sl_2)$ by the fusion method.  
We also discuss the case of mixed spins. 
In section 5 we formulate an explicit derivation of the 
pseudo-diagonalized forms of the $B$-operators. 
We also show it for the $C$-operator in Appendix E. 
In section 6 we derive 
determinant expressions of scalar products 
for the higher-spin XXZ spin chains.    
In section 7, for the higher-spin cases 
we show the method by which we can express local operators 
  in terms the global operators. 
We give some useful formulas of the quantum inverse scattering problem 
for the higher spin case. Finally, we derive some examples of 
form factors for the integrable higher-spin XXZ spin chains.

\newpage 
%%%%%%%%%%%%%%%%%%%%%%%%%%%%%%%%%%%%%%%%%%%%%%%%%%%%%%%%%%%
%
%                  SECTION 2 
%
%%%%%%%%%%%%%%%%%%%%%%%%%%%%%%%%%%%%%%%%%%%%%%%%%%%%%%%%%%%
 \setcounter{equation}{0} 
 \renewcommand{\theequation}{2.\arabic{equation}}
\section{$R$-matrices and $L$-operators}

\subsection{Symmetric $R$-matrix}

We shall introduce the $R$-matrix for the XXZ spin chain \cite{TF}. 
We consider two types of $R$-matrices, 
$R_{ab}(u)$ and $R_{ab}(\lambda, \mu)$.  
The $R$-matrix with a single rapidity argument,  
$R_{ab}(u)$, acts on the tensor product of 
 two vector spaces $V_a$ and $V_b$, 
i.e. $R_{ab}(u) \in End(V_a \otimes V_b)$,  
%(which are often given by representations of $U_q(sl_2)$) 
where parameter $u$ is independent of $V_a$ or $V_b$. 
The $R$-matrix with two rapidity arguments,  
$R_{ab}(\lambda, \mu)$,  
acts on the tensor product of 
vector spaces with parameters,  
%(evaluation representations)  
$V_a(\lambda)$ and $V_b(\mu)$, i.e.  
$R_{ab}(\lambda, \mu) \in End(V_a(\lambda) \otimes V_b(\mu))$.

Let us denote by $e^{a,b}$ such a matrix 
that has only one nonzero element 
equal to 1 at entry $(a,b)$. 
We denote by $V$ the two-dimensional vector space. 
We define the $R$-matrix acting on the tensor product $V \otimes V$ by 
\be
R(u)=\sum_{a,b,c,d=1,2} R^{ab}_{cd}(u) e^{a,c} \otimes e^{b,d} \, . 
\label{eq:R(u)}
\ee
Here matrix elements $R^{ab}_{cd}(u)$ satisfy the 
charge conservation, i.e.   
$R^{ab}_{cd}(u) = 0$ unless $a+b=c+d$, and 
all the nonzero elements are given by the following: 
\bea 
R^{11}_{11}(u)& = & R^{22}_{22}(u)=1 \, , \quad 
R^{12}_{12}(u)=R^{21}_{21}(u)= b(u), \non \\
R^{12}_{21}(u) & =& R^{21}_{12}(u)= c(u) \, , \label{R-matrix-elements}
\eea
where functions $b(u)$ and $c(u)$ are given by  
\be 
b(u)= {\frac {\sinh(u)} {\sinh(u+\eta)}} \, , \quad
c(u)= {\frac {\sinh(\eta)} {\sinh(u+\eta)}} \, . 
\label{bc}
\ee 
Here, parameter $\eta$ is related to $q$ of $\Delta = (q+q^{-1})/2$ 
by $q=\exp(\eta)$. 

We now introduce operators acting on the $L$th power of tensor product 
of vector spaces with parameters, 
$V(\lambda_1) \otimes \cdots \otimes V(\lambda_L)$. We  
generalize the notation of (\ref{eq:R(u)}). 
Let us take a pair of integers  $j$ and $k$ satisfying $1 \le j < k \le L$. 
For a given set of matrix elements 
$A^{a, \, b}_{c, \, d}(\lambda_j, \lambda_k)$ ($a,b,c,d=1,2$)  
we define operators $A_{j, k}(\lambda_j, \lambda_k)$ and 
$A_{k,j}(\lambda_k, \lambda_j)$ by 
\bea 
A_{j, k}(\lambda_j, \lambda_k) & = & \sum_{a,b,\alpha, \beta=1,2} 
A^{a, \, \alpha}_{b, \, \beta}(\lambda_j, \lambda_k) 
I_1 \otimes \cdots \otimes I_{j-1} \non \\ 
& & \quad \otimes e^{a, b}_j \otimes 
I_{j+1} \otimes \cdots \otimes I_{k-1} 
\otimes e_k^{\alpha, \beta}  \otimes I_{k+1} \otimes \cdots \otimes I_{L} 
\, , \non \\  
A_{k, j}(\lambda_k, \lambda_j) & = & 
\sum_{a, b, \alpha, \beta=1,2} 
A^{\alpha, \, a}_{\beta, \, b}(\lambda_k, \lambda_j) 
I_1 \otimes \cdots \otimes I_{j-1} \non \\ 
& & \quad \otimes e^{a, b}_j \otimes I_{j+1} \otimes \cdots \otimes I_{k-1} 
\otimes e_k^{\alpha, \beta} \otimes  I_{k+1} 
\otimes \cdots \otimes I_{L} \, . \label{defAjk}
\eea
Here $I$ is the two-by-two unit matrix, and  
$I_j$ and $e_j^{a, b}$ act on the $j$th vector space $V(\lambda_j)$ of 
$V(\lambda_1) \otimes \cdots \otimes V(\lambda_L)$.  
In terms of matrices we express operators $A_{j, k}$ and 
$A_{k, j}$ for $j < k$ by  
\be 
A_{j, k} = 
\left(
\begin{array}{cccc} 
A^{11}_{11} &   A^{11}_{12} & A^{11}_{21} &   A^{11}_{22} \\
A^{12}_{11} &   A^{12}_{12} & A^{12}_{21} &   A^{12}_{22} \\
A^{21}_{11} &   A^{21}_{12} & A^{21}_{21} &   A^{21}_{22} \\
A^{22}_{11} &   A^{22}_{12} & A^{22}_{21} &   A^{22}_{22} 
\end{array} 
\right)_{[j,k]} \, , \quad 
A_{k, j} = 
\left(
\begin{array}{cccc} 
A^{11}_{11} &   A^{11}_{21} & A^{11}_{12} &   A^{11}_{22} \\
A^{21}_{11} &   A^{21}_{21} & A^{21}_{12} &   A^{21}_{22} \\
A^{12}_{11} &   A^{12}_{21} & A^{12}_{12} &   A^{12}_{22} \\
A^{22}_{11} &   A^{22}_{21} & A^{22}_{12} &   A^{22}_{22} 
\end{array} 
\right)_{[j,k]} 
\ee
Here by the symbol $[j,k]$ we express that 
matrix element $A^{ab}_{cd}$ corresponds to 
$e_j^{a,c} \otimes e_k^{b,d}$ for 
$A_{j, k}$, and to $e_j^{b,d} \otimes e_k^{a,c}$ for $A_{k, j}$.

Let us now introduce operators $R_{j, k}(\lambda_j, \lambda_k)$ 
and $R_{k, j}(\lambda_k, \lambda_j)$ 
acting on the tensor product $V(\lambda_1) 
\otimes \cdots \otimes V(\lambda_L)$.
We define them by putting $A^{ab}_{cd}(\lambda_j, \lambda_k) = 
R^{ab}_{cd}(\lambda_j-\lambda_k)$ in (\ref{defAjk}).  
Here the matrix elements $R^{ab}_{cd}(u)$ 
are given in (\ref{R-matrix-elements}).      
For instance, setting $u=\lambda_1 - \lambda_2$, 
we have explicitly
\be
R_{12} (\lambda_1, \lambda_2)  
= \left(
\begin{array}{cccc} 
1 &   0 & 0 & 0 \\
0 &   b(u) & c(u) & 0 \\
0 &   c(u) & b(u) & 0 \\
0 &   0 & 0 & 1
\end{array} 
\right)_{[1,2]} \, .     
\ee

The $R$-matrices satisfy the Yang-Baxter equations: 
\be 
R_{12}(\lambda_1, \lambda_2) R_{13}(\lambda_1, \lambda_3) 
R_{23}(\lambda_2, \lambda_3) = R_{23}(\lambda_2, \lambda_3) 
R_{13}(\lambda_1, \lambda_3) R_{12}(\lambda_1, \lambda_2)  
\ee
They also satisfy the inversion relations (unitarity conditions): 
\be 
R_{j k}(\lambda_j, \lambda_k) R_{kj}(\lambda_k, \lambda_j)= I^{\otimes L} 
\qquad {\rm for } \, \, 1 \le j, k \le L \, . 
\label{eq:inv-rel}
\ee
Here $I^{\otimes L}$ denotes the $L$th power of tensor product of $I$.

Hereafter we often abbreviate $R_{jk}(\lambda_j, \lambda_k)$ 
simply by $R_{jk}$.

%%%%%%%%%%%%%%%%%%%%%%%%%%%%%%%%%%%%%%%%%%%%%%%%%%%%%%%%%%%%%%
%
%
\subsection{$L$-operators and the monodromy matrix}

Let us introduce parameters $\xi_1, \xi_2, \ldots, \xi_L$, 
which we call the inhomogeneous parameters. 
In the case of the monodromy matrix, we 
assume that parameters $\lambda_j$ of the tensor product 
$V(\lambda_1) \otimes \cdots \otimes V(\lambda_L)$ 
are given by the inhomogeneous parameters, i.e.  
$\lambda_j = \xi_j$ for $j=1, 2, \ldots, L$. 
Let us denote by $0$ the suffix for the auxiliary space. 
We define $L$-operators acting on the $m$th site 
for $m=1, 2, \ldots, L$, by 
\be 
L_{m}(\lambda, \xi_m) = R_{0 m}(\lambda, \xi_m) \, . 
\ee
We define the monodromy matrix  
acting on the $L$ lattice-sites in one dimension by 
\be 
T_{0, 1 2 \cdots L}(\lambda; \xi_1, \ldots, \xi_L) = L_{L}(\lambda, \xi_L) 
L_{L-1}(\lambda, \xi_{L-1}) \cdots 
L_{2}(\lambda, \xi_2) L_{1}(\lambda, \xi_1) \, . 
\label{eq:monodromy}
\ee
We shall also denote it by 
$R_{0, 1 2 \cdots L}(\lambda_0; \xi_1, \ldots, \xi_L)$ in \S 2.4.    
Hereafter we often suppress the symbols of 
inhomogeneous parameters 
and express the monodromy matrix 
$T_{0, 1 2 \cdots L}(\lambda; \xi_1, \ldots, \xi_L)$ 
simply as  $R_{0, 1 2 \cdots L}(\lambda; \{ \xi_j \})$ or $T_{0}(\lambda)$. 
 
Let us consider two auxiliary spaces with 
suffices $a$ and $b$. We define monodromy matrices 
$T_{a}(\lambda_{a})$ and $T_{b}(\lambda_b)$ similarly as 
(\ref{eq:monodromy}) with $0$ replaced by $a$ and $b$, respectively.    
It is clear that they 
satisfy the following Yang-Baxter equations. 
\be 
R_{ab}(\lambda_a, \lambda_b) T_a(\lambda_a) T_b(\lambda_b) 
= T_b(\lambda_b) T_a(\lambda_a) R_{ab}(\lambda_a, \lambda_b) \, . 
\label{eq:RTT=TTR}
\ee

Let us introduce operator $A_j$ acting on the $j$th site by  
\be 
A_j = \sum_{a,b=1,2} A^{a}_{b} I_0 \otimes \cdots \otimes I_{j-1} \otimes 
e^{a, \, b}_j \otimes \cdots \otimes I_L
\ee
We express it in terms of the matrix notation as follows: 
\be 
A_j = 
\left(
\begin{array}{cc} 
A^{1}_{1} &   A^{1}_{2}  \\
A^{2}_{1} &   A^{2}_{2}  
\end{array} 
\right)_{[j]} 
\ee
We express the matrix elements of the monodromy matrix by 
\be 
T_{0, 1 2 \cdots L}(u; \xi_1, \ldots, \xi_L) = 
\left(
\begin{array}{cc} 
A_{12 \cdots L}(u; \xi_1, \ldots, \xi_L) &   
B_{12 \cdots L}(u; \xi_1, \ldots, \xi_L)  \\
C_{12 \cdots L}(u; \xi_1, \ldots, \xi_L) &   
D_{12 \cdots L}(u; \xi_1, \ldots, \xi_L)  
\end{array} 
\right)_{[0]} 
\ee
The transfer matrix, $t(u)$, is given by 
the trace of the monodromy matrix with respect to the 0th space:  
\bea 
t(u; \xi_1, \ldots, \xi_L) & = & 
{\rm{tr}_0}\left(T_{0, 1 2 \cdots L}(u; \xi_1, \ldots, \xi_L) \right) 
\non \\ 
& = & A_{12 \cdots L}(u; \xi_1, \ldots, \xi_L) 
+ D_{12 \cdots L}(u; \xi_1, \ldots, \xi_L) \, . 
\eea
Here we note that the transfer matrix $t(u)$ is nothing but the 
transfer matrix of the six-vertex model defined on the two-dimensional 
square lattice \cite{Baxter-Book}. 

Hereafter,  
we shall often denote $B_{12 \cdots L}(u; \xi_1, \ldots, \xi_L)$ 
by $B(u; \{ \xi_j \})$ or $B(u)$, briefly.  

%%%%%%%%%%%%%%%%%%%%%%%%%%%%%%%%%%%%%%%%%%%%%%%%%%%%%%%%%%%%%%%%
%
%
\subsection{Products of $R$-matrices and the symmetric group}

Let us  consider the symmetric group ${\cal S}_{n}$ of 
$n$ integers, $1, 2, \ldots, n$. 
We denote by $\sigma$ an element of  ${\cal S}_{n}$.   
Then $\sigma$ maps $j$ to $\sigma(j)$ for $j=1, 2, \ldots, n$

\begin{df}
Let $p$ be a sequence of $n$ integers, $1, 2, \ldots, n$,  
and $\sigma$ an element of the symmetric group ${\cal S}_n$. 
We define the action of $\sigma$ on $p$ by 
\be 
\sigma(p) = (p_{\sigma(1)}, \ldots, p_{\sigma(n)}) \, . 
\ee
\end{df} 
Here we remark that $(\sigma_A \sigma_B) \, p = \sigma_B (\sigma_A p)$ 
for $\sigma_A, \sigma_B \in {\cal S}_n$. We shall show it in Appendix A.  

Let us recall that $R_{jk}$ denote $R_{jk}(\lambda_j, \lambda_k)$.  
\begin{df} 
  Let $p=(p_1, p_2, \ldots, p_n)$ be a sequence of 
$n$ integers, $1, 2, \ldots, n$. 
We define $R_{p_1, \, p_2 p_3 \cdots p_n}$ 
and $R_{p_1 p_2 \cdots p_{n-1}, \, p_n}$ by  
\bea
R_{p_1, \, p_2 p_3 \cdots p_n} & = & 
R_{p_1 p_n} R_{p_1 p_{n-1}} \cdots R_{p_1 p_2} \, , \non \\ 
R_{p_1 p_2 \cdots p_{n-1}, \, p_n} & = & R_{p_1 p_n} R_{p_2 p_n} \cdots 
R_{p_{n-1} p_n} \, . 
\eea
\end{df}
For $p=(1, 2, \ldots, n)$ we have 
\be 
R_{1, 2 3 \cdots n}  =  R_{1 n} R_{1 n-1} \cdots R_{1 2} \, , \quad   
R_{1 2 \cdots n-1, n}  = R_{1 n} R_{2 n} \cdots R_{n-1 n}. 
\ee 
We thus express the monodromy matrix as follows 
\be 
T_{0, 1 2 \cdots L}(\lambda_0; \{ \xi_k \}) 
= R_{0, 1 2 \cdots L}(\lambda_0; \{ \xi_k \}) \, . 
\ee
Here we have assumed that $\lambda_k = \xi_k$ for $k=1, 2, \ldots, n$.  

Let us express by $s_j=(j \, \, j+1)$ such a permutation that 
maps $j$ to $j+1$ and $j+1$ to $j$ and does not change other integers.   
\begin{df} Let $p$ be a sequence of $n$ integers, $1, 2, \ldots, n$.  
We define $R^{s_j}_{p}$ by  
\be 
R^{s_j}_{p} = R_{p_j, p_{j+1}}(\lambda_{p_j}, \lambda_{p_{j+1}}) \, .   
\ee
For the unit element $e$ of ${\cal S}_{n}$,   
we define $R^{e}_{p}$ by $R^{e}_{p}=1$. For a given element 
$\sigma$ of ${\cal S}_n$,    
we define $R^{\sigma}_{p}$ recursively by the following: 
\be 
R_{p}^{\sigma_A \sigma_B} = R_{\sigma_A(p)}^{\sigma_B} R_{p}^{\sigma_A} \, .     \label{eq:Rp}
\ee  
\label{df:Rp}
\end{df} 

We remark that 
every permutation $\sigma$ is expressed 
as a product of some $s_j=(j \, j+1)$ with 
$j=1, 2, \ldots, n-1$.   
We thus obtain $R_{p}^{\sigma}$ 
as a product of $R_{p}^{s_j}$ for some $j$ s.  
For an illustration,  let us calculate $R^{(123)}_{(1, 2, 3)}$.  
Noting $(123)=(1 \, 2)(2 \, 3)$, we have 
\bea 
R^{(12)(23)}_{(1,2,3)} & = & R^{(23)}_{(2,1,3)} R^{(12)}_{(1,2,3)} \non \\ 
& = & R_{13} R_{12} = R_{1, 23} \, . 
\eea
Through the defining relations of the symmetric group ${\cal S}_n$ 
\cite{Magnus},  we can show that definition \ref{df:Rp} is well defined.
The proof is given in proposition A.1 of Appendix A.

Let us denote by $\sigma_c$ such a cyclic permutation that  
maps $j$ to $j+1$ for $j=1, \ldots, n-1$ and $n$ to 1. 
We also express it as $\sigma_c=(1 2 \cdots n )$. 
Noting $(1 2 \cdots n)=(1\, \,  2) \cdots (n-1 \,\,  n) 
= s_1 s_2 \cdots s_{n-1}$, 
we can show  the following lemma 
\begin{lem} Let us denote by $p_q$ the sequence $p_q=(1,2, \ldots, n)$. 
For  $\sigma_c=(1 2 \cdots n)$ we have 
\be 
R^{\sigma_c}_{p_q} = R_{1, 2 \cdots n} \, . 
\ee
\label{lem:cyclic}
\end{lem} 
The proof of lemma \ref{lem:cyclic} 
is given in lemma  \ref{lem:p-cyclic} of Appendix A.  

%%%%%%%%%%%%%%%%%%%%%%%%%%%%%%%%%%%%%%%%%%%%%%%%%%%%%%%%%%%
%
%
\subsection{${\check R}$-matrices and permutation operators}

Let us consider  two-dimensional 
vector spaces $V_a$ and $V_b$. 
We define permutation operator  $\Pi_{ab}$ which maps elements of 
$V_a \otimes V_b$ to those of $V_b \otimes V_a$ as follows. 
\be 
\Pi_{ab} \, v_a \otimes v_b = v_b \otimes v_a  \, , \quad 
v_a \in V_a \, , \quad v_b \in V_b \, . 
\ee
We define ${\check R}_{ab}(u)$ by 
\be 
{\check R}_{ab}(u) = \Pi_{ab} R_{ab}(u) 
\ee
The operators ${\check R}_{ab}$ satisfy the Yang-Baxter equations  
\be 
{\check R}_{12}(u){\check R}_{23}(u+v){\check R}_{12}(v) 
={\check R}_{23}(v){\check R}_{12}(u+v){\check R}_{23}(u) 
\label{eq:YB-Rcheck}
\ee
The operator ${\check R}_{ab}$ gives a linear map 
from $V_a \otimes V_b$ to $V_b \otimes V_a$.  
If $V_a$ and $V_b$ are equivalent, then 
we may regard ${\check R}_{ab}$ as a map from 
$V_a^{\otimes 2}$ to $V_a^{\otimes 2}$.

We add 0 to the $n$ integers. 
For a given element $\sigma$ 
of the symmetric group ${\cal S}_{n+1}$, 
we define $\Pi^{\sigma}$ acting 
on integers $0, 1, \ldots, n$,  
as follows.    
We first express $\sigma$ in terms of $s_j=(j \,\, j+1)$  
such  as $\sigma=s_{j_1} s_{j_2} \cdots s_{j_r}$, and   
then we define $\Pi^{\sigma}= \Pi^{(j_1 \, j_1+1) \cdots (j_r \, j_r+1)}$ 
by 
\be 
\Pi^{\sigma} 
 = \Pi_{j_1, j_1+1} \Pi_{j_2, j_2+1} \cdots \Pi_{j_r, j_r+1} \, . 
\ee

\begin{lem} We have the following relation between 
$R$-matrices and operators ${\check R}_{ab}$:  
\be 
R_{0, 1 2 \cdots n} = \Pi^{(0 1 \cdots n)} \, 
{\check R}_{n-1 \, n}(\lambda_0-\xi_{n-1}) \cdots 
{\check R}_{12}(\lambda_0-\xi_{2}) 
{\check R}_{01}(\lambda_0-\xi_{1}) 
\label{eq:product-Rcheck}
\ee
\end{lem}

%%%%%%%%%%%%%%%%%%%%%%%%%%%%%%%%%%%%%%%%%%%%%%%%%%%%%%%%%%
%
%          SECTION 3 
%
%%%%%%%%%%%%%%%%%%%%%%%%%%%%%%%%%%%%%%%%%%%%%%%%%%%%%%%%%%
 \setcounter{equation}{0} 
 \renewcommand{\theequation}{3.\arabic{equation}} 
\section{The quantum group invariance}

We shall show that the monodromy matrix, 
$T_{0, 1 2 \cdots L}(\lambda; \xi_1, \ldots, \xi_L)$, 
has the symmetry of the affine quantum group, $U_q(\widehat{sl_2})$.   

\subsection{Quantum group $U_q(sl_2)$ and the asymmetric $R$-matrices}

The quantum algebra $U_q(sl_2)$ 
is an associative algebra over ${\bf C}$ generated by  
$X^{\pm}, K^{\pm}$  with the following relations: \cite{Jimbo-review}
\bea 
K K^{-1} & = & K K^{-1} = 1 \, , \quad 
K X^{\pm} K^{-1}  =  q^{\pm 2} X^{\pm} \, ,  \quad 
\, , \non \\
{[} X^{+}, X^{-} {]} & = &  
{\frac   {K - K^{-1}}  {q- q^{-1}} } \, . 
\eea
The algebra $U_q(sl_2)$ is also a Hopf algebra over ${\bf C}$ 
with comultiplication 
\bea 
\Delta (X^{+}) & = & X^{+} \otimes 1 + K \otimes X^{+}  \, , 
 \quad 
\Delta (X^{-})  =  X^{-} \otimes K^{-1} + 1 \otimes X^{-} \, ,  \non \\
\Delta(K) & = & K \otimes K  \, , 
\eea 
and antipode:  
$S(K)=K^{-1} \, , S(X^{+})= - K^{-1} X^{+} \, , S(X^{-}) = -  X^{-} K$, and   
coproduct: $\epsilon(X^{\pm})=0$ and $\epsilon(K)=1$. 

In association with the quantum group, 
we define the $q$-integer of an integer $n$ by 
$[n]_q=(q^n -q^{-n})/(q-q^{-1})$. 

The universal $R$-matrix, ${\cal R}$, of $U_q(sl_2)$ 
satisfies the following relations: 
\be 
{\cal R} \Delta (x) = 
\tau \circ \Delta (x)  {\cal R} 
\quad {\rm for \, all} \quad x \in U_q(sl_2) \, .  
\label{universalR}
\ee 
Here $\tau$ denotes a permutation such that 
$\tau \, a \otimes b = b \otimes a$ for $a, b \in U_q(sl_2)$. 

We now introduce some notation of a Hopf algebra.  
Let $x(1), \ldots, x(n)$ be elements of Hopf algebra ${\cal A}$.  
For a given permutation $\sigma$ of ${\cal S}_n$,  
we define its action 
on the tensor product $x(1) \otimes \cdots \otimes x(n)$ as follows: 
\be 
\sigma \circ \left( x(1) \otimes \cdots \otimes x(n) \right) 
=  x(\sigma^{-1} 1) \otimes \cdots x(\sigma^{-1} n) \, . 
\ee
We note that  ${\cal A}$ has coassociativity: 
$\left( \Delta \otimes id \right) \Delta(x) = 
\left( id \otimes \Delta  \right) \Delta(x)$ 
for any element $x$ of  ${\cal A}$. 
We therefore denote it by $\Delta^{(2)}(x)$. We  
define $\Delta^{(n)}(x)$ recursively by  
\be 
\Delta^{(n)}(x) = \left(\Delta^{(n-1)} \otimes id \right) \Delta(x) \quad   
{\rm for}  \, \,  x \in {\cal A} . 
\ee

Let us now introduce the following asymmetric $R$-matrices:   
\be
R^{\pm}(u)  
= \left(
\begin{array}{cccc} 
1 &   0 & 0 & 0 \\
0 &   b(u) & c^{\mp}(u) & 0 \\
0 &   c^{\pm}(u) & b(u) & 0 \\
0 &   0 & 0 & 1
\end{array} 
\right) \, ,  
\ee
where $c^{\pm}(u)$ are defined by 
\be 
c^{\pm}(u)= {\frac {e^{\pm u} \sinh(\eta)} {\sinh(u+\eta)}} \, . 
\label{fg}
\ee
In the spin-1/2 representation of $U_q(sl_2)$, we have 
the following  relations: 
\be 
R_{12}^{+}(u) \Delta (x) = 
\tau \circ \Delta (x)  R^{+}_{12}(u) 
\quad {\rm for \, \,} \quad x= X^{\pm}, K \, .  
\label{eq:RD}
\ee
Here we remark that spectral parameter $u$ is arbitrary and 
independent of $X^{\pm}$ or $K$. 
Similarly as in the symmetric case,  
we define the monodromy matrix $R^{+}_{0, 1 \cdots n}$ 
by $R^{+}_{0, 1 \cdots n}=R^{+}_{0, n} \cdots R^{+}_{0, 1}$,  
and ${\check R}^{+}$ by ${\check R}^{+}_{12}(u) = \Pi_{12} R^{+}(u)$.  
\begin{lem} 
The monodromy matrix expressed in terms of ${\check R}$'s 
commutes with the action of the quantum group $U_q(sl_2)$:   
\be 
{[} {\check R}^{+}_{L-1, L}(\lambda-\xi_L) \cdots 
{\check R}^{+}_{1, 2}(\lambda-\xi_2) 
{\check R}^{+}_{0, 1}(\lambda-\xi_1), \quad \Delta^{(L)}(x) {]} = 0 \, , \quad 
{\rm for} \, {\rm all} \, x \in U_q(sl_2) \, . 
\ee 
Here parameters $\lambda$, $\xi_1, \ldots, \xi_L$ are independent of 
element $x$ of $U_q(sl_2)$. 
\label{lem:Rcheck-q-inv}
\end{lem}
We shall show lemma \ref{lem:Rcheck-q-inv} and eqs. (\ref{eq:RD}) 
through the Temperley-Lieb algebra in \S 3.2.  

Lemma \ref{lem:Rcheck-q-inv}  leads to the following 
symmetry relations of the monodromy matrix $R^{+}_{0, 1 2 \cdots L}$ 
with respect to the quantum group $U_q(sl_2)$: 
\begin{prop} 
Let $\sigma_c$ be a cyclic permutation:  $\sigma_c=(0 1 \cdots L)$. 
Then we have    
\be 
R^{+}_{0, 1 2 \cdots L}(\lambda; \xi_1, \ldots, \xi_L) \, 
\Delta^{(L)}(x) = \sigma_c \circ \Delta^{(L)}(x) \, 
 R^{+}_{0, 1 2 \cdots L}(\lambda; \xi_1, \ldots, \xi_L)  
\quad {\rm for \, \, all} \, \,  x \in U_q(sl_2) 
\label{eq:q-inv(+)}
\ee
Here parameters $\lambda$, $\xi_1, \ldots, \xi_L$ are independent of 
element $x$ of $U_q(sl_2)$. 
\label{prop:q-inv(+)}
\end{prop} 
\begin{proof} 
Making use of lemma \ref{lem:Rcheck-q-inv} we show (\ref{eq:q-inv(+)}) 
 from (\ref{eq:product-Rcheck}) and the following relation:  
\be 
\sigma_c \circ \Delta^{(L)}(x) = \Pi^{\sigma_c} \, \Delta^{(L)}(x) \, 
\left( \Pi^{\sigma_c} \right)^{-1} \, . 
\ee
\end{proof}

%The relations (\ref{eq:q-inv(+)}) are nothing but 
%the affine quantum-group invariance 
%of the monodromy matrix of the XXZ spin chain, 
%as we shall show in section 3.4. 
%

%%%%%%%%%%%%%%%%%%%%%%%%%%%%%%%%%%%%%%%%%%%%%%%%%%%%%%%
%
%
\subsection{Derivation in terms of the Temperley-Lieb algebra} 

Let us define $U_j^{\pm}$ for $j=0, 1, \ldots, L-1$, by 
\be 
U^{\pm}_j = \left( 
\begin{array}{cccc} 
0 & 0 & 0 & 0 \\ 
0 & q^{\mp} & -1 & 0 \\ 
0 & -1 & q^{\pm} & 0 \\ 
0 & 0 & 0 & 0  
\end{array} 
\right)_{[j, j+1]} \, . 
\ee
They satisfy the defining relations of the Temperley-Lieb algebra: 
\cite{Baxter-Book}
\bea 
U^{\pm}_{j} U^{\pm}_{j + 1} U^{\pm}_j & = & U^{\pm}_j, \non \\  
U^{\pm}_{j+1} U^{\pm}_{j } U^{\pm}_{j+1} & = & U^{\pm}_j, 
\quad {\rm for} \, j=0, 1, \ldots, L-2, 
\non \\ 
\left( U^{\pm}_j \right)^2 
& = &  (q+q^{-1}) \, U^{\pm}_j \quad {\rm for} 
\, j=0, 1, \ldots, L-1, 
\non \\  
U^{\pm}_j U^{\pm}_k & = & U^{\pm}_k U^{\pm}_j 
\quad {\rm for} \, \, |j-k| > 1\, .  
\eea

Let us now show commutation relations (\ref{eq:q-inv(+)}), 
making use of the Temperley-Lieb algebra. 
The operator ${\check R}_{j, j+1}^{+}(u)$ is decomposed  
in terms of the generators of the Temperley-Lieb algebra as follows 
\cite{Baxterization}. 
\be 
{\check R}_{j,j+1}^{+}(u) = I - b(u) U_{j}^{+} \, . 
\label{eq:TLdecomp+}
\ee 
\begin{lem} The monodromy matrix of the six-vertex model 
is expressed in terms of the generators of the Temperley-Lieb algebra 
as follows.  
\bea 
& & {\check R}^{+}_{L-1, L}(\lambda-\xi_L) \cdots 
{\check R}^{+}_{1, 2}(\lambda-\xi_2) 
{\check R}^{+}_{0, 1}(\lambda-\xi_1) \non \\ 
& & =  \sum_{k=0}^{L} (-1)^k \sum_{0 \le i_1 < \cdots < i_k < L} 
\, \left( \prod_{j=1}^{k} b(\lambda-\xi_{i_j}) \right) 
U^{+}_{i_k}  \cdots U^{+}_{i_2} U^{+}_{i_1} \, .  
\eea
\label{lem:expU}
\end{lem}

\begin{lem}
The generators $U^{+}_j$ commute with the generators of 
$U_q(sl_2)$.  For $x=X^{\pm}, K$ 
and for $j=0, 1, \ldots, L-1$, we have 
in the tensor-product representation  
\be 
{\big[} U^{+}_j, \Delta^{(L)}(x) {\big]} = 0 \, . 
\label{eq:Uq-inv}
\ee
\label{lem:Uq-inv}
\end{lem} 

\par \noindent  
{\it Proof of proposition \ref{prop:q-inv(+)}. }
%\begin{proof} 
From lemmas \ref{lem:expU} and \ref{lem:Uq-inv} 
we have lemma \ref{lem:Rcheck-q-inv}, which is equivalent to 
proposition  \ref{prop:q-inv(+)}.  
%\end{proof} 

We now show that in the limit of taking $u$ to $- \infty$,    
${\hat R}^{+}(u)$ is equivalent to  
the spin-1/2 matrix representation of the universal $R$-matrix ${\cal R}$ 
of $U_q(sl(2))$. 
An explicit expression of ${\cal R}$ is given by  
\be 
{\cal R} = q^{- {\frac 1 2 } H \otimes H} \exp_q \left(- (q-q^{-1}) 
K^{-1} X^{+} \otimes X^{-} K \right) 
\label{eq:universal-R}
\ee
where $\exp_q x$ denotes the following series: 
\be 
\exp_q x = \sum_{n=0}^{\infty} {\frac {q^{-n(n-1)/2}} {[n]_q!}} x^n \, .  
\ee
Here $q$ is generic. We recall that
 $[n]_q$ denotes the $q$-integer of an integer $n$: 
$[n]_q=(q^n -q^{-n})/(q-q^{-1})$.
Putting $X^{+} = e^{1,2}$, $X^{-} = e^{2,1}$ 
 and $K= {\rm diag}(q, q^{-1})$ in the series 
(\ref{eq:universal-R}),  
we have the following matrix representation. 
\be
R_{{\frac 1 2}, {\frac 1 2}}  
= \left(
\begin{array}{cccc} 
q^{-1/2} &   0 & 0 & 0 \\
0 &   q^{1/2} & -q^{1/2}(q-q^{-1}) & 0 \\
0 &   0 & q^{1/2} & 0 \\
0 &   0 & 0 & q^{-1/2}
\end{array} 
\right) \, .  
\ee
Thus, for arbitrary $u$, we have the following:   
\be 
R^{+}_{12}(u) = c^{+}(u) \Pi_{12} + 
q^{-1/2} b(u) R_{{\frac 1 2}, {\frac 1 2}}  
\ee
Therefore, we have 
${R}^{+}(-\infty)= q^{1/2} R_{{\frac 1 2}, {\frac 1 2}} $.

We remark that some relations equivalent to (\ref{eq:Uq-inv}) have been 
shown in association with the $sl(2)$ loop algebra symmetry of 
the XXZ spin chain at roots of unity \cite{DFM}.

%%%%%%%%%%%%%%%%%%%%%%%%%%%%%%%%%%%%
%
%   S 3.3
%
%
\subsection{Gauge transformations}

Let us introduce operators ${\Phi}_j$ 
with arbitrary parameters $\phi_j$ for $j=0, 1, \ldots, L$ as follows:  
\be 
{\Phi}_j = \left( 
\begin{array}{cc}  
1 & 0 \\
0 & e^{\phi_j} 
\end{array}
\right)_{[j]} =  
I^{\otimes (j)} \otimes \left( 
\begin{array}{cc}  
1 & 0 \\
0 & e^{\phi_j} 
\end{array}
\right) \otimes I^{\otimes (L-j)}  . 
\label{eq:Phi}
\ee
In terms of ${\chi}_{jk}= {\Phi}_j {\Phi}_k$,  
we define a similarity transformation on the $R$-matrix by  
\be 
R_{jk}^{\chi} = {\chi}_{jk} R_{jk} {\chi}_{jk}^{-1} 
\ee
Explicitly, the following two matrix elements are transformed. 
\be 
\left( R_{jk}^{\chi} \right)^{21}_{12} = 
c(\lambda_j, \lambda_k) e^{\phi_j-\phi_k} \, , \quad 
\left( R_{jk}^{\chi} \right)^{12}_{21} = 
c(\lambda_j, \lambda_k) e^{-\phi_j+\phi_k} \, . 
\ee

We now put $\phi_j= \lambda_j$ 
in eq. (\ref{eq:Phi}) for $j=0, 1, \ldots, L$.  
For $j,k=0, 1, \ldots, L$, we have  
\be 
R_{jk}^{\pm}(\lambda_j, \lambda_k) = \left( \chi_{jk} \right)^{\pm 1}  
\, R_{jk}(\lambda_j, \lambda_k) \, 
\left( \chi_{jk} \right)^{\mp 1}  \, .  
\ee
Thus, the asymmetric $R$-matrices $R_{12}^{\pm}(\lambda_1, \lambda_2)$  
are derived from the symmetric one through the gauge transformation 
$\chi_{jk}$.  

 For the monodromy matrix, 
in terms of the inhomogeneous parameters, 
$\xi_1, \ldots, \xi_L$, we put $\lambda_j=\xi_j$ for $j=1, \ldots, L$.   
We define ${\chi}_{0 1 2 \cdots L}$ by 
${\chi}_{0 1 2 \cdots L}= \Phi_0 \Phi_1 \cdots \Phi_L$. 
Then, the asymmetric monodromy matrices are transformed into the 
symmetric one as follows.   
\be 
R_{0, 1 2 \cdots L}^{\pm} 
= \left( \chi_{0 1 2 \cdots L} \right)^{\pm 1} \,  
R_{0, 1 2 \cdots L} \left( \chi_{0 1 2 \cdots L} \right)^{\mp 1} \, . 
\ee

We note that the asymmetric $R$-matrices 
${\check R}^{\pm}_{j, j+1}(u)$ are derived from the 
symmetric $R$-matrix through the gauge transformations, 
and they are related to  the Jones polynomial.  \cite{AW}

%%%%%%%%%%%%%%%%%%%%%%%%%%%%%%%%%%%%%%
%
%  S 3.4
%
\subsection{Affine quantum group symmetry}

The affine quantum algebra $U_q(\widehat{sl_2})$ 
is an associative algebra over ${\bf C}$ generated by  
$X_i^{\pm}, K_i^{\pm}$ for $i=0,1$ with the following relations: 
\bea 
K_i K_i^{-1} & = & K^{-1}_i K_i = 1 \, , \quad 
K_i X_i^{\pm} K_i^{-1}  =  q^{\pm 2} X_i^{\pm} \, ,  \quad 
K_i X_j^{\pm} K_i^{-1}  =  q^{\mp 2} X_j^{\pm}  
\quad (i \ne j) \, , \non \\
{[} X_i^{+}, X_j^{-} {]} & = & \delta_{i,j} \, 
{\frac   {K_i - K_i^{-1}}  {q- q^{-1}} } \, ,     \non \\
(X_i^{\pm})^{3} X_j^{\pm} & - & [3]_q \, (X_i^{\pm})^{2} X_j^{\pm} X_i^{\pm}   
+ [3]_q \, X_i^{\pm} X_j^{\pm} (X_i^{\pm})^2 -   
 X_j^{\pm} (X_i^{\pm})^3 = 0 \quad (i \ne j) \, . 
\label{ defrl}
\eea
The algebra $U_q(\widehat{sl_2})$ is also a Hopf algebra over ${\bf C}$ 
with comultiplication
\bea 
\Delta (X_i^{+}) & = & X_i^{+} \otimes 1 + K_i \otimes X_i^{+}  \, , 
 \quad 
\Delta (X_i^{-})  =  X_i^{-} \otimes K_i^{-1} + 1 \otimes X_i^{-} \, ,  
\non \\
\Delta(K_i) & = & K_i \otimes K_i  \, , 
\eea 
and antipode:  
$S(K_i)=K_i^{-1} \, , S(X_i)= - K_i^{-1} X_i^{+} \, , 
S(X_i^{-}) = - X_i^{-} K_i $.  

%%%%%%%%%%%%%%%%%%%%%%%%%%%%%%%%%%%%%%%%%%%%%%%%%%%%%%%%%%%%%
%
% Here we have X_i^{\pm} -> X_i^{\mp} and K_i -> K_i^{-1}
%
%%%%%%%%%%%%%%%%%%%%%%%%%%%%%%%%%%%%%%%%%%%%%%%%%%%%%%%%%%%%%

We now introduce evaluation representations for $U_q(\widehat{sl_2})$ 
\cite{Jimbo-QG}. 
For a given complex number $a$ there is a homomorphism 
of algebras  $\varphi_a$: $U_q(\widehat{sl_2}) \rightarrow U_q({sl_2})$  
such that  
\bea 
\varphi_a(X_0^{\pm}) & = & \exp(\pm a) \, 
%a^{\pm 1} 
X^{\mp} \, , 
\quad \varphi_a(K_0)=K^{-1} \, ,  \non \\ 
\varphi_a(X_1^{\pm}) & = & X^{\pm} \, , \quad \varphi_a(K_1) =K \, .  
\label{eq:eval-rep}
\eea
We denote by $(\pi, V)$ a representation of an algebra ${\cal A}$ 
such that $\pi(x)$ give linear maps on vector space $V$ 
for $x \in {\cal A}$. 
For a given finite-dimensional representation $(\pi_V, V)$ of $U_q(sl_2)$ 
 we have a finite-dimensional representation $(\pi_{V(a)}, V(a))$ 
of $U_q(\widehat{sl_2})$  through homomorphism $\varphi_a$, i.e. 
$\pi_{V(a)}(x)= \pi_V(\varphi_a(x))$ for $x \in U_q(\widehat{sl_2})$. 
We call $(\pi_{V(a)}, V(a))$ or $V(a)$ 
the {\it evaluation representation} of $V$ and  nonzero parameter $a$ 
the {\it evaluation parameter} of $V(a)$.  
If $V$ is $(2s+1)$-dimensional, 
then we also denote it  by $V^{(2s)}(a)$. 
Hereafter we express $2s$ by an integer $\ell$.

Similarly as (\ref{eq:TLdecomp+}), we have the following decomposition: 
\be 
{\check R}_{j,j+1}^{-}(u) = I - b(u) U_{j}^{-} \, , 
\quad {\rm for} \, \,  j=0, 1, \ldots, L-1 .
\label{eq:TLdecomp(-)}
\ee 
\begin{lem}
Generators $U^{-}_j$ commute with $X_0^{\pm}$ and $K_0$ 
of $U_q(\widehat{ sl_2})$ in the tensor-product representation 
$V^{(1)}(a_0) \otimes \cdots \otimes V^{(1)}(a_L)$ with 
$a_0 = a_1 = \cdots = a_L =a$. 
 For $j=0, 1, \ldots, L-1$, we have  
\be 
{\big[} U^{-}_j, \, \, \varphi_a^{\otimes (L+1)} 
\left( \Delta^{(L)}(x) \right) {\big]} = 0 \,  \quad 
\left( x=X_0^{\pm}, \,  K_0. \right)
\label{UD-}
\ee
\label{lem:Uq-inv(-)}
\end{lem} 
\begin{proof} Let us denote by $U^{-}$ the $4 \times 4$ matrix given by  
the decomposition: 
${\check R}(u) = I - b(u) U^{-}$. 
Through an explicit calculation we show that 
$U^{-}$ and $\varphi_a \otimes \varphi_b 
\left( \Delta(X_0^{\pm}) \right)$
commute if $a=b$: 
\be 
{[} 
U^{-}, \,  \varphi_a \otimes \varphi_a 
\left( \Delta(X_0^{\pm}) \right)  {]} = 0 \, .  \label{eq:UD2-} 
\ee
We derive (\ref{UD-}) through (\ref{eq:UD2-}). 
\end{proof}

In the spin-1/2 representation of $U_q(sl_2)$, we thus have 
the following  relations: 
\be 
R_{12}^{-}(u) \, \, \varphi_a^{\otimes 2} \Delta (x) = 
 \varphi_a^{\otimes 2} \left( \tau \circ \Delta (x) \right) 
\, \,  R^{-}_{12}(u) 
\quad {\rm for \, \,} \quad x= X_0^{\pm}, K_0 \, .  
\label{eq:RD-}
\ee
Here we note that $u$ is arbitrary and independent of $X_0^{\pm}, K_0$. 

Similarly as lemmas 
\ref{lem:expU} and \ref{lem:Uq-inv}, 
we have from lemma \ref{lem:Uq-inv(-)} the quantum-group symmetry of 
the monodromy matrix $R^{-}_{0, 1 2 \cdots L}$: 
\begin{prop} 
Let $\sigma_c$ be a cyclic permutation: 
 $\sigma_c= (0 1 \cdots L)$. In the evaluation representation 
(\ref{eq:eval-rep}) we have, for $x=X_0^{\pm}, K_0$, the following:     
\be 
R^{-}_{0, 1 2 \cdots L}(\lambda; \xi_1, \ldots, \xi_L) \, 
 \varphi_a^{\otimes (L+1)} \left( \Delta^{(L)}(x) \right) = 
 \varphi_a^{\otimes (L+1)} \left( \sigma_c \circ \Delta^{(L)}(x)  \right) \, 
 R^{-}_{0, 1 2 \cdots L}(\lambda; \xi_1, \ldots, \xi_L)  \, .  
\label{eq:q-inv(-)}
\ee
\label{prop:q-inv(-)}
Here parameters $\lambda$, $\xi_1, \ldots, \xi_L$ are arbitrary and independent of $x=X_0^{\pm}, K_0$. 
\end{prop} 

Let us now make a summary of the symmetry relations of $R^{+}_{12}$. 
Here we recall that 
$R^{+}_{12}(\lambda_1, \lambda_2) \in End(V(\lambda_1) \otimes V(\lambda_2))$. 
For simplicity, we put $a=0$ in (\ref{eq:q-inv(-)}). 
Combining (\ref{eq:RD}) and (\ref{eq:RD-})
We have the following relations:  
\bea 
R^{+}_{12}(\lambda_1, \lambda_2)  
\, \, \varphi_{0}^{\otimes 2} \left( \Delta(X_1^{\pm}) \right) 
& = &  \varphi_{0}^{\otimes 2} 
\left( \tau \circ \Delta(X_1^{\pm}) \right) \, \, 
R^{+}_{12}(\lambda_1, \lambda_2)   \, , 
\non \\
R^{+}_{12}(\lambda_1, \lambda_2) \, (\chi_{12})^2  
\varphi_{0}^{\otimes 2} \left( \Delta(X_0^{\pm}) \right) (\chi_{12})^{-2} 
& = & (\chi_{12})^2   \varphi_{0}^{\otimes 2} \left( 
\tau \circ \Delta(X_0^{\pm}) \right) (\chi_{12})^{-2} \, 
R^{+}_{12}(\lambda_1, \lambda_2)  \, ,   \non \\
R^{+}_{12}(\lambda_1, \lambda_2) \,  \varphi_{0}^{\otimes 2} 
\left( \Delta(K_i^{\pm}) \right) 
& = &  \varphi_{0}^{\otimes 2} \left(\tau \circ \Delta(K_i^{\pm}) \right) \, 
R^{+}_{12}(\lambda_1, \lambda_2) \, {\rm for} \, i=0, 1 \, . \non \\  
\label{eq:affineCR}
\eea
Let us now consider $\varphi_{a_1} \otimes \varphi_{a_2}$ 
with $a_j=2 \lambda_j$ for $j=1, 2$.  We have 
\be 
\varphi_{2\lambda_1} \otimes \varphi_{2 \lambda_2} 
\left( \Delta(X_0^{\pm})  \right) 
= (\chi_{12})^2  \, \, 
\varphi_{0} \otimes \varphi_{0} 
\left( \Delta(X_0^{\pm}) \right)  \, \, (\chi_{12})^{-2}
\ee
Thus, relations (\ref{eq:affineCR}) are now expressed as follows.  
\be
R^{+}_{12}(\lambda_1, \lambda_2) \, \varphi_{2 \lambda_1} \otimes 
\varphi_{2 \lambda_2}  \left( \Delta(x) \right) 
= \varphi_{2 \lambda_1} \otimes \varphi_{2 \lambda_2}  
\left( \tau \circ \Delta(x) \right) \, 
R^{+}_{12}(\lambda_1, \lambda_2) \,, \, \,  {\rm for} \,\,  
x=X_0^{\pm}, X_1^{\pm}, K_0, K_1 \, .  
\label{eq:affine-sym}
\ee
In (\ref{eq:affine-sym}) all the parameters are now associated with 
the evaluation parameters of the tensor product 
$V(2 \lambda_1) \otimes V(2 \lambda_2)$.  
Therefore, we conclude that 
the asymmetric $R$-matrix $R^{+}_{12}(\lambda_1, \lambda_2)$ 
satisfies the affine quantum-group symmetry. 

The fundamental commutation relations (\ref{eq:q-inv(+)}) 
and (\ref{eq:q-inv(-)}) are summarized as follows. 
\begin{prop} Let $\sigma_c$ be a cyclic permutation: 
$\sigma_c= (0 1 \cdots L)$. The asymmetric $R$-matrix $R^{+}$ 
satisfies the commutation relations 
for the affine-quantum group:  
\be 
R^{+}_{0, 1 2 \cdots L}(\lambda_0)  
\left( \Delta^{(n)}(x) \right)_{0 1 \cdots L}  = 
\left( \sigma_c \circ \Delta^{(L)}(x) \right)_{0 1 \cdots L}  
 R^{+}_{0, 1 2 \cdots L}(\lambda_0)
\quad {\rm for \, \, all} \, \,  x \in U_q({\widehat{ sl_2}}) \, . 
\label{eq:affineR+}
\ee
Here the symbol $( x )_{0 1 \cdots n}$ denotes the matrix representation 
of $x$ in the tensor product of evaluation representations,  
$V({2 \lambda_0}) \otimes V({2 \xi_1}) \otimes \cdots \otimes 
V({2 \xi_L})$.  
\end{prop}

\subsection{Symmetry relations of $U_q(\widehat{sl_2})$ for all permutations}

Let us generalize relations (\ref{eq:affineR+}). 
Making an extensive use of relations 
(\ref{eq:affine-sym}), 
we can show commutation relations 
for $R_p^{\sigma}$ for all permutations $\sigma$. 
In fact,  we can prove (\ref{eq:affineR+}) also by the method 
for showing proposition \ref{prop:A}. 
In Appendix A we shall show in proposition \ref{prop:A} 
how we generalize the symmetry relation of $R_{12}$ 
such as (\ref{eq:affine-sym})
 into those of $R_p^{\sigma}$ for any permutation $\sigma$.

We now formulate the symmetry relations 
in terms of the symmetric $R$-matrices.   
Let us denote by ${\bar \chi}$  the inverse of the  gauge transformation 
$\chi$. We express by      
$\left( \Delta^{(n)}(x) \right)_{0 1 \cdots n}^{\bar \chi}$ the following: 
\be  
\left( \Delta^{(n)}(x) \right)_{0 1 \cdots n}^{\bar \chi} = 
\left( \chi_{01 \cdots n} \right)^{-1} \, 
\varphi_{2 \lambda_0} \otimes \varphi_{2 \xi_1} \otimes \cdots 
\otimes \varphi_{2 \xi_n}  \left( \Delta^{(n)}(x) \right) \, \chi_{01 \cdots n}
\ee 

\begin{prop}
Let $p_q$ be an increasing 
sequence of $n+1$ integers: $p_q=(0, 1, 2, \ldots ,n)$, 
and $\sigma$ a permutation on $n+1$ integers, $0, 1, \ldots, n$. 
With the symmetric $R$-matrices, we have  
\be 
 R^{\sigma}_{p_q}(\lambda_0) \,   
\left( \Delta^{(n)}(x) \right)_{0 1 \cdots n}^{\bar \chi}
= \left( \sigma \circ \Delta^{(n)}(x) \right)^{\bar \chi}_{0 1 \cdots n} \, 
 R^{\sigma}_{p_q}(\lambda_0)
\quad {\rm for \, \, all} \, \,  x \in U_q({\widehat{ sl_2}}) \, . 
\label{eq:affineR-sigma}
\ee 
%
%In particular, for $\sigma_c=(0 1 \cdots n) $, we have  
%\be  R_{0, 1 2 \cdots n}(\lambda_0) \,   
%\left( \Delta^{(n)}(x) \right)_{0 1 \cdots n}^{\bar \chi} =   
%\left( \sigma_c \circ \Delta^{(n)}(x) \right)^{\bar \chi}_{0 1 \cdots n} \, 
% R_{0, 1 2 \cdots n}(\lambda_0)
%\quad {\rm for \, \, all} \, \,  x \in U_q({\widehat{ sl_2}}) \, . 
%\label{eq:affineR-cyclic} \ee
%
\label{prop:affine-q-inv}
\end{prop}

%%%%%%%%%%%%%%%%%%%%%%%%%%%%%%%%%%%%%%%%%%%%%%%%%%%%%%%%%%
%
%              SECTION 4 
%
%%%%%%%%%%%%%%%%%%%%%%%%%%%%%%%%%%%%%%%%%%%%%%%%%%%%%%%%%%
 \setcounter{equation}{0} 
 \renewcommand{\theequation}{4.\arabic{equation}}
\section{Projection operators and the fusion procedure }

\subsection{Projection operators}

Let us recall that ${\check R}_{12}^{+}(u)$ has been defined by 
$$ 
{\check R}_{12}^{+}(u) = \Pi_{12} R_{12}^{+}(u) \,. 
$$
We define operator $P_{12}^{+}$  by 
\be 
P_{1 2}^{+} = {\check R}_{1 2}^{+}(\eta)  
\ee
Explicitly we have 
\be
P_{12}^{+}  
= \left(
\begin{array}{cccc} 
1 &   0 & 0 & 0 \\
0 &   {\frac q {[2]}} & {\frac 1 {[2]}} & 0 \\
0 &   {\frac 1 {[2]}} & {\frac {q^{-1}} {[2]}} & 0 \\
0 &   0 & 0 & 1
\end{array} 
\right)_{[12]} \, . \label{P+}  
\ee
By making use of the matrix representation (\ref{P+})
it is easy to show that operator $P_{12}^{+}$ is idempotent: 
\be 
\left( P_{12}^{+} \right)^2 = P_{12}^{+} 
\ee
Thus, we may consider that the operator $P_{12}^{+}$ is a 
projection operator. In fact, $P_{12}^{+}$  
is nothing but the $q$-analogue of the 
Young operator which projects out the 
spin-1 representation of $U_q(sl_2)$ 
from of the tensor product of two spin-1/2 representations.  
It should be noticed that in the case of symmetric $R$-matrix, 
$R(\eta)^2$ is not equal to $R(\eta)$.

We now introduce projection operators 
for the spin-$s$ irreducible representations of $U_{q}(sl_2)$. 
Hereafter we set $\ell=2s$. 
We define projection operator $P_{12 \cdots \ell}^{(\ell)}$ 
 acting on the tensor product $V_{12 \cdots \ell}$ 
of spin-1/2 representations $V$, 
i.e. $V_{12 \cdots \ell}=V^{\otimes \ell}$,  
by the following recursive relations:
\be 
P^{(\ell)}_{12 \cdots \ell} =  
P^{(\ell-1)}_{12 \cdots \ell-1} {\check R}_{\ell-1, \ell}^{+}((\ell-1)\eta) 
P^{(\ell-1)}_{12 \cdots \ell-1} 
\ee
Here $P_{12}^{(1)}=P_{12}^{+}$. 

Making use of the Yang-Baxter equations (\ref{eq:YB-Rcheck}) and 
through induction on $\ell$  
we can show  
\be
\left( P^{(\ell)}_{12 \cdots \ell} \right)^2 = 
P^{(\ell)}_{12 \cdots \ell} 
\label{eq:nilpotency}
\ee
Thus, operator $P^{(\ell)}_{12 \cdots \ell}$ gives a projection operator. 
Similarly, we define projection operators 
acting on $V_{j \, j+1 \cdots \, j+\ell-1}$ recursively by 
\be 
P^{(\ell)}_{j \, j+1 \cdots j+\ell-1} =  
P^{(\ell-1)}_{j \, j+1 \cdots j+\ell-2} 
{\check R}_{j+\ell-2, \, j+\ell-1}^{+}((\ell-1)\eta) 
P^{(\ell-1)}_{j \, j+1 \cdots j+\ell-2}  
\ee
Hereafter we shall abbreviate $P^{(\ell)}_{j \, j+1 \cdots j+\ell-1}$  
by $P^{(\ell)}_{j}$ .

From idempotency (\ref{eq:nilpotency}) and by 
the Yang-Baxter equations we can show the following:
\begin{lem}  
Suppose that  
inhomogeneous parameters $\xi_j, \xi_{j+1}, \ldots, \xi_{j+\ell-1}$ are 
given by $\xi_{j+i-1}= z -(i-1) \eta$ for $i=1, 2, \ldots, \ell$ 
with a constant $z$. Then, 
the monodromy matrix $R^{+}_{0, 1 2 \cdots L}$ satisfies the following property:\be 
P_{j}^{(\ell)} \, R^{+}_{0, 1 2 \cdots L} 
=
P_{j}^{(\ell)} \, 
R^{+}_{0, 1 2 \cdots L} \, P_{j}^{(\ell)} \, . 
\label{eq:PR=PRP}
\ee 
\label{lem:PR=PRP}
\end{lem} 

Making use of (\ref{eq:product-Rcheck}) we can express 
projection operators in terms of $R$-matrices. 
\be 
P^{(\ell)}_1= \left( \prod_{j=1}^{\ell} \Pi^{(j \, \, \ell-j+1)} \right) \, \, 
R^{+}_{\ell-1, \, \ell}
\cdots 
R^{+}_{2,\,  3 \cdots \ell}
R^{+}_{1, \, 2 \cdots \ell}     
\label{proj}
\ee

%%%%%%%%%%%%%%%%%%%%%%%%%%%%%%%%%%%%%
%
%    S 4.2
%
\subsection{Fusion of monodromy matrices}
\subsubsection{The case of tensor product of spin-$s$ representations}

We first consider the case of tensor product of 
spin-$s$ representations. 
We set $L=N_{s} \ell$. 
We introduce a set of parameters, 
$\xi_1^{(\ell)}$,  $\xi_2^{(\ell)}$, \ldots,  $\xi_L^{(\ell)}$, as follows: 
\be 
\xi_{(k-1)\ell+j}^{(\ell)} = {\zeta}_k - (j-1) \eta + 
(\ell-1) \eta/2 \qquad  \mbox{for} \quad 
j=1, \ldots \ell, \, \, \mbox{and} \, \, k= 1, \ldots, N_{s}.  
\ee
Let us set inhomogeneous parameters $\xi_1, \xi_2, \ldots, \xi_L$, 
by $\xi_j=\xi_j^{(\ell)}$ for $j=1, 2, \ldots, L$. 
We define the monodromy matrix 
$T^{(\ell+)}_0(u; \zeta_1, \ldots, \zeta_{N_{s}})$
acting on the tensor product of spin-$s$ representations, 
$V^{(2s)}(\zeta_1) \otimes V^{(2s)}(\zeta_{N_{s}})$  by 
\be 
T^{(\ell+)}_0(\lambda_0; \zeta_1, \ldots, \zeta_{N_{s}}) =  
\prod_{k=0}^{N_{s}-1} P^{(\ell)}_{\ell k+1} \, \cdot \,  
R^{+}_{0, 1 2 \cdots L}(\lambda_0; \xi_1^{(\ell)}, 
\ldots, \xi_L^{(\ell)}) \, 
\cdot \,  
\prod_{k=0}^{N_{s}-1} P^{(\ell)}_{\ell k+1} 
\ee
Making use of properties of projection operators (\ref{eq:nilpotency}) 
and (\ref{eq:PR=PRP}), we can show the Yang-Baxter equation 
\cite{fusionXXX,V-WDA} 
\bea 
& & R^{+}_{ab}(\lambda_a-\lambda_b)
T^{(\ell+)}_{a}(\lambda_a; \zeta_1, \ldots, \zeta_{N_{s}}) 
T^{(\ell+)}_{b}(\lambda_b; \zeta_1, \ldots, \zeta_{N_{s}}) \non \\ 
& = & 
T^{(\ell+)}_{b}(\lambda_b; \zeta_1, \ldots, \zeta_{N_{s}}) 
T^{(\ell+)}_{a}(\lambda_a; \zeta_1, \ldots, \zeta_{N_{s}}) 
R^{+}_{ab}(\lambda_a-\lambda_b)
\eea

Through the inverse of the gauge transformation, we derive 
the symmetric spin-$s$ monodromy matrix as follows:   
\be
T^{(\ell)}_0(\lambda_0; \zeta_1, \ldots, \zeta_{N_{s}}) 
 = \prod_{k=0}^{N_{s}-1} \left( P^{(\ell)}_{\ell k+1} \right)^{\bar \chi}
 \, \cdot \,   
R_{0, 1 2 \cdots L}(\lambda_0; \xi_1^{(\ell)}, \ldots, \xi_L^{(\ell)}) 
\, \cdot \,  
\prod_{k=0}^{N_{s}-1} \left( P^{(\ell)}_{\ell k+1} \right)^{\bar \chi}
\ee
where we have defined the transformed projectors by 
\be 
\left( P^{(\ell)}_{\ell k+1} \right)^{\bar \chi} = 
\left( \chi_{0 1 \cdots L} \right)^{-1} \, P^{(\ell)}_{\ell k+1} \, 
 \left( \chi_{0 1 \cdots L} \right)
\ee

\subsubsection{The case of mixed spins}

%We next discuss the case of mixed spins. 
Let us consider the tensor product of representations 
with different spins, $s_1, s_2, \ldots, s_r$.  Here we introduce 
$\ell_j$ by $\ell_j=2 s_j$ for $j=1, 2, \ldots, r$, and 
we assume that $\ell_1+ \ell_2 + \cdots + \ell_r = L$.  @
Let us introduce a set of parameters, 
$\xi_1^{({\bm \ell})}$,  $\xi_2^{({\bm \ell})}$, \ldots,  
$\xi_L^{({\bm \ell})}$, as follows: 
\be 
\xi_{\ell_1 + \cdots + \ell_{k-1} +j}^{({\bm \ell})} = 
\zeta_k - (j-1) \eta + (\ell_k-1) \eta/2 \qquad  \mbox{for} \quad 
j=1, \ldots \ell_k, \, \, \mbox{and} \, \, k= 1, \ldots, r.  
\ee

We define the asymmetric monodromy matrix for the mixed spin case 
$T^{({\bm \ell}+)}_0(u; \zeta_1, \ldots, \zeta_{r})$ 
acting on the tensor product representation  
$V_{\ell_1}({\zeta_1}) \otimes \cdots \otimes V_{\ell_r}({\zeta_{r}})$  by 
\be 
T^{({\bm \ell}+)}_0(\lambda_0; \zeta_1, \ldots, \zeta_{r}) =  
\prod_{k=1}^{r} P^{(\ell_k)}_{\ell_1 + \cdots + \ell_{k-1}+1} \, \cdot \,  
R_{0, 1 2 \cdots L}(\lambda_0; \xi_1^{(\bm \ell)}, \ldots, \xi_L^{(\bm \ell)}) 
\, \cdot \,  
\prod_{k=1}^{r} 
P^{(\ell_k)}_{\ell_1 + \cdots + \ell_{k-1}+1} \, . 
\ee
It is easy to show that they satisfy the Yang-Baxter equations. 
\bea 
& & R^{+}_{ab}(\lambda_a-\lambda_b)
T^{({\bm \ell}+)}_{a}(\lambda_a; \zeta_1, \ldots, \zeta_{N_{s}}) 
T^{({\bm \ell}+)}_{b}(\lambda_b; \zeta_1, \ldots, \zeta_{N_{s}}) \non \\ 
& = & 
T^{({\bm \ell}+)}_{b}(\lambda_b; \zeta_1, \ldots, \zeta_{N_{s}}) 
T^{({\bm \ell}+)}_{a}(\lambda_a; \zeta_1, \ldots, \zeta_{N_{s}}) 
R^{+}_{ab}(\lambda_a-\lambda_b)
\label{eq:TBE-mixed}
\eea

We also define the symmetric monodromy matrix for the mixed spin case 
$T^{({\bm \ell})}_0(u; \zeta_1, \ldots, \zeta_{r})$ as follows. 
\be 
T^{({\bm \ell})}_0(\lambda_0; \zeta_1, \ldots, \zeta_{r}) =  
\prod_{k=1}^{r} \left( P^{(\ell_k)}_{\ell_1 + \cdots + \ell_{k-1}+1} 
\right)^{\bar \chi} \, \cdot \,  
R_{0, 1 2 \cdots L}(\lambda_0; \xi_1^{(\bm \ell)}, \ldots, 
\xi_L^{(\bm \ell)}) 
\, \cdot \,  
\prod_{k=1}^{r} 
\left( P^{(\ell_k)}_{\ell_1 + \cdots + \ell_{k-1}+1} \right)^{\bar \chi}
\ee

%%%%%%%%%%%%%%%%%%%%%%%%%%%%%%%%%%%%%%%%%%%%%
%
%            4.3
%
\subsection{Higher-spin $L$-operators}

We now define the basis vectors of the $(\ell+1)$-dimensional 
irreducible representation of $U_q(sl_2)$,  
$|| \ell, n  \rangle$ for $n=0, 1, \ldots, \ell$ as follows. 
We define $||\ell, 0 \rangle$ by 
\be 
||\ell , 0 \rangle = |1 \rangle_1 \otimes |1 \rangle_2 \otimes 
\cdots |1 \rangle_\ell 
\ee   
Here $|\alpha \rangle_j$ for $\alpha=1, 2$ 
denote the basis vectors of the spin-1/2 representation defined 
on the $j$th position in the tensor product.  We define 
$|| \ell, n \rangle$ for $n \ge 1$ by     
\be 
|| \ell, n \rangle =   
\left( \Delta^{(\ell-1)} (X^{-}) \right)^n ||\ell, 0 \rangle \,  
{\frac 1 {[n]_q!}} \, . 
\ee 
Then we have 
\be 
|| \ell, n \rangle = 
\sum_{1 \le i_1 < \cdots < i_n \le \ell} 
\sigma_{i_1}^{-} \cdots \sigma_{i_n}^{-} | 0 \rangle \, 
q^{i_1+ i_2 + \cdots + i_n  - n \ell + n(n-1)/2}
\label{eq:|ell,n>}
\ee
It is easy to show the following: 
\be 
P^{(\ell)}_{1 2 \cdots \ell} || \ell, n \rangle = || \ell, n \rangle 
\ee
We define the conjugate vectors by the following conditions:  
\be
\langle \ell, n || P^{(\ell)}_{1 2 \cdots \ell} = \langle \ell, n || 
\ee
with the normalization condition: 
$\langle \ell, n || \, || \ell, n \rangle = 1$.
Let us define the $q$-factorial, $[n]_q!$, by 
\be 
[n]_q ! = [n]_q [n-1]_q \cdots [1]_q \, .  
\ee
For integers $m$ and $n$ satisfying $m \ge n$ 
we define the $q$-binomial coefficients as follows
\be 
\left[ 
\begin{array}{c} 
m \\ 
n 
\end{array}  
 \right]_q 
= {\frac {[m]_q !} {[m-n]_q ! [n]_q!}}  
\ee
Then we have the following expression  of the conjugate vectors 
\be 
\langle \ell, n || =  
\left[ 
\begin{array}{c} 
\ell \\ 
n 
\end{array}  
 \right]_q^{-1} \, q^{n(\ell-n)} \, 
\sum_{1 \le i_1 < \cdots < i_n \le \ell} 
\langle 0 | \sigma_{i_1}^{+} \cdots \sigma_{i_n}^{+} \, 
q^{i_1 + \cdots + i_n - n \ell + n(n-1)/2}   
\label{eq:<ell,n|}
\ee
The projection operators are given explicitly as follows.  
\be 
P^{(\ell)}_{1 2 \cdots \ell} 
= \sum_{n=0}^{\ell}  || \ell, n \rangle \, \langle \ell, n ||
\ee

We define the $L$-operator of the spin-$\ell/2$ XXZ model by    
\be 
L^{(\ell+)}(\lambda_0) = P^{(\ell)}_{1} \,  
R^{+}_{0, 1 2 \cdots \ell} \,  P^{(\ell)}_{1}  \, . 
\ee
and then by the inverse gauge transformation 
we have 
\be 
L^{(\ell)}(\lambda_0) = 
\left( P^{(\ell)}_{1} \right)^{\bar \chi} \, 
R_{0, 1 2 \cdots \ell} \, 
\left( P^{(\ell)}_{1} \right)^{\bar \chi} \, . 
\ee

Let us define $|\ell, n \rangle$ and their conjugates $\langle \ell, n |$ by  
\bea 
| \ell, n \rangle & = & N(\ell,n) \, || \ell, n \rangle \non \\ 
\langle \ell, n | & = &  \langle \ell, n || \,  {\frac 1 {N(\ell,n)}} \, .  
\eea
The matrix elements of the $L$-operator are  given by 
\be 
\langle \ell, a | \, L^{(\ell)}(\lambda) \, | \ell, b \rangle  
= \left( 
\begin{array}{cc}
\langle \ell, a | \,  L^{(\ell)}_{11}(\lambda) \, | \ell, b \rangle 
&
\langle \ell, a | \,  L^{(\ell)}_{12}(\lambda) \, | \ell, b \rangle 
\\
\langle \ell, a | \,  L^{(\ell)}_{21}(\lambda) \, | \ell, b \rangle 
&
\langle \ell, a | \,  L^{(\ell)}_{22}(\lambda) \, | \ell, b \rangle 
\end{array}
\right)_{[0]} 
\ee
for $a, b = 0, 1, \ldots, \ell$. 
Choosing the normalization factors $N(\ell,n)$, 
we can derive the following symmetric expression of the 
 $L$-operator:    
\be 
L^{(\ell)}(\lambda)  =  
{\frac 1 { 2 \sinh(u+\ell\eta/2)}}  
\left( 
\begin{array}{cc}
z K^{1/2} - z^{-1} K^{-1/2} &  2 \sinh \eta X^{-} \\
 2 \sinh \eta X^{+} & z K^{-1/2} - z^{-1} K^{1/2} 
\end{array}
\right)_{[0]} 
\ee
Here $X^{\pm}$ and $K$ are in the 
$(\ell+1)$-dimensional representation of $U_q(sl_2)$, and 
$u=\lambda - \xi_1 + \ell\eta/2$ and  $z=\exp u$. Explicitly 
they are given by 
\bea  
\langle \ell, a | \, X^{+} \, | \ell, b \rangle & = & 
\delta_{a, b-1} \, [\ell-a]_q \, , \non \\    
\langle \ell, a | \, X^{-} \, | \ell, b \rangle & = & 
\delta_{a, b+1} \, [a]_q \, , \non \\    
\langle \ell, a | \, K \, | \ell, b \rangle & = & 
\delta_{a, b} \, q^{\ell - 2 a} \quad {\rm for} \, \, 
a, b= 0, 1, \ldots, \ell , .     
\eea

%%%%%%%%%%%%%%%%%%%%%%%%%%%%%%%%%%%%%%%%%%%%%%%%%%%%
%
%        S 4.4 
%
\subsection{Algebraic Bethe-ansatz method for higher-spin cases}

We now discuss the eigenvalues of the transfer matrix 
of an integrable  higher-spin XXZ spin chain  
constructed by the fusion method. 
We consider the case of mixed spins, where we define 
the transfer matrix on 
the tensor product of spin-$s_j$ representations for 
$j=1, 2, \ldots, r$.

We define $A$, $B$, $C$, and $D$ operators of the algebraic Bethe ansatz  
for higher-spin cases  
by the following matrix elements of the monodromy matrix: 
\be 
\left(
\begin{array}{cc} 
A^{({\bm \ell}+)}(\lambda_0; \zeta_1, \ldots, \zeta_r) 
& B^{({\bm \ell}+)}(\lambda_0; \zeta_1, \ldots, \zeta_r) \\ 
C^{({\bm \ell}+)}(\lambda_0; \zeta_1, \ldots, \zeta_r) 
& D^{({\bm \ell}+)}(\lambda_0; \zeta_1, \ldots, \zeta_r)  
\end{array} 
\right)
= T^{({\bm \ell}+)}_0(\lambda_0; \zeta_1, \ldots, \zeta_r) \, . 
\ee
In terms of projection operators we have 
\bea 
B^{({\bm \ell}+)}_{1 \cdots r}(\lambda_0; \zeta_1, \ldots, \zeta_{r}) 
 = \prod_{k=1}^{r} P^{(\ell_k)}_{\ell(k-1)+1}   
\, \cdot \ 
B^{+}_{1 2 \cdots L}(\lambda_0; \xi_1^{({\bm \ell})}, \ldots, 
\xi_L^{({\bm \ell})})   
\prod_{k=1}^{r}  P^{(\ell_k)}_{\ell(k-1)+1}  \, . 
\eea
We define $A^{({\bm \ell})}, B^{({\bm \ell})}$,  
$C^{({\bm \ell})}$ and $D^{({\bm \ell})}$ similarly for the monodromy 
matrix  $T^{({\bm \ell})}_0(\lambda_0; \zeta_1, \ldots, \zeta_r)$.

The operators $A^{({\bm \ell}+)}$s of the asymmetric monodromy matrix 
$R^{+}_{0, 1 2 \cdots n}$  are related to the 
 symmetric ones $A^{({\bm \ell})}$ s as follows. 
\bea 
& & R^{+}_{0, 1 2 \cdots n}(\lambda_0, \{ \zeta_i \} )  
= \left(
\begin{array}{cc} 
A^{({\bm \ell}+)}(\lambda_0; \{ \zeta_i \}) 
& B^{({\bm \ell}+)}(\lambda_0; \{ \zeta_i \}) \\ 
C^{({\bm \ell}+)}(\lambda_0; \{ \zeta_i \}) 
& D^{({\bm \ell}+)}(\lambda_0; \{ \zeta_i \})  
\end{array} 
\right)
\non \\ 
& =& 
\left(
\begin{array}{cc} 
\chi_{12\cdots L} A^{({\bm \ell})}(\lambda_0; \{ \zeta_i \}) 
\chi_{12\cdots L}^{-1}
& e^{-\lambda_0} \,  \chi_{12\cdots L} 
B^{({\bm \ell})}(\lambda_0; \{ \zeta_i \}) \chi_{12\cdots L}^{-1} \\ 
e^{\lambda_0} \, \chi_{12\cdots L} 
C^{({\bm \ell})}(\lambda_0; \{ \zeta_i \}) \chi_{12\cdots L}^{-1}
& \chi_{12\cdots L}D^{({\bm \ell})}(\lambda_0; \{ \zeta_i \}) 
\chi_{12\cdots L}^{-1} 
\end{array} 
\right)
\eea

It follows from the Yang-Baxter equations (\ref{eq:TBE-mixed}) 
that the $A, B, C, D$ operators in the higher-spin case 
also satisfy the standard commutation relations.  
\be 
A^{({\bm \ell}+)}(\lambda_1) B^{({\bm \ell}+)}(\lambda_2)    
= {\frac 1  {b(\lambda_2 - \lambda_1)} }   
B^{({\bm \ell}+)}(\lambda_2) A^{({\bm \ell}+)}(\lambda_1)  
- {\frac {c^{-}(\lambda_2 - \lambda_1)} {b(\lambda_2 - \lambda_1)}} 
B^{({\bm \ell}+)}(\lambda_1) A^{({\bm \ell} +)}(\lambda_2)
\ee
Through the inverse gauge transformation ${\bar \chi}$  we have  
\be 
A^{({\bm \ell})}(\lambda_1) B^{({\bm \ell})}(\lambda_2)    
= {\frac 1  {b(\lambda_2 - \lambda_1)} }   
B^{({\bm \ell})}(\lambda_2) A^{({\bm \ell})}(\lambda_1)  
- {\frac {c(\lambda_2 - \lambda_1)} {b(\lambda_2 - \lambda_1)}} 
B^{({\bm \ell})}(\lambda_1) A^{({\bm \ell})}(\lambda_2)
\ee
Therefore, we derive Bethe ansatz eigenvectors of the 
higher-spin transfer matrix by the same method as the case of spin-1/2.    

Let us  denote by $| 0 \rangle$ the vacuum state where all spins are up.  
Noting 
\be 
\prod_{k=1}^{r} P^{(\ell_k)}_{\ell(k-1)+1} | 0 \rangle = | 0 \rangle \, ,  
\ee
it is easy to show the following relations:  
\bea 
A^{({\bm \ell})}(\lambda) |0 \rangle & = & 
a^{({\bm \ell})}(\lambda; \{ \zeta_k \}) | 0 \rangle \, , \non \\  
D^{({\bm \ell})}(\lambda) |0 \rangle & = & 
d^{({\bm \ell})}(\lambda; \{ \zeta_k \}) | 0 \rangle  \,  , 
\eea
where $a^{({\bm \ell})}(\lambda; \{ \zeta_k \})$ and 
$d^{({\bm \ell})}(\lambda; \{ \zeta_k \})$ are given by 
\bea 
a^{({\bm \ell})}(\lambda; \{ \zeta_k \}) & = & 
a^{({\bm \ell})}(\lambda; \{ \xi_j^{({\bm \ell})} \})=1 \, , \non \\  
d^{({\bm \ell})}(\lambda; \{ \zeta_k \}) & = & 
d^{({\bm \ell})}(\lambda; \{ \xi_j^{({\bm \ell})} \}) 
= \prod_{j=1}^{L} b(\lambda - \xi_j^{({\bm \ell})})  
= \prod_{k=1}^{r} {\frac {\sinh(\lambda- \zeta_k - (\ell_k-1) \eta/2)}  
{\sinh(\lambda- \zeta_k + (\ell_k+1) \eta/2)}}  \non \\ 
\eea
Thus, the vector 
$B^{{(\bm \ell)}} (\lambda_1) \cdots B^{(\bm \ell)} (\lambda_n) | 0 \rangle$ 
becomes an eigenvector of the transfer matrix 
$A^{({\bm \ell})}(\lambda) + D^{({\bm \ell})}(\lambda)$
with the following eigenvalue 
\be 
\tau^{({\bm \ell})}(\mu)= 
%a^{({\bm \ell})}(\mu) 
\prod_{j=1}^{n} 
{\frac {\sinh(\lambda_j - \mu + \eta)} {\sinh(\lambda_j - \mu)}} 
+ 
\prod_{k=1}^{r} 
{\frac {\sinh(\mu -\zeta_k - (\ell_k-1) \eta/2)}
       {\sinh(\mu -\zeta_k + (\ell_k+1) \eta/2)}} \, \cdot \, 
\prod_{j=1}^{n} 
{\frac {\sinh(\mu - \lambda_j + \eta)} {\sinh(\mu - \lambda_j)}} 
\ee
if rapidities $\tilde{\lambda}_j=\lambda_j + \eta/2$ satisfy 
the Bethe ansatz equations 
\be 
\prod_{k=1}^{r} {\frac 
{\sinh(\tilde{\lambda}_{\alpha} - \zeta_k + \ell_k \eta/2)} 
{\sinh(\tilde{\lambda}_{\alpha} - \zeta_k - \ell_k \eta/2)} }  
= \prod_{\beta=1; \beta \ne \alpha}^{n}  
{\frac {\sinh(\tilde{\lambda}_{\alpha} - \tilde{\lambda}_{\beta} + \eta)} 
       {\sinh(\tilde{\lambda}_{\alpha} - \tilde{\lambda}_{\beta} - \eta)}} 
\ee

%%%%%%%%%%%%%%%%%%%%%%%%%%%%%%%%%%%%%%%%%%%%%%%%%%%%%%%%%%%%%%%%%%%%
%
%                     SECTION  5
%
%%%%%%%%%%%%%%%%%%%%%%%%%%%%%%%%%%%%%%%%%%%%%%%%%%%%%%%%%%%%%%%%%%%%
 \setcounter{equation}{0} 
 \renewcommand{\theequation}{5.\arabic{equation}}
\section{Pseudo-diagonalization of the $B$ and $C$ operators} 

\subsection{Diagonalizing the $A$ and $D$ operators}
\subsubsection{The $F$-basis}
In order to formulate the derivation of 
the pseudo-diagonalized forms of $B$ and $C$ operators for the XXZ case, 
 we briefly formulate some symbols and review some useful formulas 
shown in Ref. \cite{MS2000} in \SS 5.1. 
First, we introduce the $F$-basis . 

\begin{df}(Partial $F$ and total $F$)
We define partial $F$ by 
\bea 
F_{1, \, 2 \cdots n} & = & e^{11}_1 + e_1^{22} R_{1, \, 2 \cdots n}  \non \\ 
F_{1 2 \cdots n-1, \, n} & = & e^{22}_n + e_n^{11} R_{1 2 \cdots n-1, \, n} 
\label{eq:def-F}
\eea
We define total $F$ recursively with respect to $n$ by 
\be 
F_{1 2 \cdots n} = F_{1 2 \cdots n-1} F_{ 1 2 \cdots n-1, \, n}  
\ee
\label{df:def-F}
\end{df}

\begin{lem} (Cocycle conditions) 
\bea 
F_{1, \, 2} F_{1 2, \, 3} & = & F_{2 3} F_{1, \, 2 3} \non \\  
F_{1, \, 2 \cdots n-1} F_{1 2 \cdots n-1 \, n} & = & F_{2 \cdots n-1, \, n} 
F_{1, \, 2 \cdots n} 
\eea
\end{lem}
\begin{proof}
Expressing the $F$-basis in terms of $R$-matrices through (\ref{eq:def-F}), 
we show that the cocycle conditions of the $F$-basis are reduced 
to those of the $R$-matrices, which are shown in Appendix B.  
\end{proof} 

From the cocycle conditions we have the following: 
\begin{lem}
\be 
F_{12 \cdots n} = F_{2 \cdots n} F_{1, 2 \cdots n}
\ee
\end{lem}

%\begin{prop} Operators $A_{1 \cdots n}$ and $C_{1 \cdots n}$ are 
%upper triangular matrices while 
%$B_{1 \cdots n}$ and $D_{1 \cdots n}$ are 
%lower triangular matrices. 
%The eigenvalues of operators $A$ and $D$ are given by  
%\bea & & diag(D_{1 \cdots n}(\lambda_0)) = \bigotimes_{j=1}^{n} 
%\left( \begin{array}{cc} 
%b_{0j} & 0 \\
%0 &  1  
%\end{array} \right)_{[j]} \, ,  \\ 
%& & diag(A_{1 \cdots n}(\lambda_0)) = \bigotimes_{j=1}^{n} 
%\left( \begin{array}{cc} 
%1 & 0 \\
%0 &  b_{0j} 
%\end{array} \right)_{[j]} \, , \eea
%where $b_{0j}=b(\lambda_0-\lambda_j)$. 
%\end{prop}
%

\subsubsection{Basic properties of the $R$-matrix}

%%%%%%%%%%%%%%%%%%%%%%%%%%%%%%%%
%\subsubsection{${\cal C}$}

Let us introduce some important properties of the $R$-matrix of the 
XXZ spin-chain.  

The $R$-matrix is invariant under the charge conjugation. 
For the symmetric $R$-matrix, 
we define the charge conjugation operator ${\cal C}$ by  
\be 
{\cal C}_{1 2 \cdots n} = \sigma^x_1 \cdots \sigma_n^{x} 
\ee
For a given operator $A \in End(V(\lambda_1) \otimes 
\cdots \otimes V(\lambda_n))$ 
we define ${\bar A}$ by 
${\bar A} = {\cal C}_{1 \cdots n} A {\cal C}_{1 \cdots n}$.  
For instance, we define ${\bar F}_{0, 1 \cdots n}$ by 
\be 
{\bar F}_{0, 1 \cdots n} = {\cal C}_{0 1 \cdots n} F_{0, 1 \cdots n} 
{\cal C}_{0 1 \cdots n} 
\ee

%In terms of the Pauli matrix $\sigma^x=e^{1,2}+e^{2,1}$ 

\begin{prop}
The charge conjugation operator ${\cal C}$ commutes with the 
monodromy matrix of the symmetric $R$-matrix: 
\be 
{[} {\cal C}_{0 1 \cdots n}, \,  R_{0, 1 \cdots n}  {]}= 0  
\ee
We thus have ${\bar A}_{1 \cdots n}(\lambda_0)= D_{1 \cdots n}(\lambda_0)$ and 
${\bar B}_{1 \cdots n}(\lambda_0)= C_{1 \cdots n}(\lambda_0)$.  
%\be {\cal C}_{1 \cdots n} A_{1 \cdots n}(\lambda_0) {\cal C}_{1 \cdots n} 
%= D_{1 \cdots n}(\lambda_0)  \, , \quad 
%{\cal C}_{1 \cdots n} B_{1 \cdots n}(\lambda_0)  {\cal C}_{1 \cdots n} 
%= C_{1 \cdots n}(\lambda_0) \, .  \ee
\end{prop}

%\subsubsection{Crossing symmetry}

\begin{lem}[Crossing symmetry]  
The $R$-matrix has the crossing symmetry relation: 
\be 
\left( \gamma \otimes I \right) R_{12}(\lambda_1-\eta, \lambda_2) 
\left( \gamma \otimes I \right) 
= b_{21}^{-1} R_{21}^{t_1}(\lambda_2, \lambda_1)  
\ee
where $b_{21}=b(\lambda_2 - \lambda_1)$ and $\gamma$ is given by 
\be 
\gamma= \sigma^y = \left( 
\begin{array}{cc}  
0 & -i \\ 
i & 0 
\end{array}
\right) 
\ee
\end{lem}
Here the crossing symmetry 
is slightly different from \cite{MS2000}.

\begin{lem}(Crossing symmetry of the monodromy matrix) 
\be 
\gamma_0 R_{0, 1 \cdots n}(\lambda_0-\eta; \xi_1, \ldots, \xi_n) \gamma_0 
=  \left(  \prod_{i=1}^{n} b^{-1}(\xi_i - \lambda_0) \right)
R^{t_0}_{1 \cdots n, 0}(\lambda_0; \xi_1, \ldots, \xi_n)  
\ee
\end{lem}

%%%%%%%%%%%%%%%%%%%%%%%%%%%%%%%%%%%%%%%%%%%%
%\subsubsection{Operation $X^{\dagger}$}

Let us introduce $\dagger$ operation. 
We shall use it when we pseudo-diagonalize the $B$ operators. 
\begin{df}
For $X_{1 \cdots n}(\lambda_1, \cdots, \lambda_n) 
\in End(V(\lambda_1) \otimes \cdots \otimes V(\lambda_n))$ 
we define $X^{\dagger}_{1 \cdots n}$ by 
\be 
X^{\dagger}_{1 \cdots n}(\lambda_1, \ldots, \lambda_n)
= X^{t_1 \cdots t_n}_{1 \cdots n}(-\lambda_1, \cdots, -\lambda_n)
\ee  
\end{df}
Here we note 
$(X^{\dagger})^{\dagger} = X$, and 
$(XY)^{\dagger} = Y^{\dagger} X^{\dagger}$.  
It is easy to show $R_{12}^{\dagger} = R_{21}$.

\begin{lem} Under the $\dagger$ operation the monodromy matrix is 
given by the following:  
\be 
 R^{\dagger}_{0, 1 \cdots n} = R_{0, 1 \cdots n}^{-1} 
= R_{1 \cdots n, \, 0}  
\label{eq:Rdag-inverse} \\  
\ee
\end{lem}

We define operators $A^{\dagger}_{1 \cdots n}$, 
$B^{\dagger}_{1 \cdots n}$, $C^{\dagger}_{1 \cdots n}$ and 
$D^{\dagger}_{1 \cdots n}$ by 
\be 
R^{\dagger}_{0, 1 \cdots n} = 
\left( \begin{array}{cc} 
A^{\dagger}_{1 \cdots n}(\lambda_0)  &  C^{\dagger}_{1 \cdots n}(\lambda_0) \\
B^{\dagger}_{1 \cdots n}(\lambda_0) & D^{\dagger}_{1 \cdots n}(\lambda_0)
\end{array} 
\right)_{[0]} 
\ee

\begin{prop} Under the $\dagger$ operation the monodromy matrix  
is given by 
\bea 
%R^{\dagger}_{0, 1 \cdots n} 
&  & 
R_{0, 1 2\cdots n}^{\dagger} = 
\left( \begin{array}{cc} 
A^{\dagger}_{1 \cdots n}(\lambda_0)  &  C^{\dagger}_{1 \cdots n}(\lambda_0) \\
B^{\dagger}_{1 \cdots n}(\lambda_0) & D^{\dagger}_{1 \cdots n}(\lambda_0)
\end{array} 
\right)_{[0]} \non \\
&  & \quad = \left(\prod_{i=1}^{n} b(\xi_i-\lambda_0) \right)
\left( \begin{array}{cc} 
A_{1 \cdots n}(\lambda_0-\eta)  &  - B_{1 \cdots n}(\lambda_0-\eta) \\
- C_{1 \cdots n}(\lambda_0-\eta) & D_{1 \cdots n}(\lambda_0-\eta)
\end{array} 
\right)_{[0]} 
\eea
\end{prop}

\subsubsection{The diagonalized forms of operators $A$ and $D$}  
%and $A^{\dagger}$ and $D^{\dagger}$} 

%%%%%%%%%%%%%%%%%%%%%%%%%%%%%%%%%%%%%%%%%%%%%%%%%%%%%%%%%%%%%
%\subsection{Diagonalizing  operators} 

Let us give the diagonalized forms of the $A$ and $D$ operators 
\cite{MS2000}.   

The following criterion for the $F$-basis to be non-singular 
should be useful. 
\begin{prop}
The determinants of the partial and total $F$ matrices are given by 
\be
{\rm det} F_{0, 1 \cdots n} 
= \prod_{j=1}^{n} b(\lambda_{0}-\xi_{j}) \, , \quad 
{\rm det} F_{1\cdots n} = \prod_{1 \le i <j \le n} b(\xi_i - \xi_j) 
\ee
\label{prop:detF}
\end{prop}
Proposition \ref{prop:detF} follows 
from lemma \ref{lem:eigenAD} of Appendix D. 

We can show the diagonalized forms of 
operators $A$ and $D$ as follows \cite{MS2000}.  
\begin{prop}[Diagonalization of $A$ and $D$]
\bea 
%\widehat{D}_{1 \cdots n}(\lambda_0) 
F_{1 \cdots n} D_{1 \cdots n}(\lambda_0) F_{1 \cdots n}^{-1} 
 & = &  
\bigotimes_{i=1}^{n} 
\left( 
\begin{array}{cc}
b_{0 i} & 0 \\
0 & 1 
\end{array}
\right)_{[i]} 
\label{eq:diagD} \\ 
%
%
%\widehat{A}_{1 \cdots n}(\lambda_0) 
{\bar F}_{1 \cdots n} A_{1 \cdots n}(\lambda_0) 
{\bar F}_{1 \cdots n}^{-1}  & = & 
\bigotimes_{i=1}^{n} 
\left( 
\begin{array}{cc}
1 & 0 \\ 
0 & b_{0 i} 
\end{array}
\right)_{[i]} \, , 
\label{eq:diagA} 
\eea
where  $b_{0i}=b(\lambda_0 - \xi_i)$
\label{prop:diagDA}
\end{prop}

\begin{prop} [Diagonalization of $A^{\dagger}$ and $D^{\dagger}$]
\bea 
F_{1 \cdots n} A^{\dagger}_{1 \cdots n}(\lambda_0) F_{1 \cdots n}^{-1} 
& = & 
\bigotimes_{i=1}^{n} 
\left( 
\begin{array}{cc}
1 & 0 \\
0 & b_{i 0} 
\end{array}
\right)_{[i]} \label{eq:diagA-dag} \\ 
{\bar F}_{1 \cdots n} D^{\dagger}_{1 \cdots n}(\lambda_0) 
{\bar F}_{1 \cdots n}^{-1} & = & 
\bigotimes_{i=1}^{n} 
\left( 
\begin{array}{cc}
b_{i 0} & 0 \\ 
0 & 1  
\end{array}
\right)_{[i]}  \label{eq:diagD-dag}           
\eea  
where $b_{i0}=b(\xi_i - \lambda_0)$ 
\label{prop:diagADdagger}
\end{prop}
The derivation of the diagonalized forms and some useful formulas 
are briefly reviewed in Appendix D.

For a given operator $X_{1 \cdots n}(\lambda_1, \cdots, \lambda_n) 
\in End(V(\lambda_1) \otimes \cdots \otimes V(\lambda_n))$ 
we denote $F A F^{-1}$ by ${\widetilde F}$:   
\be 
{\widetilde X}_{1 \cdots n} = F_{1 2 \cdots n} X_{1 2 \cdots n} 
F_{1 2 \cdots n}^{-1} \, . 
\ee
For instance we have  ${\widetilde{D}}_{1 \cdots n}(\lambda_0) = 
{F}_{1 \cdots n} D_{1 \cdots n}(\lambda_0) {F}_{1 \cdots n}^{-1}$.

%%%%%%%%%%%%%%%%%%%%%%%%%%%%%%%%%%%%%%%%%%%%%%%%
%
%
\subsection{Pseudo-diagonalization of the $B$ operator}

%\subsubsection{Fundamental commutation relations}

Let us recall that the matrix elements of the monodromy matrix 
$R^{+}_{0, 1 \cdots L}$ are related to the symmetric ones as follows.    
\bea 
& & R^{+}_{0, 1 2 \cdots L}(u; \xi_1, \ldots, \xi_L)  = 
\left(
\begin{array}{cc} 
A^{+}_{12 \cdots L}(u; \xi_1, \ldots, \xi_L) &   
B^{+}_{12 \cdots L}(u; \xi_1, \ldots, \xi_L)  \\
C^{+}_{12 \cdots L}(u; \xi_1, \ldots, \xi_L) &   
D^{+}_{12 \cdots L}(u; \xi_1, \ldots, \xi_L)  
\end{array} 
\right)_{[0]} \non \\ 
& = & 
\left(
\begin{array}{cc} 
\chi_{1 2 \cdots L} A_{12 \cdots L}(u; \{ \xi_j \})
\left( \chi_{1 2 \cdots L} \right)^{-1} &   
e^{- \lambda_0} \chi_{1 2 \cdots L} B_{12 \cdots L}(u; \{ \xi_j \} )
\left( \chi_{1 2 \cdots L} \right)^{-1}  \\
e^{\lambda_0} \chi_{1 2 \cdots L} C_{12 \cdots L}(u; \{ \xi_j \} )
\left( \chi_{1 2 \cdots L} \right)^{-1} &   
\chi_{1 2 \cdots L} D_{12 \cdots L}(u; \{ \xi_j \} ) 
\left( \chi_{1 2 \cdots L} \right)^{-1}
\end{array} 
\right)_{[0]} \non \\ 
\eea
Then, from the quantum-group invariance (\ref{eq:q-inv(+)})
%(\ref{eq:affineR-cyclic}) 
we have the following commutation relations: 
\bea 
 B^{+}_{1 \cdots n}(\lambda) & = & D^{+}_{1 \cdots n}(\lambda) 
\Delta^{(n-1)}(X^{-}) -  q \Delta^{(n-1)}(X^{-}) D^{+}_{1 \cdots n}(\lambda)  
\label{eq:BDelta} \\
 C^{+}_{1\cdots n}(\lambda) & = & \Delta^{(n-1)}(X^{+}) 
D^{+}_{1 \cdots n}(\lambda)  
- q^{-1} D^{+}_{1 \cdots n}(\lambda)  \Delta^{(n-1)}(X^{+}) \label{eq:CDelta} 
\eea
Here $X^{\pm}$ are generators of $U_q(sl_2)$, and 
$\Delta^{(n-1)}(X^{-})$ denote the tensor-product representation of 
$\Delta^{(n-1)}(X^{-})$ acting on the $n$ sites from the 1st to $n$th.  
We remark that more generally, we have commutation relations 
 (\ref{eq:affineR-sigma}) for the affine quantum group 
$U_q({\widehat{sl_2}})$.

In this subsection we abbreviate the superscript $+$ for the asymmetric 
monodromy matrix, for simplicity. In fact, the essential parts of 
formulas such as the fundamental commutation relations 
are invariant under gauge transformations 
if we express them in terms of the generators of 
the quantum affine algebra $U_q(\widehat{sl_2})$ in  
the evaluation representation (\ref{eq:eval-rep}).    
Here we remark that the matrix representation 
of the evaluation representation of $U_q(\widehat{sl_2})$ 
can be changed through gauge transformations.

%%%%%%%%%%%%%%%%%%%%%%%%%%%%%%%%%%%%%% 
%\subsubsection{Useful formulas }

Let us now introduce some symbols. 
\begin{df}
We define operators ${\widehat \delta}_{jk}(\lambda_j, \lambda_k)$ 
for $j, k$ satisfying $0 \le j < k \le L$ by 
\be 
{\widehat \delta}_{jk}(\lambda_j, \lambda_k) = 
\left( 
\begin{array}{cccc}  
1 & 0 & 0 & 0 \\
0 & b^{-1}_{kj} & 0 & 0 \\
1 & 0 & b^{-1}_{jk} & 0 \\
1 & 0 & 0 & 1 \\
\end{array}
\right)_{[jk]}  \, ,  
\ee
where $b_{jk}=b(\lambda_j-\lambda_k)$ and $b_{kj}=b(\lambda_k-\lambda_j)$. 
We define ${\widehat \delta}_{1 \cdots n}$ 
and ${\widehat \delta}_{0, 1 \cdots n}$ by 
\bea 
{\widehat \delta}_{1 \cdots n} & = & 
\prod_{1\le j <k \le n} {\widehat \delta}_{j k}(\lambda_j, \lambda_k) 
\non \\  
{\widehat \delta}_{0, 1 \cdots n} & = & {\widehat \delta}_{0 1 \cdots n} 
{\widehat \delta}^{-1}_{1 \cdots n} 
= \prod_{j=1}^{n} {\widehat \delta}_{0j}(\lambda_0, \lambda_j)  
\eea 
\end{df} 
We define ${\widehat \delta}^{1 \cdots n}_i$ by  
\be 
{\widehat \delta}^{1 \cdots n}_i = 
\widehat{\delta}_{i, i+1 \cdots n 1 \cdots i-1} = 
\prod_{j=1; j \ne i}^{n} {\widehat \delta}_{ij}
\label{eq:delta-in}
\ee

Some useful formulas are given in Appendix D.

%\subsubsection{Derivation of the diagonalized $B$ operators}

Let us denote $I^{\otimes m} \otimes \Delta^{(\ell-1)}(x) 
\otimes I^{\otimes (n-\ell -m)}$ 
 by $\Delta^{(\ell-1)}_{m+1 \, m+2 \, \cdots \, m+\ell}(x)$ 
or $\Delta_{m+1 \, m+2 \, \cdots \, m+\ell}(x)$ 
for $x \in U_q(sl_2)$ 
in the tensor-product representation. 
\begin{lem} 
Let $X^{-}$ denote the generator of the quantum group $U_q(sl_2)$ 
and $X_j^{-}$ the spin-1/2 representation of $X^{-}$ acting 
on the $j$th site in the tensor product 
representation $(V^{(1)})^{\otimes n}$. We  have 
\be 
\widetilde{\Delta}_{1 \cdots n}(X^{-})  
= \left(X_1^{-} + e_1^{11} \widetilde{\Delta}_{2 \cdots n}(X^{-}) 
\widetilde{ A^{\dagger}}_{2 \cdots n}(\xi_1) 
 + e_1^{22} \widetilde{D}_{2 \cdots n}(\xi_1) 
\widetilde{\Delta}_{2 \cdots n}(X^{-}) \right) 
\widehat{ \delta}_{1, 2 \cdots n} 
\label{eq:Dn(X)}
\ee
\end{lem}
\begin{proof}
Making use of (\ref{eq:F^{-1}}) we show  
$F^{-1}_{1\cdots n} 
%= F^{t_1 \cdots t_n}_{n \cdots 1} \widehat{\delta}_{1 \cdots n}
= F^{t_1 \cdots t_n}_{n \cdots 2, 1} F_{2 \cdots n}^{-1}  
\widehat{\delta}_{1, 2 \cdots n}$. 
We have 
\bea 
\widetilde{\Delta}_{1 \cdots n}(X^{-}) & = & 
F_{2 \cdots n} F_{1, 2 \cdots n} {\Delta}^{(n-1)}(X^{-}) 
F^{t_1 \cdots t_n}_{n \cdots 2, 1} F_{2 \cdots n}^{-1}  
\widehat{\delta}_{1, 2 \cdots n}
\eea 
Putting $F_{1, 2 \cdots n} = e_1^{11} + e_1^{22} R_{1, 2 \cdots n}$ and  
$F^{t_1 \cdots t_n}_{n \cdots 2, 1}= e_1^{22} + R_{2 \cdots n, 1} e_1^{11}$,  
we have 
\bea 
& &  F_{1, 2 \cdots n} {\Delta}^{(n-1)}(X^{-}) 
F^{t_1 \cdots t_n}_{n \cdots 2, 1}
= \left( e_1^{11} + e_1^{22} R_{1, 2 \cdots n} \right) \, 
{\Delta}^{(n-1)}(X^{-}) \, 
\left( e_1^{22} + R_{2 \cdots n, 1} e_1^{11} \right)
\non \\ 
& & = 
 e_1^{11} {\Delta}^{(n-1)}(X^{-}) e_1^{22}  
+ e_1^{11} {\Delta}^{(n-1)}(X^{-})R_{2 \cdots n, 1} e_1^{11} 
\non \\
& & \quad + e_1^{22} R_{1, 2 \cdots n} {\Delta}^{(n-1)}(X^{-}) e_1^{22}  
+  e_1^{22} R_{1, 2 \cdots n} {\Delta}^{(n-1)}(X^{-})
R_{2 \cdots n, 1} e_1^{11} 
\non \\ 
& & = 0 + e_1^{11} {\Delta}_{2 \cdots n}(X^{-})R_{2 \cdots n, 1} e_1^{11} 
+ e_1^{22} R_{1, 2 \cdots n} {\Delta}_{2 \cdots n}(X^{-}) e_1^{22}  + X_1^{-} 
\, . 
\eea
Here we have made use of the following: 
\be {\Delta}^{(n-1)}(X^{-}) = {\Delta}^{(n-2)} \Delta(X^{-}) 
= X^{-} \otimes \Delta^{(n-2)} (K^{-1}) + I \otimes \Delta^{(n-2)}(X^{-}) \, .  \ee
We thus have  
\be 
\widetilde{\Delta}_{1 \cdots n}(X^{-}) =  \left( X_1^{-}
+ e_1^{11} \widetilde{\Delta}_{2 \cdots n}(X^{-}) 
\widetilde{R}_{2 \cdots n, 1} e_1^{11} 
+
 e_1^{22} \widetilde{R}_{1, 2 \cdots n} 
\widetilde{\Delta}_{2 \cdots n}(X^{-}) 
 e_1^{11} \right) \widehat{\delta}_{1, 2 \cdots n} \, . 
\ee
We obtain the case of $n$ from (\ref{eq:Rdag-inverse}). 
\end{proof} 

\begin{lem} 
In the tensor-product representation $(V^{(1)})^{\otimes n}$ 
we have 
\be 
{\widetilde \Delta}_{1 \cdots n}(X^{-}) = \sum_{i=1}^{n} 
X_i^{-} {\widehat \delta}^{1 \cdots n}_i \, . 
\label{eq:Xdelta}
\ee
\end{lem}
\begin{proof} 
We show it by induction on $n$. 
The case of $n=1$ is trivial. 
Let us assume the case of $n-1$. 
In eq. (\ref{eq:Dn(X)}),  
the first term gives the following: 
$X_1^{-} \widehat{\delta}_{1, 2\cdots n} 
= X_1^{-} \widehat{\delta}_1^{1 \cdots n}$.     
Assuming (\ref{eq:Xdelta}) for $\widehat{\Delta}_{2 \cdots n}$ and putting it 
into the second term of (\ref{eq:Dn(X)}), 
we have 
\bea 
e_1^{11} \widetilde{\Delta}_{2 \cdots n}(X^{-}) 
\widetilde{A^{\dagger}}_{2 \cdots n}(\xi_1) 
\widehat{\delta}_{1, 2 \cdots n}
& = & e_1^{11} \sum_{i=2}^{n} X_i^{-} 
\widehat{\delta}_i^{2 \cdots n} \bigotimes_{k=2}^{n}  
\left( 
\begin{array}{cc}
1 & 0 \\ 
 0 & b_{k 1} 
\end{array} 
\right)_{[k]} \widehat{\delta}_{1, 2 \cdots n} 
\non \\
& = & \sum_{i=2}^{n} X_i^{-} \widehat{\delta}_{i}^{1 \cdots n} 
e_1^{11} \non \\ 
& = & e_1^{11} \sum_{i=2}^{n} X_i^{-} \widehat{\delta}_{i}^{1 \cdots n} 
\eea
Here we have made use of (\ref{eq:diagA-dag}).  
Similarly, we have 
\be 
e_1^{22} \widetilde{D}_{2 \cdots n}(\xi_1) 
\widetilde{\Delta}_{2 \cdots n}(X^{-}) e_1^{22}
\widehat{\delta}_{1, 2 \cdots n}
 = e_1^{22} \sum_{i=2}^{n} X_i^{-} \widehat{\delta}_{i}^{1 \cdots n}  
\ee
Thus, we have the case of $n$ as follows. 
\bea 
{\widetilde \Delta}_{1 \cdots n}(X^{-}) 
& = & X_1^{-} {\widehat \delta}_{1}^{1 \cdots n} + 
(e_1^{11} + e_1^{22}) 
\sum_{i=2}^{n} X_1^{-} {\widehat \delta}_{i}^{1 \cdots n} 
\non \\
& = & \sum_{i=1}^{n} X_i^{-} {\widehat \delta}_i^{1 \cdots n} \, .    
\eea
\end{proof} 

From the fundamental commutation relation (\ref{eq:BDelta})
we have the following: 
\begin{lem} In the tensor product $(V^{(1)})^{\otimes n}$ we have 
\be 
{\widetilde B}_{1 \cdots n}(\lambda) = \sum_{i=1}^{n} X_i^{-} 
({\widetilde D}_{1 \cdots i-1, i+1 \cdots n}(\lambda)
 - q {\widetilde D}_{1 \cdots n}(\lambda) ) 
\widehat{\delta}_i^{1 \cdots n} \, .   
\label{eq:BDD}
\ee
\end{lem}
\begin{proof} 
We transform the both sides of 
the fundamental commutation relation (\ref{eq:BDelta}) by 
$F_{1 \cdots n}$, and  
put (\ref{eq:diagD}) and (\ref{eq:Xdelta}) into it, 
we have the following: 
\bea 
{\widetilde B}_{1 \cdots n}(\lambda) & = & 
{\widetilde D}_{1 \cdots n}(\lambda) {\widetilde \Delta}_{1 \cdots n}(X^{-}) 
- q {\widetilde \Delta_{1 \cdots n}}(X^{-}) 
{\widetilde D}_{1 \cdots n}(\lambda) \non \\
& = & 
\bigotimes_{j=1}^{n} 
%{\rm diag}(b_{0j}, 1)_{[j]} 
%
\left( 
\begin{array}{cc}
b_{0j} & 0 \\ 
0  &  1 
\end{array}
\right)_{[j]} 
\sum_{i=1}^{n} X_i^{-} {\widehat \delta}^{1 \cdots n}_{i} 
  - q \sum_{i=1}^{n} X_i^{-} {\widehat \delta}^{1 \cdots n}_{i} 
\bigotimes_{j=1}^{n} 
\left( 
\begin{array}{cc} 
b_{0j} & 0 \\ 
0 & 1
\end{array}
\right)_{[j]} \non \\
&  = & \sum_{i=1}^{n} X_i^{-} 
({\widetilde D}_{1 \cdots i-1, i+1 \cdots n}(\lambda)
 - q {\widetilde D}_{1 \cdots n}(\lambda) ) 
\widehat{\delta}_i^{1 \cdots n} \, . 
\eea
\end{proof}

\begin{prop}[Pseudo-diagonalization of $B$ operator] 
We have   
\be 
{\widetilde B}_{1 \cdots n}(\lambda) = 
\sum_{i=1}^{n} c_{0i}^{-} \, X^{-}_i \bigotimes_{j=1; j \ne i}^{n} 
\left( 
\begin{array}{cc} 
b_{0j} & 0 \\
0 & b_{ji}^{-1} 
\end{array}
\right)_{[j]} \, ,  
\label{eq:diagB}
\ee
where $b_{0i}=b(\lambda_0-\xi_i)$,  
$b_{ji}=b(\xi_j - \xi_i)$ and 
$c_{0i}^{-} = c^{-}(\lambda_0 - \xi_i)=\exp(-(\lambda_0-\xi_i)) 
c(\lambda_0 - \xi_i)$. 
\end{prop}
\begin{proof} 
Let us denote 
$b_{0 i}=\sinh(\lambda - \xi_i)/\sinh(\lambda - \xi_i + \eta)$ and 
 $c_{0 i}=\sinh(\eta)/\sinh(\lambda - \xi_i + \eta)$, 
by $b_{0i}$ and $c_{0i}$ respectively, .  
Putting $1 - q b_{01} = c_{0i}^{-}$ in (\ref{eq:BDD}) 
we show 
\be 
{\widetilde B}_{1 \cdots n}(\lambda) = 
\sum_{i=1}^{n} c_{0i}^{-} 
X_i^{-} {\widetilde D}_{1 \cdots i-1 \, i+1 \cdots n} 
{\widehat \delta}_{i}^{1 \cdots n} \, .  
\ee
After some calculation, we have (\ref{eq:diagB}).  
\end{proof}

Similarly, making use of lemmas \ref{lem:C1}, 
\ref{lem:C2} and \ref{lem:C3}, 
 we can show the diagonalized form of operator $C$. 
\begin{prop}
Let $X_i^{+}$ denote the 
spin-1/2 representation of $X^{+}$ acting on the 
$i$th site in the tensor product representation. We have   
\be 
{\widetilde C}_{1 \cdots n}(\lambda_0) = 
\sum_{i=1}^{n} c_{0i}^{+} \, X^{+}_i \bigotimes_{j=1; j \ne i}^{n} 
\left( 
\begin{array}{cc} 
b_{0j}b_{ij}^{-1} & 0 \\
0 &  1 
\end{array}
\right)_{[j]} \, ,  \label{eq:diagC} 
\ee
where $b_{0i}=b(\lambda_0-\xi_i)$, 
$b_{ji}=b(\xi_j - \xi_i)$ 
and $c_{0i}^{+} = c^{+}(\lambda_0 - \xi_i)$.  
\end{prop}

\subsection{Pseudo-diagonalized forms of 
the symmetric $B$ and $C$ operators}

Let us show the pseudo-diagonalized forms of 
the $B$ and $C$ operators 
of the symmetric monodromy matrix $R_{0, 1 \cdots n}$.  
Here we recall that expressions 
(\ref{eq:diagB}) and (\ref{eq:diagC}) are for 
$\widetilde{B}^{+}_{1 2 \cdots n}(\lambda)$ and 
$\widetilde{C}^{+}_{1 2 \cdots n}(\lambda)$, 
respectively. They are matrix elements of 
the asymmetric monodromy matrix 
$R^+_{0, 1 \cdots n} = \chi_{0 1 \cdots n} R_{0, 1 \cdots n} \chi^{-1}_{0 1 \cdots n}$. 
We have the following relations: 
\bea 
B^{+}_{1 2 \cdots n}(\lambda) & = & e^{-\lambda}  
\chi_{0 1 \cdots n} R_{0, 1 \cdots n} \chi^{-1}_{0 1 \cdots n}
\non \\ 
C^{+}_{1 2 \cdots n}(\lambda) & = & e^{\lambda}  
\chi_{0 1 \cdots n} R_{0, 1 \cdots n} \chi^{-1}_{0 1 \cdots n}
\eea
Therefore, applying the inverse gauge transformation ${\bar \chi}$ 
 to (\ref{eq:diagB}) and (\ref{eq:diagC}), we obtain   
\be 
{\widetilde B}_{1 \cdots n}(\lambda) = 
\sum_{i=1}^{n} c_{0i} \, \sigma^{-}_i \bigotimes_{j=1; j \ne i}^{n} 
\left( 
\begin{array}{cc} 
b_{0j} & 0 \\
0 & b_{ji}^{-1} 
\end{array}
\right)_{[j]} \, ,  
\label{eq:diag-symB}
\ee
and 
\be 
{\widetilde C}_{1 \cdots n}(\lambda_0) = 
\sum_{i=1}^{n} c_{0i} \, \sigma^{+}_i \bigotimes_{j=1; j \ne i}^{n} 
\left( 
\begin{array}{cc} 
b_{0j}b_{ij}^{-1} & 0 \\
0 &  1 
\end{array}
\right)_{[j]} \, . \label{eq:diag-symC}   
\ee
Here we recall $c_{0i}= \sinh(\eta)/\sinh(\lambda-\xi_i + \eta)$.   

We should remark that expressions 
(\ref{eq:diag-symB}) and (\ref{eq:diag-symC}) coincide with 
eq. (2.29) and (2.30) of Ref. \cite{KMT1999}, respectively.

%%%%%%%%%%%%%%%%%%%%%%%%%%%%%%%%%%%%%%%%%%%%%%%%%%%%%%%%%%%
%
%                  SECTION 6 
%
%%%%%%%%%%%%%%%%%%%%%%%%%%%%%%%%%%%%%%%%%%%%%%%%%%%%%%%%%%%
 \setcounter{equation}{0} 
 \renewcommand{\theequation}{6.\arabic{equation}}
\section{Scalar products formulas}

\subsection{Formula for higher-spin scalar products}

Let us consider the case of tensor product of spin-$s$ representations. 
We recall that $\ell=2s$ and $L=\ell N_{s}$. 
We introduce parameters $\xi_j^{(\ell; \epsilon)}$ for $j=1, 2, \ldots, L$, 
as follows:   
\be 
\xi_{(k-1) \ell + j}^{(\ell; \epsilon)} 
= \zeta_k - (j-1) \eta + \ell \eta/2 +  \epsilon r_j \qquad  j=1, \ldots \ell; 
k= 1, \ldots, N_{s} \, .  
\label{eq:complete-strings}
\ee
Here $r_j$ $(j=1, 2, \ldots, \ell)$ are distinct and nonzero parameters, and 
$\epsilon$ is an arbitrary small number. 
We also introduce the following symbol: 
\be 
P^{(\ell)}_{1 \cdots L} = 
\prod_{j=1}^{N_{s}} P^{(\ell)}_{(j-1) \ell +1}  
\, , \quad 
P^{(\ell) \, {\bar \chi}}_{1 \cdots L} = 
\prod_{j=1}^{N_{s}} \left( P^{(\ell)}_{(j-1) \ell +1} \right)^{\bar \chi} 
\ee
Here we recall that $B$ operator acting on the 
tensor product of spin-$s$ representations, 
$\left( V^{(2s)} \right)^{\otimes N_{s}}$, is given by 
$B$ operator acting on the tensor product of 
spin-1/2 representations $\left( V^{(1)} \right)^{\otimes L}$ 
with $L=N_{s} \ell$ and multiplied 
by the projection operators:  
$$
B^{(\ell)}_{1 \cdots N_{s}}(u;  \zeta_1, \ldots, \zeta_{N_{s}}) 
= P_{1 \cdots L}^{(\ell)} \, B_{1 \cdots L}^{(1)}
(u; \xi_1^{(\ell)}, \ldots, \xi_{L}^{(\ell)}) 
\, P_{1\cdots L}^{(\ell)} \, . 
$$

We now define the scalar product for the spin-$\ell/2$ case as follows. 
\begin{df} Let $\{ \lambda_{\alpha} \}$ ($\alpha= 1, 2, \ldots, n$) be 
a set of solutions of the Bethe ansatz equations and 
$\{ \mu_{j} \}$ ($j= 1, 2, \ldots, n$) be arbitrary numbers. 
We define the scalar product 
 $S^{(\ell)}_{n}(\{ \mu_j \}, \{ \lambda_{\alpha} \} ; \{ \zeta_k \})$ 
by the following:  
\be 
S^{(\ell)}_{n}(\{ \mu_j \}, \{ \lambda_{\alpha} \} ; \{ \zeta_k \})  
= 
\langle 0 | \, C^{(\ell)}(\mu_1) \cdots C^{(\ell)}(\mu_n) 
B^{(\ell)}(\lambda_1) 
\cdots B^{(\ell)}(\lambda_n) 
\, | 0 \rangle
\ee
Here $C^{(\ell)}(\mu_j)$ and $B^{(\ell)}(\lambda_{\alpha})$ abbreviate 
$C^{(\ell)}_{1 \cdots N_{s}}(\mu_j; \{ \zeta_j \})$ and 
$B^{(\ell)}_{1 \cdots N_{s}}(\lambda_{\alpha}; \{ \zeta_j \})$, 
respectively, and $\zeta_k$ denote the
 centers of $\ell_k$-strings of the inhomogeneous parameters 
$\{ \xi_k^{\ell} \}$.   
\end{df}

%Here we recall that for a set of solutions of the Bethe ansatz equations 
%$\{ \lambda_{\alpha}$ we define the scalar product for the spin-1/2 case 
%as follows. 
%\be S^{(1)}_{n}(\{ \mu_j \}, \{ \lambda_{\alpha} \} ; \{ \xi_j \})  
%= \langle 0 | \, C(\mu_1) \cdots C(\mu_n) B(\lambda_1) 
%\cdots B(\lambda_n) \, | 0 \rangle \ee
%Here $C(\mu_j)$ and $B(\lambda_{\alpha})$ briefly express  
%$C^{(1)}_{1 \ldots L}(\mu_j; \{ \xi_k \})$ and  
%$B^{(1)}_{1 \ldots L}(\lambda_{\alpha}; \{ \xi_k \})$, respectively,  
%and $\xi_j$ are inhomogeneous parameters.  

We calculate the scalar product for the higher-spin XXZ chains 
by the  formula in the next proposition. 
\begin{prop} 
Let $\{ \lambda_{\alpha} \}$ satisfy the Bethe ansatz equations for the 
spin-$\ell/2$ case. 
The scalar product of the spin-$\ell/2$ XXZ spin chain is 
reduced into that of the spin-1/2 XXZ spin chain as follows: 
\be 
S^{(\ell)}_{n}(\{ \mu_j \}, \{ \lambda_{\alpha} \} ; \{ \zeta_k \})  
= \lim_{\epsilon \rightarrow 0} \Bigg[ 
S^{(1)}_{n}(\{ \mu_j \}, \{ \lambda_{\alpha} \} ; 
\{ \xi^{(\ell; \epsilon)}_k \})  \Bigg]
\ee
\end{prop} 
\begin{proof} 
We now calculate 
the scalar product making use of eq. (\ref{eq:PR=PRP}) 
of lemma \ref{lem:PR=PRP} as follows.  
\bea
& & \langle 0 | \, C^{(\ell)}_{1 \cdots N_{s}}(\mu_1; \{ \zeta_j \}) 
\cdots C^{(\ell)}_{1 \cdots N_{s}}(\mu_n; \{ \zeta_j \}) 
B^{(\ell)}_{1 \cdots N_{s}}(\lambda_1; \{ \zeta_j \}) 
\cdots B^{(\ell)}_{1 \cdots N_{s}}(\lambda_n; \{ \zeta_j \}) 
\, | 0 \rangle \non \\ 
& = & 
\langle 0 | \left( P_{1 \cdots L}^{(\ell) \, {\bar \chi}} 
C^{(1)}_{1 \cdots L}(\mu_1; \{ \xi_j^{(\ell)} \}) 
P_{1 \cdots L}^{(\ell) \, {\bar \chi}} \right) 
\, \cdots \, \left( P_{1 \cdots L}^{(\ell) \, {\bar \chi}} 
C^{(1)}_{1 \cdots L}(\mu_n; \{ \xi_j^{(\ell)} \}) 
P_{1 \cdots L}^{(\ell) \, {\bar \chi}} \right) 
\non \\  
& & \quad \cdot \, \left( P_{1 \cdots L}^{(\ell) \, {\bar \chi}}  
B^{(1)}_{1 \cdots L}(\lambda_1; \{ \xi_j^{(\ell)} \}) P_{1 \cdots L}^{(\ell)} \right) \, 
\cdots \left( P_{1 \cdots L}^{(\ell) \, {\bar \chi}}  
B^{(1)}_{1 \cdots L}(\lambda_n; \{ \xi_j^{(\ell)} \})P_{1 \cdots L}^{(\ell)} 
\right) \, | 0 \rangle  \non \\ 
& = & 
\langle 0 | P_{1 \cdots L}^{(\ell) \, {\bar \chi}} \, \times \, 
C^{(1)}_{1 \cdots L}(\mu_1; \{ \xi_j^{(\ell)} \})  
\, \cdots \, C^{(1)}_{1 \cdots L}(\mu_n; \{ \xi_j^{(\ell)} \})  
\non \\  
& & \quad \cdot \,  
B^{(1)}_{1 \cdots L}(\lambda_1; \{ \xi_j^{(\ell)} \}) \, 
\cdots   
B^{(1)}_{1 \cdots L}(\lambda_n; \{ \xi_j^{(\ell)} \}) \, \times \, 
P_{1 \cdots L}^{(\ell) \, {\bar \chi}} \, | 0 \rangle  
\eea
Here we note that we have 
$\langle 0| P_{1 \cdots L}^{(\ell)} = \langle 0|$ and 
 $P_{1 \cdots L}^{(\ell)} | 0 \rangle = | 0 \rangle$. 
Moreover, we have  
$\langle 0| P_{1 \cdots L}^{(\ell) \, {\bar \chi}}  = \langle 0|$ and  
$P_{1 \cdots L}^{(\ell) \, {\bar \chi}} | 0 \rangle  = | 0 \rangle$.   
We thus have 
\bea
& & \langle 0 | \, C^{(\ell)}_{1 \cdots N_{s}}(\mu_1; \{ \zeta_j \}) 
\cdots C^{(\ell)}_{1 \cdots N_{s}}(\mu_n; \{ \zeta_j \}) 
B^{(\ell)}_{1 \cdots N_{s}}(\lambda_1; \{ \zeta_j \}) 
\cdots B^{(\ell)}_{1 \cdots N_{s}}(\lambda_n; \{ \zeta_j \}) 
\, | 0 \rangle \non \\ 
& = & 
\langle 0 | C^{(1)}_{1 \cdots L}(\mu_1; \{ \xi_j^{(\ell)} \})  
\, \cdots \, C^{(1)}_{1 \cdots L}(\mu_n; \{ \xi_j^{(\ell)} \}) \, \cdot \,  
B^{(1)}_{1 \cdots L}(\lambda_1; \{ \xi_j^{(\ell)} \}) \, 
\cdots   
B^{(1)}_{1 \cdots L}(\lambda_n; \{ \xi_j^{(\ell)} \}) \, | 0 \rangle \, 
\non \\
\eea
We evaluate the last line through  
the following limit of sending $\epsilon$ to zero: 
\bea 
& & \langle 0 | C^{(1)}_{1 \cdots L}(\mu_1; \{ \xi_j^{(\ell)} \})  
\, \cdots \, C^{(1)}_{1 \cdots L}(\mu_n; \{ \xi_j^{(\ell)} \}) \, \cdot \,  
B^{(1)}_{1 \cdots L}(\lambda_1; \{ \xi_j^{(\ell)} \}) \, 
\cdots   
B^{(1)}_{1 \cdots L}(\lambda_n; \{ \xi_j^{(\ell)} \}) \, | 0 \rangle \, 
\non \\
& = & \lim_{\epsilon \rightarrow 0} 
\langle 0 | C^{(1)}_{1 \cdots L}(\mu_1; \{ \xi_j^{(\ell; \epsilon)} \})  
\cdots C^{(1)}_{1 \cdots L}(\mu_n; \{ \xi_j^{(\ell; \epsilon)} \})  \non \\  
& & \qquad \times 
B^{(1)}_{1 \cdots L}(\lambda_1; \{ \xi_j^{(\ell; \epsilon)} \}) \, 
\cdots   
B^{(1)}_{1 \cdots L}(\lambda_n; \{ \xi_j^{(\ell; \epsilon)} \}) \, | 0 \rangle 
\non \\ 
& = & \lim_{\epsilon \rightarrow 0} \Bigg[ 
\langle 0 | {\widetilde C}^{(1)}_{1 \cdots L}
(\mu_1; \{ \xi_j^{(\ell; \epsilon)} \})  
\cdots 
{\widetilde C}^{(1)}_{1 \cdots L}(\mu_n; \{ \xi_j^{(\ell; \epsilon)} \}) 
 \non \\  
& & \qquad \quad \times \, 
{\widetilde B}^{(1)}_{1 \cdots L}(\lambda_1; \{ \xi_j^{(\ell; \epsilon)} \}) 
\, \cdots \,   
{\widetilde B}^{(1)}_{1 \cdots L}(\lambda_n; \{ \xi_j^{(\ell; \epsilon)} \}) 
\Bigg] \, | 0 \rangle \, 
\eea
\end{proof}

We evaluate the spin-1/2 scalar product 
taking the limit of sending $\epsilon$ to 0, 
so that we can make 
the determinant of $F_{L \cdots 2 1}$ being nonzero.  
Here we remark that 
the operator $F_{L \cdots 2 1}$ appears in the 
pseudo-diagonalization process of the $B$ and $C$ operators, 
as shown in Section 5, and also that    
the determinant of $F_{L \cdots 2 1}$ vanishes at $\epsilon=0$, 
when parameters $\xi_j$ are given by eq. (\ref{eq:complete-strings}).  
In fact,  if we put some inhomogeneous parameters $\xi_j$ 
in the form of a ``complete $\ell$-string'' \cite{Odyssey},  
that is, for some integers $\ell$, $m$ and a constant $z$, 
we have $\xi_{m + j} = z - j \eta$ for $j= 1, 2, \ldots, \ell$, 
then the determinant of $F_{L \cdots 2 1}$ vanishes.  
Here we also note that ${\rm det} F_{1 2 \cdots L} \ne 0$ even 
at $\epsilon=0$.

Let us discuss the mixed spin case.  
We set the inhomogeneous parameters as follows: 
\be 
\xi_{\ell_1 + \cdots \ell_{k-1} + j}^{({\bm \ell}; \epsilon)} 
= \zeta_k - (j-1) \eta + 
(\ell-1) \eta/2 +  \epsilon r_j \qquad  j=1, \ldots \ell; \, 
k= 1, \ldots, r \, .  
\ee
Let us define $P^{({\bm \ell})}_{1 2 \cdots L}$ 
and $P^{({\bm \ell}) \, {\bar \chi}}_{1 2 \cdots L}$ by 
\bea 
P^{({\bm \ell})}_{1 2 \cdots L} & = & 
\prod_{k=1}^{r} P^{(\ell_k)}_{\ell(k-1)+1} \non \\ 
P^{({\bm \ell}) \, {\bar \chi} }_{1 2 \cdots L} & = & 
\prod_{k=1}^{r} \left( P^{(\ell_k)}_{\ell(k-1)+1} \right)^{\bar \chi} 
\eea
It is easy to see the following: 
\bea 
\langle 0 | P^{({\bm \ell}) }_{1 2 \cdots L} & = & \langle 0 | \, , \quad 
P^{({\bm \ell})}_{1 2 \cdots L} | 0 \rangle = | 0 \rangle \, , 
\non \\ 
\langle 0 | P^{({\bm \ell}) \, {\bar \chi} }_{1 2 \cdots L} & = & 
\langle 0 | \, , \quad 
P^{({\bm \ell}) \, {\bar \chi} }_{1 2 \cdots L} | 0 \rangle = | 0 \rangle 
\eea
We now define the scalar product for the mixed spin case as follows. 
\be 
S^{({\bm \ell})}_{n}(\{ \mu_j \}, \{ \lambda_{\alpha} \} ; \{ \zeta_k \})  
= 
\langle 0 | \, C^{({\bm \ell})}(\mu_1) \cdots C^{({\bm \ell})}(\mu_n) 
B^{({\bm \ell})}(\lambda_1) 
\cdots B^{({\bm \ell})}(\lambda_n) 
\,  | 0 \rangle
\ee
Here $C^{({\bm \ell})}(\mu_j)$ and $B^{({\bm \ell})}(\lambda_{\alpha})$ 
abbreviate 
$C^{({\bm \ell})}_{1 \cdots r}(\mu_j; \{ \zeta_j \})$ and 
$B^{({\bm \ell})}_{1 \cdots r}(\lambda_{\alpha}; \{ \zeta_j \})$, 
respectively.

\begin{prop} 
Let $\{ \lambda_{\alpha} \}$ satisfy the Bethe ansatz equations 
for the mixed spin case. 
The scalar product of the mixed-spin XXZ spin chain is 
reduced into that of the spin-1/2 XXZ spin chain as follows: 
\be 
S^{({\bm \ell})}_{n}(\{ \mu_j \}, \{ \lambda_{\alpha} \} ; \{ \zeta_k \})  
= 
\lim_{\epsilon \rightarrow 0} \Bigg[ S^{(1)}_{n}(\{ \mu_j \}, 
\{ \lambda_{\alpha} \} ; \{ \xi^{({\bm \ell}; \epsilon)}_k \}) \Bigg] 
\ee
\end{prop} 

\subsection{Determinant expressions of the scalar products}

Let us review the result of the spin-1/2 case \cite{KMT1999}. 
Suppose that  $\lambda_{\alpha}$ for $\alpha=1, 2, \ldots, n$,  
are solutions of the Bethe ansatz equations with 
in homogeneous parameters $\xi_j$ for $j=1, 2, \ldots, L$,   
 the scalar product is defined by  
\be 
S_n(\{ \mu_j \}, \, \{ \lambda_{\alpha} \}; \{\xi_k \})   
= \langle 0 | \prod_{j=1}^{n} C(\mu_j; \{\xi_k \}) 
\prod_{\alpha=1}^{n} B(\lambda_{\alpha}; \{\xi_k \}) | 0 \rangle 
\ee
Here $\mu_{j}$ for $j=1, 2, \ldots, n$ are arbitrary. 
We note that 
$S_n(\{ \mu_j \}, \, \{ \lambda_{\alpha} \}; \{\xi_k \})$ 
has been denoted by 
$S_n^{(1)}(\{ \mu_j \}, \, \{ \lambda_{\alpha} \}; \{\xi_k \})$ 
in the last subsection. 
Then, the exact expression of the scalar product has been shown 
through the pseudo-diagonalized forms of the $B$ and $C$ operators as follows 
\cite{KMT1999}: 
\bea 
& & S_n(\{ \mu_j \}, \, \{ \lambda_{\alpha} \} ; \{ \xi_k \})  
 =  \langle 0 | \prod_{j=1}^{n} \widetilde{C}(\mu_j) 
\prod_{\alpha=1}^{n} \widetilde{B}(\lambda_{\alpha}) | 0 \rangle \non \\ 
& = & 
{\frac {\prod_{\alpha=1}^{n} \prod_{j=1}^{n} 
\sinh(\mu_j - \lambda_{\alpha})}   
{\prod_{j >k} \sinh(\mu_k- \mu_j) 
 \prod_{\alpha < \beta} \sinh(\lambda_{\beta} - \lambda_{\alpha}) } } \,    
{\rm det} T( \{ \mu_j \}, \, \{ \lambda_{\alpha} \} ; \{\xi_k \} ) 
\eea
Here the matrix elements 
$T_{ab}$ for $a, b=1, \ldots, n$, are given by 
\be 
T_{ab}= {\frac {\partial} {\partial \lambda_{\alpha} } } 
\tau(\mu_b, \{ \lambda_k \}; \{\xi_k \}) 
\ee
where 
\bea 
\tau(\mu, \{ \lambda_k \}; \{\xi_k \})  
& = & a(\mu) \prod_{k=1}^{n} b^{-1}(\lambda_k-\mu) 
+ d(\mu; \{\xi_j \}; \{\xi_k \}) \prod_{k=1}^{n} b^{-1}(\mu- \lambda_k) 
\non \\  
\eea
and 
\be 
a(\mu)=1 \, , \quad d(\mu; \{ \xi_j \})=\prod_{j=1}^{L} b(\mu-\xi_j)  \, . 
\ee

Let us express the scalar product 
of the higher-spin case in terms of 
 the determinant of the matrix $T$.  
In the tensor product of spin-$\ell/2$ representations we have 
\bea 
& & S^{(\ell)}_n(\{ \mu_j \}, \, \{ \lambda_{\alpha} \} ; \{ \zeta_k \}) 
\non \\ 
& = & 
{\frac {\prod_{\alpha=1}^{n} \prod_{j=1}^{n} \sinh(\mu_j - \lambda_{\alpha})}  
{\prod_{j >k} \sinh(\mu_k- \mu_j) 
 \prod_{\alpha < \beta} \sinh(\lambda_{\beta} - \lambda_{\alpha}) } } \,    
{\rm det} T( \{ \mu_j \}, \, \{ \lambda_{\alpha} \} ; \{ \xi_k^{(\ell)} \}) 
\eea
In the mixed-spin case we have 
\bea 
& & S^{({\bm \ell})}_n(\{ \mu_j \}, \, 
\{ \lambda_{\alpha} \} ; \{ \zeta_k \}) 
\non \\ 
& = & 
{\frac {\prod_{\alpha=1}^{n} 
\prod_{j=1}^{n} \sinh(\mu_j - \lambda_{\alpha})} 
{\prod_{j > k} \sinh(\mu_k- \mu_j) 
 \prod_{\alpha < \beta} \sinh(\lambda_{\beta} - \lambda_{\alpha}) } } \,    
{\rm det} T( \{ \mu_j \}, \, \{ \lambda_{\alpha} \} ; 
\{ \xi_k^{({\bm \ell})} \}) 
\eea

%%%%%%%%%%%%%%%%%%%%%%%%%%%%%%%%%%%%%%%%%%%
\subsection{Norms of the Bethe states in the higher-spin case}

For two sets of $n$ parameters, $\mu_1, \ldots, \mu_n$ 
and $\lambda_1, \ldots, \lambda_n$, 
we define matrix elements $H_ab$ by  
\be 
H_{ab}(\{ \mu_j \}, \{ \lambda_{\alpha}; \{ \xi_k \} )  
 = {\frac {\sinh \eta} {\sinh(\lambda_a - \mu_b)}} 
\left( {\frac {a(\mu_b)} {d(\mu_b)}} 
\prod_{k=1; \ne a}^{n}\sinh(\lambda_k - \mu_b + \eta)  
- \prod_{k=1; \ne a}^{n}\sinh(\lambda_k - \mu_b - \eta) \right) 
\ee

Let us assume that 
$\lambda_1, \ldots, \lambda_n$ are solutions of 
the Bethe ansatz equations. We have 
\be 
{\rm det} H(\{ \lambda_{\alpha} \}, \{ \mu_j \}; \{ \xi_k \})   
= {\rm det} T(\{ \lambda_{\alpha} \}, \{ \mu_j \}; \{ \xi_k \}) 
\prod_{\alpha=1}^{n} \prod_{j=1}^{n} \sinh(\mu_j- \lambda_{\alpha}) 
\left( \prod_{j=1}^{n} d(\mu_j) \right)^{-1}  
\ee 
Let us now take the limit of sending $\mu_j$ to $\lambda_j$ for each $j$. 
Then we have 
\be 
\lim_{\mu_j \rightarrow \lambda_j} 
det H(\{ \lambda_{\alpha} \}, \{ \mu_{j} \}; \{ \xi_k \} )   
= \sinh^n \eta \prod_{\beta=1}^{n} \prod_{m=1; m \ne \beta}^{n} 
\sinh(\lambda_m - \lambda_{\beta} - \eta) \cdot det \Phi^{'}
(\{ \lambda_{\alpha} \}) 
\ee
where matrix elements $\Phi^{'}_{ab}$ for 
$a, b= 1, \ldots, n$, are given by 
\be 
\Phi^{'}_{ab}(\{ \lambda_{\alpha} \}; \{ \xi \}) 
= - {\frac {\partial} {\partial \lambda_b} } 
\left( { \frac {a(\lambda_a;  \{ \xi_j \})} {d(\lambda_b;  \{ \xi_j \})} } 
\prod_{k=1; k \ne a}^{n} 
{\frac {b(\lambda_a - \lambda_k)} {b(\lambda_k - \lambda_a)} } \right)    
\ee

Suppose that  $\lambda_{\alpha}$ for $n=1, 2, \ldots, n$ are 
solutions of the Bethe ansatz equations. 
Gaudin's formula for the square of he norm of the Bethe state 
is given by 
\bea 
{\bm N}_n(\{ \lambda_{\alpha} \}; \{ \xi_j \} ) 
& = & \langle 0 | \prod_{j=1}^{n} C(\lambda_j) \prod_{j=1}^{n}B(\lambda_k)|
 0 \rangle \non \\ 
& = & \sinh^n \eta \prod_{\alpha, \beta=1; \alpha \ne \beta}^{n} 
b^{-1}(\lambda_{\alpha} - \lambda_{\beta}) \cdot {\rm det} 
\Phi^{'}({\lambda_{\alpha}}; \{ \xi_j \})  
\eea

Let us define the norm of the Bethe state for the mixed spin case of 
${\bm \ell}$ 
as follows. 
\be 
{\bm N}_n^{({\bm \ell})}(\{ \lambda_{\alpha} \}; \{ \xi^{({\bm \ell})}_k \}) 
 = \langle 0 | \prod_{j=1}^{n} C^{({\bm \ell})}(\lambda_j) 
\prod_{j=1}^{n} B^{({\bm \ell})}(\lambda_k) | 0 \rangle  
\ee
Then we have 
\be 
{\bm N}_n^{({\bm \ell})}(\{ \lambda_{\alpha} \}; \{ \xi^{({\bm \ell})}_k \}) 
= 
\sinh^n \eta \prod_{\alpha, \beta=1; \alpha \ne \beta}^{n} 
b^{-1}(\lambda_{\alpha} - \lambda_{\beta}) \cdot {\rm det} 
\Phi^{'}({\lambda_{\alpha}}; \{ \xi_j^{\bm \ell} \})  
\ee

%%%%%%%%%%%%%%%%%%%%%%%%%%%%%%%%%%%%%%%%%%%%%%%%%%%%%%%%%%%
%
%                  SECTION 7 
%
%%%%%%%%%%%%%%%%%%%%%%%%%%%%%%%%%%%%%%%%%%%%%%%%%%%%%%%%%%%
 \setcounter{equation}{0} 
 \renewcommand{\theequation}{7.\arabic{equation}}
\section{Form factors and the inverse-scattering for the higher-spin case}

\subsection{Formulas of the quantum inverse scattering problem}

Let us briefly review the derivation of the fundamental lemma of the 
quantum-inverse scattering problem for the spin-1/2 XXZ spin chain 
\cite{MS2000}. It will thus becomes clear how 
the pseudo-diagonalization of $B$ 
and $C$ are important.

\begin{prop} Let us denote by $p_q$  sequence $p_q=(1, 2, \ldots, n)$. 
Recall the notation 
$\widetilde{ R}_{0, p_q}= F_{1 2 \cdots n} R_{0, 1 2 \cdots n} 
F^{-1}_{1 2 \cdots n}$. Then, $\widetilde{ R}_{0, p_q}$  
 is invariant under any permutation. 
 We have 
\be 
\widetilde{ R}_{0, p_q} = \widetilde{ R}_{0, \sigma(p_q)} \, , 
\quad {\rm for}  \, \,  \sigma \in {\cal S}_n \, . 
\label{eq:Rinv}
\ee
\label{prop:Rinv}
\end{prop}
We shall show (\ref{eq:Rinv}) in Appendix B. 

The following lemma plays a central role 
in the quantum inverse-scattering problem \cite{KMT1999}.   
\begin{lem} For arbitrary inhomogeneous parameters 
$\xi_1, \xi_2, \ldots, \xi_L$ we have  
\be
x_i = \prod_{\alpha=1}^{i-1} (A+D)(\xi_{\alpha}) \,  
{\rm tr}_0(x_0 R_{0, 1 \cdots L}(\xi_i)) \, 
\prod_{\alpha=i+1}^{L} (A+D)(\xi_{\alpha})  
\label{eq:QISP}
\ee
\label{lem:QISP}
\end{lem} 
\begin{proof} 
For any operator $x_0$ defined on the auxiliary space, we have 
\bea 
{\rm tr}_0(x_0 R_{0, 1 2 \cdots L}(\lambda=\xi_i)) & = &  
F_{1 2 \cdots L}^{-1} 
{\rm tr}_0(x_0 \widetilde{R}_{0, 1 2 \cdots L}(\lambda=\xi_i)) 
F_{1 2 \cdots L} \non \\ 
& = &  
F_{1 2 \cdots L}^{-1} 
{\rm tr}_0(x_0 \widetilde{R}_{i, i+1 \cdots L 1 \cdots i-1}(\lambda=\xi_i)) 
F_{1 2 \cdots L} 
\non \\ 
& = &  
F_{1 2 \cdots L}^{-1} F_{i \cdots L 1 \cdots i-1} \,  
{\rm tr}_0(x_0 R_{i, i+1 \cdots L 1 \cdots i-1}(\lambda=\xi_i)) \, 
F_{i \cdots L 1 \cdots i-1}^{-1} F_{1 2 \cdots L}  
\non \\ 
& = &  
\left( F_{i \cdots L 1 \cdots i-1}^{-1} F_{1 2 \cdots L} \right)^{-1} 
x_i (A(\xi_i) + D(\xi_i)) \, 
\left( F_{i \cdots L 1 \cdots i-1}^{-1} F_{1 2 \cdots L} \right) 
\non  
\eea
Here we have used (\ref{eq:Rinv}), i.e. 
$\widetilde{R}_{0, 1 2 \cdots L}
= \widetilde{R}_{i, i+1 \cdots L 1 \cdots i-1}$. 
 From  the expression 
of $F^{-1}_{i \cdots L \, 1 \cdots i-1} F_{1 \cdots L}$
we now have   
\be 
{\rm tr}_0(x_0 R_{0, 1 \cdots L}) = \prod_{\alpha=1}^{i-1}  
\left( (A+D)(\xi_{\alpha}) \right)^{-1} 
\cdot x_i \cdot \prod_{\alpha=1}^{i} (A+D)(\xi_{\alpha})  \,  .  
\ee
\end{proof}

\subsection{Quantum inverse-scattering problem 
for the higher-spin operators}

Let us consider monodromy matrix $T^{+}_{0, 1 \cdots \ell N_{s}}$. 
Here we recall $L=\ell N_{s}$. For simplicity, 
we shall suppress the superscript `$+$' 
for $A, B, C$ and $D$ operators through this subsection.

We recall the following: 
$\Delta^{(n-1)}(K) =  K^{\otimes n}$ and 
\bea
\Delta^{(n-1)}(X^{+}) & = & \sum_{j=1}^{n} 
K^{\otimes (j-1)} \otimes X^{+}_j \otimes 
I^{\otimes (n-j)} \, , \non \\  
\Delta^{(n-1)}(X^{-}) & = & \sum_{j=1}^{n} 
I^{\otimes (j-1)} \otimes X^{-}_j \otimes 
\left( K^{-1} \right)^{\otimes (n-j)} \, . 
\eea
It is useful to note that for $i=1, 2, \ldots, \ell N_{s}$ we have 
\be 
K_i = \prod_{\alpha=1}^{i-1} (A+D)(\xi_{\alpha}) (q A + q^{-1} D)(\xi_{i}) 
\prod_{\alpha=i+1}^{\ell N_{s}} (A+D)(\xi_{\alpha}) \, . 
\ee

In the tensor product of $2 N_{s}$ spin-1/2 representations, 
$\left( V_{1}^{(1)} \otimes V_{2}^{(1)} \right)^{\otimes N_{s}}$,  
we have 
\bea
\Delta_{12}(X^{-}) & = & 
\left( X^{-}_1 \otimes K^{-1}_2 + I_1 \otimes X^{-}_2 \right) 
\otimes I^{\otimes 2(N_s-1)}   
\non \\ 
& = & 
\Bigg\{ \left( B(\xi_1) \cdot 
\prod_{\alpha=2}^{2 N_{s}} (A+D)(\xi_{\alpha}) \right) \cdot 
\bigg( (A+D)(\xi_1) (q^{-1} A + q D)(\xi_2) 
\prod_{\alpha=3}^{2 N_{s}} (A+D)(\xi_{\alpha}) \bigg) \non \\
& & \quad + (A+D)(\xi_1) \cdot B(\xi_2) \cdot 
\prod_{\alpha=3}^{2 N_{s}}(A+D)(\xi_{\alpha}) \Bigg\}  
\eea

In the $\ell$th tensor product of spin-1/2 representations,  
$V_{1}^{(1)} \otimes \cdots \otimes V_{\ell}^{(1)}$, 
%operators $\Delta^{(\ell-1)}(X^{\pm})$ are invariant 
%by multiplying projection operators. 
we have 
\be 
 P_{1 \cdots \ell}^{(\ell)} \,  
\cdot \, \Delta^{(\ell-1)}(X^{\pm}) 
\, \cdot \, P_{1 \cdots \ell}^{(\ell)} 
= 
P_{1 \cdots \ell}^{(\ell)} \,  
\cdot \,  \Delta^{(\ell-1)}(X^{\pm}) 
=
\Delta^{(\ell-1)}(X^{\pm}) 
\, \cdot \, P_{1 \cdots \ell}^{(\ell)}
\ee
In the tensor product of spin-$\ell/2$ representations, 
$V_{1}^{(\ell)} \otimes \cdots \otimes V_{N_s}^{(\ell)}$,  
we have for $i=1, 2, \ldots, N_s$ the following relations:  
\bea 
& & X_i^{- (\ell+)} = 
\Delta^{(\ell-1)}_{(i-1)\ell +1 \, \cdots i \ell } (X^{-}) 
\non \\
& & \quad 
=  \sum_{k=1}^{\ell} \prod_{\alpha=1}^{(i-1)\ell+k-1} 
(A+D)(\xi_{\alpha}) \cdot B(\xi_{(i-1)\ell + k}) \cdot 
\prod_{\alpha=(i-1)\ell+k+1}^{\ell N_{s}} (A+D)(\xi_{\alpha})  \non \\  
& & \quad 
\times \prod_{j=k+1}^{\ell}
%
%{(N_{s}-i+1) \ell} 
%
\prod_{\alpha=1}^{(i-1)\ell+j-1} (A+D)(\xi_{\alpha}) \cdot 
(q^{-1} A + q D)(\xi_{(i-1)\ell+j}) \cdot 
\prod_{\alpha=(i-1)\ell + j+1}^{\ell N_{s}} (A+D)(\xi_{\alpha})    
\non \\
& & 
=  \sum_{k=1}^{\ell} \prod_{\alpha=1}^{(i-1)\ell+k-1} 
(A+D)(\xi_{\alpha}) \cdot B(\xi_{(i-1)\ell + k}) \cdot 
 \non \\  
& & \quad \times \prod_{j=k+1}^{\ell}
(q^{-1} A + q D)(\xi_{(i-1)\ell+j}) \cdot 
\prod_{\alpha=i \ell+1}^{\ell N_{s}} (A+D)(\xi_{\alpha}) \, .   
\label{eq:X-(ell+)}
\eea
Here we have made use of the following: 
\be 
\prod_{\alpha=1}^{\ell N_s} (A+D)(\xi_{\alpha}) = I^{\otimes \ell N_s} \, . 
\ee
Similarly, we can express $X_i^{+ (\ell+)}$ and $K_i^{(\ell+)}$ as
follows.   
\bea 
& & X_i^{+ (\ell+)} = 
\Delta^{(\ell-1)}_{(i-1) \ell +1 \, \cdots i \ell } (X^{+}) \non \\ 
%
%& & \quad = \sum_{k=1}^{\ell} \prod_{j=1}^{k-1} 
%%\prod_{j=1-(i-1)\ell}^{k-1}   
%\prod_{\alpha=1}^{(i-1)\ell+j-1} 
%(A+D)(\xi_{\alpha}) \cdot (q A + q^{-1} D)(\xi_{(i-1)\ell+j}) \non \\ 
%& & \quad \cdot \prod_{\alpha=(i-1)\ell+j+1}^{\ell N_{s}} 
%(A+D)(\xi_{\alpha})  \non \\ 
%&&  \quad \cdot  \prod_{\alpha=1}^{(i-1)\ell+k-1} (A+D)(\xi_{\alpha})
% \cdot C(\xi_{(i-1) \ell +k })  
%\cdot \prod_{\alpha=(i-1)\ell + k+1}^{\ell N_{s}} (A+D)(\xi_{\alpha})  \, .   
%
& &  = \sum_{k=1}^{\ell} 
\prod_{\alpha=1}^{(i-1)\ell} (A+D)(\xi_{\alpha}) \cdot 
\prod_{j=1}^{k-1} (q A + q^{-1} D)(\xi_{(i-1)\ell+j})  
 \cdot C(\xi_{(i-1) \ell +k })  \non \\ 
& & \quad 
\prod_{\alpha=(i-1)\ell + k+1}^{\ell N_{s}} (A+D)(\xi_{\alpha})  \, .   
\label{eq:X+(ell+)}
\eea

\bea
& & K_i^{(\ell+)} = \Delta^{(\ell-1)}_{(i-1) \ell +1 \, \cdots i \ell }(K) 
\non \\
%& & = \prod_{j=1}^{\ell} \prod_{\alpha=1}^{(i-1) \ell + j -1} 
%(A+D)(\xi_{\alpha}) \cdot (q A + q^{-1} D)(\xi_{(i-1)\ell +j}) \cdot 
%\prod_{\alpha=(i-1) \ell + j+1}^{\ell N_{s}} (A+D)(\xi_{\alpha})   \non \\   
%
& & = \prod_{\alpha=1}^{(i-1) \ell} 
(A+D)(\xi_{\alpha}) \, \prod_{j=1}^{\ell} 
(q A + q^{-1} D)(\xi_{(i-1)\ell +j}) \,  
\prod_{\alpha=i \ell+1}^{\ell N_{s}} (A+D)(\xi_{\alpha}) \, . 
  \non \\   
\label{eq:K(ell+)}
\eea

%%%%%%%%%%%%%%%%%%%%%%%%%%%%%%%%%%%%%%%%%%%%%%%%%%%%%
%
%
\subsection{Useful formulas in the higher-spin case}

Let us denote by $X^{\pm (\ell)}$ the matrix representations 
of generators $X^{\pm}$ in the spin-$\ell/2$ representation of 
$U_q(sl_2)$. Here we recall that the matrix representations of 
$X^{\pm (\ell)}$ are obtained by calculating 
the action of $\Delta^{(\ell-1)}(X^{\pm})$ 
on the basis $\{ || \ell, n \rangle \}$.

We explicitly calculate the actions of 
$\sigma_1^{-}= \sigma^{-} \otimes I^{\otimes (\ell-1)}$ 
and $\sigma_{\ell}^{+}= I^{\otimes (\ell-1)} \otimes \sigma^{+}$  
on the basis $\{ || \ell, n \rangle \}$  
in the spin-$\ell/2$ representation.    
Multiplying projection operators to them, 
we obtain the following formulas:   
\bea 
P^{(\ell)}_{1 \cdots \ell} \, \sigma_{1}^{-} \, P^{(\ell)}_{1 \cdots \ell} 
& = & {\frac 1 {[\ell]_q}} \, X^{- (\ell+)} \non \\ 
P^{(\ell)}_{1 \cdots \ell} \, \sigma_{\ell}^{+} \, P^{(\ell)}_{1 \cdots \ell} 
& = & {\frac 1 {[\ell]_q}} \, X^{-(\ell+)}  \label{eq:Xpm}
\eea
Therefore, we have for $i=1, 2, \ldots, N_{s}$,  
the following formulas:  
\bea 
& & P^{( \ell )}_{(i-1) \ell + 1} \, X_{i}^{- (\ell+)} \, 
 P^{(\ell)}_{(i-1) \ell + 1} \non \\ 
& & =  
{[} \ell {]}_q  P^{(\ell)}_{(i-1)\ell+1} \, 
\prod_{\alpha=1}^{(i-1)\ell} (A^{+}+D^{+})(\xi_{\alpha})  \cdot  
B^{+}(\xi_{(i-1) \ell +1})   
\cdot  \prod_{\alpha=(i-1)\ell+2}^{\ell N_{s}} (A^{+}+D^{+})(\xi_{\alpha}) \, 
P^{(\ell)}_{(i-1)\ell+1} \non \\
\label{eq:PB}  \\  
& & P^{(\ell)}_{(i-1)\ell+1} \, X_{i}^{+ (\ell+)} \, 
P^{(\ell)}_{(i-1)\ell+1} \non \\ 
& & =  
[\ell]_q P^{(\ell)}_{(i-1)\ell+1} \, 
\prod_{\alpha=1}^{i \ell-2}(A^{+}+D^{+})(\xi_{\alpha}) \cdot 
C^{+}(\xi_{i\ell -1}) \cdot   
\prod_{\alpha= i \ell}^{\ell N_{s}}(A^{+}+D^{+})(\xi_{\alpha}) \, 
P^{(\ell)}_{(i-1)\ell+1} \label{eq:PC} \, . 
\eea
Taking advantage of projection operators, we thus have shown that 
the summation over $k$ arising from  the $(\ell-1)$th 
comultiplication operation can be calculated by a single term. 
This reduces the calculational task very much.

In the derivation of (\ref{eq:PB}), we first note      
\be 
\chi_{1 \cdots L} \sigma_{(i-1)\ell+1}^{-}  \chi_{1 \cdots L}^{-1} 
\exp( - \xi_{(i-1)\ell+1}) \, ,    
\ee
and then we show the following:  
\be 
\sigma_{(i-1)\ell+1}^{-}
= \prod_{\alpha=1}^{(i-1)\ell} (A^{+}+ D^{+})(\xi_{\alpha}) \, 
B^{+}(\xi_{(i-1)\ell+1}) \, \prod_{\alpha=(i-1)\ell+2}^{\ell N_s} 
(A^{+}+D^{+})(\xi_{\alpha}) \, .  
\ee
We shall show relations (\ref{eq:Xpm}) in Appendix C.

We now introduce useful formulas expressing any given operator 
in the spin-$\ell$ representation. 
Let us take two sets of integers $i_1, \ldots, i_m$ and 
$j_1, \ldots, j_n$ satisfying 
$1 \le i_1 < \cdots < i_m \le \ell$ and 
$1 \le j_1 < \cdots < j_n \le \ell$, respectively. 
We can show the following: 
\bea  
 || \ell, m \rangle \langle \ell, n || 
& = & 
\left[
\begin{array}{c} 
\ell \\
m 
\end{array} 
 \right]_q \, 
q^{m(m+1)/2 - n(n-1)/2 + n \ell 
- (i_1 + \cdots + i_m + j_1 + \cdots + j_n)} \non \\ 
& & \times \, 
P_{1 \cdots \ell}^{(\ell)} 
\left( \prod_{k=1}^{m} e_{i_k}^{21}
 \cdot \prod_{p=1;p \ne i_k, j_q}^{\ell} 
 e_{p}^{22} \cdot 
\prod_{q=1}^{n} e_{j_q}^{12} \right) 
P_{1 \cdots \ell}^{(\ell)}  
\label{eq:ellmn}
\eea    
Setting $i_1=1, i_2=2$, \ldots, $i_m=m$ and 
$j_1=1, j_2=2$, \ldots, $j_n=n$, 
we have for $m > n$ 
\be  
 || \ell, m \rangle \langle \ell, n || 
= 
\left[
\begin{array}{c} 
\ell \\
n 
\end{array} 
 \right]_q \, 
q^{(\ell-n)} \, 
P_{1 \cdots \ell}^{(\ell)} 
\left( \prod_{k=1}^{n} e_{k}^{22} \, \prod_{k=n+1}^{m} e_k^{21} \, 
\prod_{k=m+1}^{\ell} e_k^{11} \right) P_{1 \cdots \ell}^{(\ell)}  
\label{eq:m>n}
\ee    
and for $m <n $ 
\be  
 || \ell, m \rangle \langle \ell, n || 
= 
\left[
\begin{array}{c} 
\ell \\
n 
\end{array} 
 \right]_q \, 
q^{(\ell-n)} \, 
P_{1 \cdots \ell}^{(\ell)} 
\left( \prod_{k=1}^{m} e_{k}^{22} \, \prod_{k=m+1}^{n} e_k^{12} \, 
\prod_{k=n+1}^{\ell} e_k^{11} \right) P_{1 \cdots \ell}^{(\ell)}  
\label{eq:m<n}
\ee    
and for $m=n$ 
\be  
 || \ell, n \rangle \langle \ell, n || 
= 
\left[
\begin{array}{c} 
\ell \\
n 
\end{array} 
 \right]_q \, 
q^{(\ell-n)} \, 
P_{1 \cdots \ell}^{(\ell)} 
\left( \prod_{k=1}^{n} e_{k}^{22}  \, 
\prod_{k=n+1}^{\ell} e_k^{11} \right) 
P_{1 \cdots \ell}^{(\ell)}  
\label{eq:m=n}
\ee    
Let us denote by $E^{mn (\ell+)}$ the unit matrices 
acting on spin-$\ell$ representation $V^{(\ell)}$ 
for $m, n= 0, 1, \ldots, \ell$. 
We now define $E_i^{mn (\ell+)}$ by the 
unit matrices acting on the $i$th component 
of the tensor product $(V^{(\ell)})^{\otimes N_s}$.  
Explicitly, we have 
\be 
 E_{i}^{mn \, (\ell+)} = (I^{(\ell)})^{\otimes (i-1)}  \otimes E^{mn} 
 \otimes (I^{(\ell)})^{\otimes (N_s-i)}  
\ee
where $I^{(\ell)}$ denotes the $(\ell+1) \times (\ell+1)$ identity matrix.   
Then, we derive the following formulas 
from (\ref{eq:m>n}), (\ref{eq:m<n}) and 
(\ref{eq:m=n}), respectively.  
For $m > n$ we have 
\bea  
 E_{i}^{mn \, (\ell+)} 
& = & \left[
\begin{array}{c} 
\ell \\
m 
\end{array} 
 \right]_q \, 
 q^{n(\ell-n)} \, 
P^{(\ell)}_{1 \cdots L} 
%P^{(\ell)}_{(i-1)\ell+1} 
\,  
\prod_{\alpha=1}^{(i-1)\ell} (A+D)(\xi_{\alpha}) 
\prod_{k=1}^{n} D(\xi_{(i-1)\ell+k} ) 
\prod_{k=n+1}^{m} B(\xi_{(i-1)\ell+k}) 
\non \\ 
& & \quad \times \, \prod_{k=m+1}^{\ell} A(\xi_{(i-1)\ell+k})  
\prod_{\alpha=i \ell +1}^{\ell N_s} (A+D)(\xi_{\alpha}) \, \,   
P^{(\ell)}_{1 \cdots L}
\, .  \label{eq:Em>n}
\eea
For $m < n$ we have 
\bea  
 E_{i}^{mn \, (\ell+)} 
& = & \left[
\begin{array}{c} 
\ell \\
m 
\end{array} 
 \right]_q \, 
 q^{n(\ell-n)} \, 
P^{(\ell)}_{1 \cdots L} 
%P^{(\ell)}_{(i-1)\ell+1} 
\,  
\prod_{\alpha=1}^{(i-1)\ell} (A+D)(\xi_{\alpha}) 
\prod_{k=1}^{m} D(\xi_{(i-1)\ell+k} ) 
\prod_{k=m+1}^{n} C(\xi_{(i-1)\ell+k}) 
\non \\ 
& & \quad \times \, \prod_{k=n+1}^{\ell} A(\xi_{(i-1)\ell+k})  
\prod_{\alpha=i \ell +1}^{\ell N_s} (A+D)(\xi_{\alpha}) \, \,   
P^{(\ell)}_{1 \cdots L}
\, .  \label{eq:Em<n}
\eea
For $m = n$ we have 
\bea  
 E_{i}^{nn \, (\ell+)} 
& = & \left[
\begin{array}{c} 
\ell \\
n 
\end{array} 
 \right]_q \, 
 q^{n(\ell-n)} \, 
P^{(\ell)}_{1 \cdots L} 
%P^{(\ell)}_{(i-1)\ell+1} 
\,  
\prod_{\alpha=1}^{(i-1)\ell} (A+D)(\xi_{\alpha}) 
\prod_{k=1}^{n} D(\xi_{(i-1)\ell+k} ) 
\non \\ 
& & \quad \times \, \prod_{k=n+1}^{\ell} A(\xi_{(i-1)\ell+k})  
\prod_{\alpha=i \ell +1}^{\ell N_s} (A+D)(\xi_{\alpha}) \, \,   
P^{(\ell)}_{1 \cdots L}
\, .  \label{eq:Em=n}
\eea

Let us now discuss the derivation of formula (\ref{eq:ellmn}). 
It is easy to show the following:   
\be 
\sigma_{i_1}^{-} \cdots  \sigma_{i_m}^{-} | 0 \rangle  
\langle 0 | \sigma_{j_1}^{+} \cdots  \sigma_{j_n}^{+} 
=
 e_{i_1}^{21} \cdots e_{i_n}^{21} 
\prod_{p=1; p \ne i_k, j_q}^{\ell} e_{p}^{11} 
e_{i_1}^{12} \cdots e_{i_n}^{12}  
\ee
Then, making use of expressions 
(\ref{eq:|ell,n>}) and (\ref{eq:<ell,n|}), 
we obtain (\ref{eq:ellmn}).

\subsection{Form factors for higher-spin operators }

Making use  of the fundamental lemma of the quantum inverse-scattering 
problem, lemma \ref{lem:QISP}, together with the useful formulas 
given in \S 7.3 such as 
(\ref{eq:PB}) and (\ref{eq:PC}), and (\ref{eq:Em>n}), (\ref{eq:Em<n}) 
and (\ref{eq:Em=n}), 
we can systematically calculate form factors for the higher-spin cases.   
Here we note that the form factors associated with generators 
$S^{\pm}$ of the spin $SU(2)$ 
have been derived  for the higher-spin XXX chains \cite{Castro-Alvaredo}. 
They are derived through the relations  
corresponding to (\ref{eq:X-(ell+)}) and (\ref{eq:X+(ell+)}) 
in the limit of $q=1$.

For an illustration, let us calculate the following form factor: 
\be 
F_n^{- (\ell+)}(i; \, \{ \mu_j \}, \{ \lambda_k \} ) 
= \langle 0 | \prod_{j=1}^{n+1} C^{(\ell+)}(\mu_j) \cdot 
X_i^{- (\ell+)} \prod_{k=1}^{n} B^{(\ell+)}(\lambda_k) | 0 \rangle \, ,  
\label{eq:F-}
\ee
where $\{\mu_j \}$ and $\{ \lambda_k \}$ are solutions of the Bethe 
ansatz equations. 
Putting (\ref{eq:PB}) into (\ref{eq:F-}) and making use of the fact 
that projector $P_{1 \cdots L}^{(\ell)}$ commutes with 
the matrix elements of $R_{0, 1 \cdots L}^{+}$,  
we have  
\bea
&  & \langle 0 | \prod_{j=1}^{n+1} C^{+}(\mu_j) \cdot 
X_i^{- (\ell+)} \prod_{k=1}^{n} B^{+}(\lambda_k) | 0 \rangle
\non \\ 
& = & [\ell]_q \, \langle 0 | \prod_{j=1}^{n+1} C^{+}(\mu_j) \cdot 
P_{(i-1)\ell+1}^{(\ell)} 
\prod_{\alpha=1}^{(i-1)\ell} (A^{+} +D^{+})(\xi_{\alpha}) \cdot   
B^{+}(\xi_{(i-1)\ell +1}) \cdot 
\non \\ 
& & \quad \times 
\prod_{\alpha=(i-1)\ell+2}^{\ell N_s} (A^{+} +D^{+})(\xi_{\alpha}) 
P_{(i-1)\ell+1}^{(\ell)}
 \prod_{k=1}^{n} B^{+}(\lambda_k) | 0 \rangle  
\non \\ 
& =  & [\ell]_q \, e^{\sum_{j=1}^{n+1} \mu_j 
- \sum_{k=1}^{n} \lambda_k - \xi_{(i-1)\ell+1} } 
{\frac {\phi_{(i-1)\ell}(\{ \mu_j \}_{n+1})} 
       {\phi_{(i-1)\ell+1}(\{ \lambda \}_n)}}  
\langle 0 | \prod_{j=1}^{n+1} C(\mu_j) \cdot 
B(\xi_{(i-1)\ell +1}) \prod_{k=1}^{n} B(\lambda_k) | 0 \rangle \non \\ 
& =  & [\ell]_q \, e^{\sum_{j=1}^{n+1} \mu_j 
- \sum_{k=1}^{n} \lambda_k - \xi_{(i-1)\ell+1} } 
{\frac {\phi_{(i-1)\ell}(\{ \mu_j \}_{n+1})} 
       {\phi_{(i-1)\ell+1}(\{ \lambda \}_n)}}  \non \\
& & \quad  \times 
S_{n+1}\left(\{ \mu_j \}_{n+1}, 
\{ \xi_{(i-1)\ell+1}, \lambda_1, \ldots, \lambda_n \}; \{ \xi_k^{(\ell)}
 \} \right)    
\eea

Let us define the form factor in the symmetric case as follows.  
\be 
F_n^{- (\ell)}(i; \, \{ \mu_j \}, \{ \lambda_k \} ) 
= \langle 0 | \prod_{j=1}^{n+1} C(\mu_j) \cdot 
\left( X_i^{- (\ell+)} \right)^{\bar \chi} 
\prod_{k=1}^{n} B(\lambda_k) | 0 \rangle  
\ee
Here we recall that $\{ \mu_j \}$ and $\{ \lambda_k \}$ are 
solutions of the Bethe ansatz equations.  
Then, we obtain the following expression. 
\bea 
& & F_n^{- (\ell)}(i; \, \{ \mu_j \}, \{ \lambda_k \} ) 
 =  [\ell]_q \, {\frac {\phi_{(i-1)\ell}(\{ \mu_j \}_{n+1})} 
      {\phi_{(i-1)\ell+1}(\{ \lambda \}_n)}}  \non \\
& & \qquad  \times 
S_{n+1}\left(\{ \mu_j \}_{n+1}, 
\{ \xi_{(i-1)\ell+1}, \lambda_1, \ldots, \lambda_n \}; \{ \xi_k^{(\ell)}
 \} \right)   \, .  
\eea

Let us define the form factor 
for $K$ by    
\be 
F^{K (\ell+)}_n(i,  \, \{ \mu_j \}, \{ \lambda_k \} )
= \langle 0 | \prod_{j=1}^{n} C^{(\ell+)} (\mu_j) \cdot 
\left( K_i^{(\ell+)} \right)^{\bar \chi} 
\prod_{k=1}^{n} B^{(\ell+)}(\lambda_k) | 0 \rangle \, .  
\ee
Here we recall that $\{ \mu_j \}$ and $\{ \lambda_k \}$ are 
solutions of the Bethe ansatz equations. Similarly, we 
can calculate as follows. 
\bea 
& & F^{K (\ell+)}_n(i,  \, \{ \mu_j \}, \{ \lambda_k \} ) 
 =  \exp\left({\sum_{j=1}^{n} \mu_j - \sum_{k=1}^{n} \lambda_k }\right)  
{\frac {\phi_{(i-1)\ell}(\{ \mu_j \}_{n})} 
       {\phi_{i\ell}(\{ \lambda \}_n)}}  \non \\
& & \quad  \times 
\sum_{n=0}^{\ell} 
\left[
\begin{array}{c} 
\ell \\ 
n 
\end{array}
\right]_q \, q^{\ell - 2n + n(\ell- n)}  
\langle 0 | \prod_{j=1}^{n} C(\mu_j) 
\prod_{k=1}^{n} D(\xi_{(i-1)\ell+k}) \prod_{k=n+1}^{\ell} 
A(\xi_{(i-1)\ell+k})  
 \prod_{k=1}^{n} B(\lambda_k) | 0 \rangle \non \\ 
\eea
Through the commutation relation between $A$ and $B$ and 
that between $C$ and $D$, the vacuum expectation 
of the product of $C$, $D$, $A$ and $B$ operators 
can be expressed in terms of the sum of scalar products.

\appendix
%%%%%%%%%%%%%%%%%%%%%%%%%%%%%%%%%%%%%%%%%%%%%%%%%%%%%%
%
%             APPENDIX 
%
%%%%%%%%%%%%%%%%%%%%%%%%%%%%%%%%%%%%%%%%%%%%%%%%%%%%%%%
 \setcounter{equation}{0} 
 \renewcommand{\theequation}{A.\arabic{equation}}
\setcounter{df}{0} 

\section{Derivation of symmetry relations for $R_p^{\sigma}$ }

\begin{lem2} Let $p$ be a sequence of $n$ integers, $1, 2, \ldots, n$.  
For any $\sigma_A, \sigma_B \in {\cal S}_n$  we have   
\be 
(\sigma_A \sigma_B) \, p = \sigma_B (\sigma_A p) \, . 
\ee
\label{lem:ABp}
\end{lem2}
\begin{proof} 
Let us denote $p_{\sigma_A i}$ by $q_i$ for $i=1, 2, \ldots, n$. 
%If we denote $\sigma_B k$ for an integer $k$ ($1 \le k \le n$) 
%by $j$, then we have 
%$p_{\sigma_A j}= p_{\sigma_A (\sigma_B k)} = p_{\sigma_A \sigma_B \, k}$.    
We thus have 
\bea
\sigma_B (q_1, \ldots, q_n) & = & (q_{\sigma_B 1}, \ldots, q_{\sigma_B n}) 
\non \\ 
& = & (p_{\sigma_A (\sigma_B 1)}, \ldots, p_{\sigma_A (\sigma_B n)}) 
\non \\ 
& = & (p_{\sigma_A \sigma_B \, 1}, \ldots, p_{\sigma_A \sigma_B \, n}) 
\non \\
& = & (\sigma_A \sigma_B) \, p 
\eea
\end{proof}

\begin{prop2}
Definition \ref{df:Rp} is well defined. That is,   
for $R_{j,k}=R_{j,k}(\lambda_j, \lambda_k)$ the following relations hold:  
\bea 
 R_p^{(\sigma_A \sigma_B) \sigma_C} & = & 
R_p^{\sigma_A (\sigma_B \sigma_C)} 
\label{eq:cond1} \\ 
R_p^{s_j^2} & = & R_p^{e} =1 
\quad {\mbox for} \, j=1, 2, \ldots, n \, ,  \label{eq:cond2} \\ 
R_p^{s_i s_{i+1} s_i} & = & R_p^{s_{i+1} s_{i} s_{i+1}} 
\quad {\mbox for} \, i=1, 2, \ldots, n-1 \, . \label{eq:cond3} 
\eea
\end{prop2} 
\begin{proof} 
We recall that in terms of generators $s_j$  
the defining relations of the symmetric group ${\cal S}_n$ are given by 
$s_i s_{i+1} s_i = s_{i+1} s_i s_{i+1}$ for $i=1, 2, \ldots, n-1$ and  
$s_j^2 = 1$ for $j=1, 2, \ldots, n$ \cite{Magnus}. 
It thus follows that definition \ref{df:Rp} is well defined if and only if 
conditions 
(\ref{eq:cond1}), (\ref{eq:cond2}) and (\ref{eq:cond3}) hold. 
We now show them. Lemma \ref{lem:ABp} leads to conditions (\ref{eq:cond1}).  
Conditions (\ref{eq:cond2}) are derived from the inversion relations 
(unitarity conditions) (\ref{eq:inv-rel}). Finally we show 
conditions  (\ref{eq:cond3}) by the Yang-Baxter equations.  
\end{proof}

\begin{lem2} Let  $p$ be a  sequence of integers $1, 2, \ldots, n$, 
and $R_{j,k}$ denote $R_{j,k}(\lambda_j, \lambda_k)$ for 
$j, k = 1, 2, \ldots, n$. 
For  $\sigma_c=(1 2 \cdots n)$ we have 
\be 
R^{\sigma_c}_{p} = R_{p_1, p_2 \cdots p_n} \, . 
\ee
\label{lem:p-cyclic}
\end{lem2} 
\begin{proof}
Noting  
$(1 2 \cdots n) = (1 2)(2 3) \cdots (n-1 \, n) 
= s_1 s_2 \cdots s_{n-1}$. 
we have 
\bea 
R_p^{\sigma_c} & = & R_p^{s_1 s_2 \cdots s_{n-1}} 
=R_{s_1 p}^{s_2 \cdots s_{n-1}} R_p^{s_1} \non  \\ 
& = & R_{s_2 (s_1 p)}^{s_3 \cdots s_{n-1}} R_{s_1 p}^{s_2}  R_p^{s_1}
=R_{(s_1 s_2) p}^{s_3 \cdots s_{n-1}} R_{s_1 p}^{s_2}  R_p^{s_1} \non \\
& =& R_{(s_1 \cdots s_{n-2}) p}^{s_{n-1}} 
\cdots R_{(s_1 s_2) p}^{s_3} R_{s_1 p}^{s_2}  R_p^{s_1} \non \\ 
& =& R_{(1 2 \cdots n-1) p}^{s_{n-1}} 
\cdots R_{(123) p}^{s_3} R_{(12) p}^{s_2}  R_p^{s_1} \non \\ 
& = & R_{p_1, p_n} \cdots R_{p_1, p_3} R_{p_1, p_2} \non \\ 
& = & R_{p_1, p_2 \cdots p_n} \, .  
\eea
\end{proof} 

\begin{prop2} 
Let ${\cal A}$ be a Hopf algebra and ${\cal R}$ is an element of 
${\cal A} \otimes {\cal A}$ such that 
${\cal R} \Delta(x) = \tau \circ \Delta(x) {\cal R}$ for all 
$x \in {\cal A}$. Suppose that ${\cal R}$ is given by 
${\cal R} = \sum_{a} r^{(a,1)} \otimes r^{(a,2)}$, where 
 $r^{(a,1)}$,   $r^{(a,2)} \in {\cal A}$. We define $R_{j,k}$ by 
\be 
R_{j, k} = \sum_{a} id_1 \otimes \cdots \otimes  r_j^{(a,1)} 
\otimes \cdots \otimes r_k^{(a,2)} \otimes \cdots \otimes id_n 
\quad \in {\cal A}^{\otimes n} 
\, .  
\ee
If $R_{j,k}$ satisfy the inversion relations and the Yang-Baxter equations: 
\bea 
R_{12}R_{21} & = & id  \, , \non \\
R_{12}R_{13}R_{23} & = & R_{23} R_{13} R_{12} \, ,  
\eea
then we have the following symmetry relations: 
\be 
R_{p_q}^{\sigma} \Delta^{(n-1)}(x) = \sigma \circ \Delta^{(n-1)}(x) 
R_{p_q}^{\sigma} \quad x \in {\cal A} \, . 
\label{eq:symA}
\ee
Here $p_q$ denotes $p_q=(1, 2, \ldots, n)$, 
a sequence of $n$ integers, and we have defined $R_{p_q}^{\sigma}$ 
as in definition \ref{df:Rp}. 
\label{prop:A}
\end{prop2}
\begin{proof}
Recall that any given permutation $\sigma$ is expressed as a product of 
generators $s_j=(j \, j+1)$. We show symmetry relations (\ref{eq:symA}) 
by induction on the number of generators $s_j$
whose product gives  permutation $\sigma$. 
Suppose that (\ref{eq:symA}) holds for 
$\sigma=\sigma_A$. We now show that (\ref{eq:symA}) holds for 
$\sigma = \sigma_A s_j$. Making use of eq. (\ref{eq:Rp}) 
in definition \ref{df:Rp} we have 
\be 
R_{p_q}^{\sigma_A s_j} \Delta^{(n-1)}(x) 
= R_{\sigma_A (p_q)}^{s_j} \, R_{p_q}^{\sigma_A} \Delta^{(n-1)}(x) 
= R_{\sigma_A (p_q)}^{s_j} \, \sigma_A \circ \Delta^{(n-1)}(x) \, 
R_{p_q}^{\sigma_A}
\ee
Expressing $\Delta^{(n-1)}(x)$ as  
$\Delta^{(n-1)}(x) = \sum x^{(1)} \otimes x^{(2)} \otimes 
\cdots \otimes x^{(n)}$,  
we have 
\be 
\sigma_A \circ \Delta^{(n-1)}(x) = 
\sum x^{(\sigma_A^{-1} 1)} \otimes x^{(\sigma_A^{-1} 2)} \otimes 
\cdots \otimes x^{(\sigma_A^{-1} n)}
\ee
It now follows from definition \ref{df:Rp} that we have 
$R_{\sigma_A(p_q)}^{s_j} 
%= R_{q_{\sigma_A j}, \, q_{\sigma_A (j+1)}} 
=  R_{\sigma_A j, \, \sigma_A (j+1)}$. 
Let us denote $\sigma_A j$ and $\sigma_A (j+1)$ by $a$ and $b$, 
respectively. Then we have $\sigma_A^{-1} a = j$ and 
 $\sigma_A^{-1} b = j+1$.  Assuming $a < b$ we have 
\bea 
& & R_{\sigma_A j, \, \sigma_A (j+1)} \, \, 
\sigma_A \circ \Delta^{(n-1)}(x)  
\non \\
& = &R_{a, b} \, 
\sum x_1^{(\sigma_A^{-1} 1)} \otimes x_2^{(\sigma_A^{-1} 2)} \otimes 
\cdots \otimes x^{(\sigma_A^{-1} a)}_a \otimes 
\cdots \otimes x^{(\sigma_A^{-1} b)}_b \otimes 
\cdots \otimes x^{(\sigma_A^{-1} n)}_n \non \\ 
& = & \sum x^{(\sigma_A^{-1} 1)}_1\otimes x^{(\sigma_A^{-1} 2)}_2 \otimes 
\cdots \otimes x^{(\sigma_A^{-1} b)}_a \otimes 
\cdots \otimes x^{(\sigma_A^{-1} a)}_b \otimes 
\cdots \otimes x^{(\sigma_A^{-1} n)}_n \, \, 
R_{a, b} \non \\ 
& = & (\sigma_A s_j) \circ \Delta^{(n-1)}(x) \, \, 
R_{\sigma_A j, \, \sigma_A (j+1)}
\eea
Here we note that  
\be 
R_{a, b} \,  x^{(\sigma_A^{-1} a)}_a \otimes x^{(\sigma_A^{-1} b)}_b 
= x^{(\sigma_A^{-1} b)}_a \otimes x^{(\sigma_A^{-1} a)}_b \, R_{a, b} 
\ee
since we have 
$$
x^{(\sigma_A^{-1} a)} \otimes x^{(\sigma_A^{-1} b)}
= x^{(j)} \otimes x^{(j+1)} = \Delta({\bar x}^{(j)}) 
$$
where ${\bar x}^{(j)}$ are defined by 
$$
\Delta^{(n-2)}(x) = \sum {\bar x}^{(1)} \otimes \cdots \otimes 
{\bar x}^{(n-1)}  
$$
We therefore have 
\bea
& & R_{\sigma_A (p_q)}^{s_j} \, \sigma_A \circ \Delta^{(n-1)}(x) \, 
R_{p_q}^{\sigma_A} 
= (\sigma_A s_j) \circ \Delta^{(n-1)}(x) \, 
R_{\sigma_A (p_q)}^{s_j} \, R_{p_q}^{\sigma_A} \non \\ 
& & = (\sigma_A s_j) \circ \Delta^{(n-1)}(x) \, 
R_{p_q}^{\sigma_A s_j} 
\eea
\end{proof}

Symmetry relations similar to (\ref{eq:symA}) hold 
for products of monodromy matrices. 
Let us consider $m$ auxiliary spaces with suffices 
$a(1), a(2), \ldots , a(m)$, respectively. 
We denote the monodromy matrix 
$T_{a(j)}(\lambda_{a(j)}; \xi_1, \ldots, \xi_L)$ 
 simply by $T_{a(j)}$. We denote by $\Delta^{(m-1)}(T)$ 
  the following operator: 
\be 
\Delta^{(m-1)}(T) = T_{a(1)} T_{a(2)} \cdots T_{a(m)}
\ee 
Let $\sigma$ an element of ${\cal S}_m$. 
We define the action of $\sigma$ on $\Delta^{(m-1)}(T)$ 
 by the following:   
\be 
\sigma \circ \Delta^{(m-1)}(T) = 
T_{a({\bar \sigma} 1)} T_{a({\bar \sigma} 2)} \cdots 
T_{a({\bar \sigma} m)}
\ee
Here ${\bar \sigma}$ denotes the inverse of $\sigma$: 
${\bar \sigma}= \sigma^{-1}$. Then we have the following. 

\begin{prop2}  Let $p_q$ be $p_q=(1, 2, \ldots, m)$. 
For any $\sigma \in {\cal S}_m$ we have 
\be 
R^{\sigma}_{p_q} \Delta^{(m-1)}(T) = \sigma \circ 
\Delta^{(m-1)}(T) R^{\sigma}_{p_q}  
\ee
\end{prop2}

%%%%%%%%%%%%%%%%%%%%%%%%%%%%%%%%%%%%%%%%%%%%%%%%%%%%%%%%%%
%
%
%
 \setcounter{equation}{0} 
 \renewcommand{\theequation}{B.\arabic{equation}}
\setcounter{df}{0} 

\section{Symmetric-group action on products of $R$-matrices and 
 the $F$-basis}

\begin{lem2}(\cite{MS2000})
(i) Cocycle conditions hold for $n \le L$. 
\be 
R_{2 \cdots n-1, \, n} R_{1, \, 2 \cdots n} = 
R_{1, \, 2 \cdots n-1} R_{1 2 \cdots n-1, \, n} 
\label{eq:cocyc-R}
\ee
(ii) The unitarity relations hold for $n \le L$. 
\be 
R_{1, \, 2 \cdots n} R_{2 3 \cdots n, \, 1} = I^{\otimes L}  
\ee
\end{lem2}
\begin{proof} 
Cocycle conditions (\ref{eq:cocyc-R}) are derived 
from the Yang-Baxter equations. 
\end{proof}

Let us denote by the symbol $(p_0, p_1, p_2, \ldots, p_n)$ 
a sequence of $n+1$ integers, $0, 1, 2, \ldots, n$, 
and we express it as $(p_0, p)$ where 
$p$ denotes the subsequence $(p_1, p_2, \ldots, p_n)$.  
%Let $\sigma$ be an element of permutation group ${\cal S}_{n+1}$ of 
%$n+1$ integers, $0, 1, \ldots, n$. 
%, i.e.  $\sigma$ maps $j$ to $\sigma(j)$ for $j=0, 1, 2, \ldots, n$. 
%We define the action of $\sigma$ on $(p_0, p)$ by  
%\be 
%(\sigma(p_0), \sigma(p)) = 
%(p_{\sigma(0)}, \,  p_{\sigma(1)}, p_{\sigma(2)}, \ldots, p_{\sigma(n)}) 
%\ee

%\begin{df2} We define $R_{\sigma(p_0), \sigma(p)}$  
%and $R_{\sigma(p), \sigma(p_0)}$ as follows. 
%\bea R_{\sigma(p_0) , \, \sigma(p)} & = &  
%R_{p_{\sigma(0)} p_{\sigma(n)}} R_{p_{\sigma(0)} p_{\sigma(n-1)}} \cdots 
%R_{p_{\sigma(0)} p_{\sigma(2)}} \non \\
%R_{\sigma(p), \, \sigma(p_0)} & = &  
%R_{p_{\sigma(0)}, \, p_{\sigma(n)}} R_{p_{\sigma(2)}, \, p_{\sigma(n)}} 
%\cdots R_{p_{\sigma(n-1)}, \,  p_{\sigma(n)}} \eea \end{df2} 
%

\begin{lem2} Let $p$ be a sequence of $n$ integers, $1, 2, \ldots, n$.  
For $s_j = (j \, j+1) \in {\cal S}_n$ we have 
\bea 
R_{p}^{s_j} R_{0, \, p} & = & R_{0, \, s_j(p)} R_{p}^{s_j} 
\non \\
R_{p}^{s_j} R_{p, \, 0} & = & R_{s_j(p), \, 0} R_{p}^{s_j} 
\label{eq:sjR}
\eea
\end{lem2} 
\begin{proof} 
We first note $R_p^{s_j}= R_{p_j, p_{j+1}}$. We have 
\bea 
R^{s_j}_p R_{0, p} & = & R_{p_j, p_{j+1}} \, 
R_{0, p_n} \cdots R_{0, p_1} \non \\  
&= & R_{0, p_n} \cdots R_{0, p_{j+2}} 
\cdot R_{p_j, p_{j+1}} R_{0, p_{j+1}} R_{0, p_{j}} \cdot 
R_{0, p_{j+2}} \cdots R_{0, p_1} 
\eea
Applying the Yang-Baxter equations 
$R_{p_j, p_{j+1}} R_{0, p_{j+1}} R_{0, p_{j}} 
= R_{0, p_{j}} R_{0, p_{j+1}}  R_{p_j, p_{j+1}}$ we  have   
\bea 
R^{s_j}_p R_{0, p} & = & 
R_{0, p_n} \cdots R_{0, p_{j+2}} 
\cdot  R_{0, p_{j}} R_{0, p_{j+1}} \cdot 
R_{0, p_{j-1}} \cdots R_{0, p_1} \cdot R_{p_j, p_{j+1}} \non \\ 
& = & R_{0, s_j(p) } R_{p_j, p_{j+1}} \, . 
\eea
\end{proof}

\begin{prop2} Let $p$ be a sequence of $n$ integers, $1, 2, \ldots, n$, 
and $\sigma$ a permutation of the $n$ integers.  We have 
\bea 
R_{p}^{\sigma} R_{0, \, p} & = & R_{0, \, \sigma(p)} R_{p}^{\sigma} 
\non \\
R_{p}^{\sigma} R_{p, \, 0} & = & R_{\sigma(p), \, 0} R_{p}^{\sigma} 
\eea
\end{prop2} 
\begin{proof} 
Expressing permutation $\sigma$ as a product of generators $s_j$, 
and applying (\ref{eq:sjR}) many times, we can show the symmetry relations.  
\end{proof}

We define the action of $\sigma$ on the $F$-basis as follows.   
\be 
F_{\sigma(p)} = F_{p_{\sigma(1)} p_{\sigma(2)} \cdots p_{\sigma(n)}} 
\ee
\begin{prop2}
Let $p$ be a sequence of integers, $1, 2, \ldots, n$. 
For $\sigma \in {\cal S}_n$ we have  
\be 
R^{\sigma}_{p} F_{0, p} = F_{0, \sigma(p)} R_{p}^{\sigma} 
\ee
We also have 
\be 
F_p = F_{\sigma(p)} R_{p}^{\sigma} \label{eq:inv-F}
\ee
\end{prop2} 
\begin{proof} 
Expressing the $F$-basis in terms of the $R$-matrices we have 
\bea 
R^{\sigma}_{p} F_{0, p} 
& = & R^{\sigma}_{p} \left( e_0^{11} + e_0^{22} R^{\sigma}_{0,p}  \right) 
\non \\
& = &  e_0^{11} R^{\sigma}_{p} 
+ e_0^{22} R^{\sigma}_{p} R^{\sigma}_{0,p} \non \\  
& = &  e_0^{11} R^{\sigma}_{p} 
+ e_0^{22}  R^{\sigma}_{0,\sigma(p)} R^{\sigma}_{p} \non \\  
& = &  \left( e_0^{11} 
+ e_0^{22}  R^{\sigma}_{0, \sigma(p)} \right)  R^{\sigma}_{p} \non \\  
& = & F_{0, \sigma(p)} R^{\sigma}_{p}
\eea
We first show (\ref{eq:inv-F}) with $\sigma=s_j$ for 
$j=1, 2, \ldots, n-1$, and then 
we derive it for all permutations $\sigma$.   
\end{proof} 

\begin{lem2} The propagator 
$F^{-1}_{i \cdots L \, 1 \cdots i-1} F_{1 \cdots L}$ is given by 
the following: 
\be 
F^{-1}_{i \cdots L \, 1 \cdots i-1} F_{1 \cdots L} 
= \prod_{\alpha=1}^{i-1} \left( A_{1 \cdots L}(\xi_{\alpha}) +  
D_{1 \cdots L}(\xi_{\alpha}) \right)  
\ee
\end{lem2}
\begin{proof} 
Let $\sigma_c$ be a cyclic permutation: $\sigma_c=( 1 2 \cdots L)$, 
and $p_q$ the sequence $p_q=(1, 2, \ldots, n)$.   
We have 
\be 
F_{i \cdots L \, 1 \cdots i-1}  =  F_{\sigma_c^{i-1}(p_q)} 
= F_{1 \cdots L} R_{p_q}^{\sigma_c^{i-1}}  
\ee
and hence we have  
\bea 
 F^{-1}_{i \cdots L \, 1 \cdots i-1} F_{1 \cdots L} & = & 
\left(F_{1 \cdots L} 
R_{p_q}^{\sigma_c^{i-1}} \right)^{-1} 
F_{1 \cdots L}  \non \\ 
& = & R_{p_q}^{\sigma_c^{i-1}} \, 
F_{1 \cdots L}^{-1} F_{1 \cdots L} \non \\ 
&= & R_{p_q}^{\sigma_c^{i-1}} 
\eea
We thus obtain  
\bea 
R_{p_q}^{\sigma_c^{i-1}} & = & R_{\sigma(p_q)}^{\sigma_c^{i-2}} 
R_{p_q}^{\sigma_c}  \non \\
& = & R_{\sigma_c^{i-2} (p_q)}^{\sigma_c} \cdots R_{p_q}^{\sigma_c} \non \\ 
& = & R_{i-1 \cdots L \, 1 \cdots i-2}^{\sigma_c} 
\cdots R_{2 \cdots L 1}^{\sigma_c} R_{1 2 \cdots L}^{\sigma_c} \non \\
& = & R_{i-1, i \cdots L \, 1 \cdots i-2} 
\cdots R_{2, 3 \cdots L 1} R_{1, 2 \cdots L}  \non \\
& = & \prod_{\alpha=1}^{i-1} \left( A_{1 \cdots L}(\xi_{\alpha}) +  
D_{1 \cdots L}(\xi_{\alpha}) \right)  \non 
\eea
\end{proof} 

%\par \noindent 
%{\it Proof of proposition}

%%%%%%%%%%%%%%%%%%%%%%%%%%%%%%%%%%%%%%%%%%%
%
 \setcounter{equation}{0} 
 \renewcommand{\theequation}{C.\arabic{equation}}
\setcounter{df}{0} 

\section{Formulas of the $q$-analogue }
\begin{lem2} For two integers $\ell$ and $n$ satisfying 
$0 \le n \le \ell$ we have  
\be
\sum_{1 \le i_1 < \cdots < i_n \le \ell} q^{2 i_1 + \cdots + 2 i_n} 
= 
q^{n (\ell+1)} \, 
\left[
\begin{array}{c} 
\ell \\ 
n 
\end{array} 
\right]_{q} \label{eq:q-sum}
\ee
\end{lem2} 
\begin{proof} 
We can show by induction on $\ell$ the $q$-binomial expansion as follows.  
\be 
\prod_{k=0}^{\ell-1} \left( 1 - z q^{2k} \right)
= \sum_{n=0}^{\ell} (-z)^n q^{n(\ell-1)} 
\left[ 
\begin{array}{c} 
\ell \\ 
n 
\end{array} 
\right]_q \label{eq:q-binomial}
\ee
It is now easy to show the following: 
\be 
\prod_{k=1}^{\ell} \left( 1 - z q^{2k} \right)
= \sum_{n=0}^{\ell} (-z)^n  
\sum_{1 \le i_1 < \cdots < i_n \le \ell}  q^{2 i_1 + \cdots + 2 i_n} 
\label{eq:exp}
\ee
Comparing (\ref{eq:exp}) with (\ref{eq:q-binomial})
we obtain formula (\ref{eq:q-sum}). 
\end{proof} 

\begin{lem2} The spin-$\ell$ matrix representations  $X^{\pm (\ell +)}$ 
of the generators $X^{\pm}$ of $U_q(sl_2)$ 
are related to $\sigma_{\ell}^{+}$ and $\sigma_{1}^{-}$, respectively,   
as follows. 
\bea 
P^{(\ell)}_{1 \cdots \ell} \, \sigma_{1}^{-} \, P^{(\ell)}_{1 \cdots \ell} 
& = & {\frac 1 {[\ell]_q}} \, X^{- (\ell+)} \label{eq:X+A} \\ 
P^{(\ell)}_{1 \cdots \ell} \, \sigma_{\ell}^{+} \, P^{(\ell)}_{1 \cdots \ell} 
& = & {\frac 1 {[\ell]_q}} \, X^{+(\ell+)}  \label{eq:X-A}
\eea
\end{lem2} 
\begin{proof} 
Expressing the projection operator $P^{(\ell)}_{1 \cdots \ell}$ as 
$$
P^{(\ell)}_{1 \cdots \ell} = \sum_{n=0}^{\ell} || \ell, n \rangle 
\langle \ell, n || 
$$
we have 
\be 
P^{(\ell)}_{1 \cdots \ell} \sigma_1^{-} P^{(\ell)}_{1 \cdots \ell}
= \sum_{n=0}^{\ell-1} || \ell, n+1 \rangle 
\langle \ell, n ||  \quad \langle \ell, n+1 || \sigma_1^{-} || \ell, n \rangle 
\ee
Making use of (\ref{eq:q-sum}) we have 
\be
\langle \ell, n+1 || \sigma_1^{-} || \ell, n \rangle = 
\frac {[n+1]}{[\ell]} , 
\ee
and then we obtain (\ref{eq:X+A}). Similarly, we can show 
(\ref{eq:X-A}). 
\end{proof}

%%%%%%%%%%%%%%%%%%%%%%%%%%%%%%%%%%%%%%%%%%%
%
 \setcounter{equation}{0} 
 \renewcommand{\theequation}{D.\arabic{equation}}
\setcounter{df}{0} 

\section{Formulas of the $F$-basis useful for the diagonalization. }

Let us review some points of the diagonalization process 
of the $A$ and $D$ operators \cite{MS2000}. 

\begin{lem2}
Operators $A$ and $D$ are upper- and lower-triangular matrices,  
respectively. Moreover, 
the eigenvalues of operators $A$ and $D$ are given by 
\bea 
{\rm diag} \left( D_{1 2 \cdots n}(\lambda_0) \right)
& = & 
\bigotimes_{i=1}^{n} 
\left( 
\begin{array}{cc}
b_{0 i} & 0 \\
0 & 1 
\end{array}
\right)_{[i]} \, , 
\label{eq:eigenD} \\ 
{\rm diag} \left( A_{1 2 \cdots n}(\lambda_0) \right)
& = & 
\bigotimes_{i=1}^{n} 
\left( 
\begin{array}{cc}
1 & 0 \\ 
0 & b_{0 i} 
\end{array}
\right)_{[i]} \, ,  
\label{eq:eigenA} 
\eea
where  $b_{0i}=b(\lambda_0 - \xi_i)$. 
\label{lem:eigenAD}
\end{lem2}
\begin{proof}
We can show it by induction on $n$. 
Noting $D_{1 2\cdots n } = C_n B_{1 \cdots n-1} + D_n D_{1 \cdots n-1}$, 
we show 
${\rm diag}(D_{1 \cdots n}= {\rm diag}(D_{1 \cdots n-1}(\lambda_0)  
\otimes {\rm diag}(b_{0n}, 1)_{[n]}$.  
We can show (\ref{eq:eigenA}) similarly.  
\end{proof} 

\begin{lem2} The partial $F_{0, 1 \cdots n}$ and ${\bar F}_{0, 1 \cdots n}$ 
are expressed as follows. 
\bea {F}_{0, 1 \cdots n} & = & \left( 
\begin{array}{cc} 
 I_{1 \cdots n} & 0  \\
C_{1 \cdots n}(\lambda_0) & D_{1 \cdots n}(\lambda_0) \end{array}
\right)_{[0]} \, , \\
{\bar F}_{0, 1 \cdots n} & = & \left( 
\begin{array}{cc} 
A_{1 \cdots n}(\lambda_0) & B_{1 \cdots n}(\lambda_0) \\
0 & I_{1 \cdots n}  \end{array}
\right)_{[0]} \, . 
\eea 
\end{lem2}
\begin{proof}
It is clear from definition \ref{df:def-F} of the $F$-basis. 
\end{proof}

\par \noindent 
{\bf Proof of propositions \ref{prop:diagDA} and \ref{prop:diagADdagger}}
\par \noindent 

For an illustration, we 
now derive the diagonalized form of the $D$ operator.  
 From $R_{0, 1 \cdots n} = R_{0, 2 \cdots n} R_{0, 1}$ we have    
\bea 
D_{1 2 \cdots n } & = & C_{2 \cdots n} B_1 + D_{2 \cdots n} D_1 
\non \\ 
& = & 
\left( 
\begin{array}{cc} 
 b_{01} D_{2 \cdots n}(\lambda_0) & 0  \\
 c_{01} C_{2 \cdots n}(\lambda_0) & D_{2 \cdots n}(\lambda_0) 
\end{array}
\right)_{[0]} \, . 
\eea
We thus calculate 
\bea 
F_{1 \cdots n} D_{1 \cdots n}(\lambda_0) 
& =& F_{2 \cdots n} \, F_{1, 2 \cdots n} \, 
D_{1 \cdots n}(\lambda_0) \non \\ 
& = & 
F_{2 \cdots n} \,  
\left( 
\begin{array}{cc} 
 I_{2 \cdots n} & 0  \\
C_{2 \cdots n}(\lambda_0) & D_{2 \cdots n}(\lambda_0) 
\end{array}
\right)_{[0]} 
\left( 
\begin{array}{cc} 
 b_{01} D_{2 \cdots n}(\lambda_0) & 0  \\
 c_{01} C_{2 \cdots n}(\lambda_0) & D_{2 \cdots n}(\lambda_0) 
\end{array}
\right)_{[0]} \non \\ 
& = & {\rm diag}(b_{01}, 1)_{[1]} 
\widetilde{D}_{2 \cdots n}(\lambda_0) F_{1 \cdots n}    
\eea 
Therefore, by induction we have the diagonalized form of 
operator $D$.     
Similarly, we can diagonalize $A$, $A^{\dagger}$ ]
and $D^{\dagger}$.

\begin{lem2} The diagonalized form of $F_{0, 1 \cdots n} 
{\bar F}^{\dagger}_{0, 1 \cdots n}$ is given by the following:   
\be 
F_{1 \cdots n} \left(F_{0, 1 \cdots n} 
{\bar F}^{\dagger}_{0, 1 \cdots n}  \right) 
F_{1 \cdots n}^{-1} = {\hat \delta}^{-1}_{0, 1 \cdots n} 
\label{eq:FFFF}
\ee
\end{lem2}
\begin{proof}
Making use of (\ref{eq:Rdag-inverse}) we show 
\bea  
& & F_{0, 1 \cdots n} {\bar F}_{0, 1 \cdots n}^{\dagger}  =  
F_{0, 1 \cdots n} {\cal C}_{0 1 \cdots n} 
{F}_{0, 1 \cdots n}^{\dagger} {\cal C}_{0 1 \cdots n} \non \\
& = &
\left( 
\begin{array}{cc} 
1 & 0 \\ 
C_{1 \cdots n} & D_{1 \cdots n}
\end{array}
\right)_{[0]}
\left( 
\begin{array}{cc} 
0 & {\cal C}_{1 \cdots n} \\ 
{\cal C}_{1 \cdots n} & 0
\end{array}
\right)_{[0]}
\left( 
\begin{array}{cc} 
1 & C^{\dagger}_{1 \cdots n} \\ 
0 & D_{1 \cdots n}
\end{array}
\right)_{[0]}
\left( 
\begin{array}{cc} 
0 & {\cal C}_{1 \cdots n} \\ 
{\cal C}_{1 \cdots n} & 0
\end{array}
\right)_{[0]}
\non \\
& = & 
\left( 
\begin{array}{cc} 
A^{\dagger} & 0 \\ 
0 & D_{1 \cdots n}
\end{array}
\right)_{[0]} \, . \non 
\eea
Calculating 
$F_{1 2 \cdots n} \cdot 
F_{0, 1 \cdots n} {\bar F}_{0, 1 \cdots n}^{\dagger} \cdot
 F^{-1}_{1 2 \cdots n}$ we have (\ref{eq:FFFF}). 
\end{proof} 

\begin{cor2} The inverse of the total $F$ is given as follows. 
\be 
F^{-1}_{1 \cdots n} = {\bar F}_{1 \cdots n}^{\dagger} 
{\hat \delta}_{1 \cdots n} 
\label{eq:Fd}
\ee
\end{cor2} 
\begin{proof} 
We show it by induction on $n$. 
Let us assume (\ref{eq:Fd}) for the case of $n$. 
We have 
$$
{\bar F}^{\dagger}_{0 1 \cdots n} \widehat{\delta}_{0 1 \cdots n} 
= {\bar F}^{\dagger}_{0 1 \cdots n} \left( 
{\bar F}^{\dagger}_{1 \cdots n}
\widehat{\delta}_{1 \cdots n} \right)
\widehat{\delta}_{0, 1 \cdots n}
= {\bar F}^{\dagger}_{0 1 \cdots n} 
{F}^{-1}_{1 \cdots n}
\widehat{\delta}_{0, 1 \cdots n}
$$
From (\ref{eq:FFFF}) we have 
${\bar F}^{\dagger}_{0 1 \cdots n} 
{F}^{-1}_{1 \cdots n}
\widehat{\delta}_{0, 1 \cdots n}= F_{0 1 \cdots n}^{-1}$, 
which corresponds to (\ref{eq:Fd}) for the case of $n+1$. 
\end{proof}

\begin{lem2} The dagger of total $F$ is given by the charge conjugation 
of transposed total $F$. 
\be 
F_{1 \cdots n}^{\dagger} = {\cal C}_{1 \cdots n}  
F^{t_1 \cdots t_n}_{n  \cdots 2 1} {\cal C}_{1 \cdots n}  
\label{eq:F-dagg}
\ee
Or equivalently we have 
\be 
{\bar F}_{1 \cdots n}^{\dagger} = F^{t_1 \cdots t_n}_{1 \cdots n}
\label{eq:F-bar-dag}
\ee
\end{lem2}
\begin{proof}
We show it by induction on $n$. Let us assume 
(\ref{eq:F-dagg}) for the case of $n-1$. 
We first show 
$$
{\bar F}_{1, 2 \cdots n}^{\dagger} 
= {\cal C}_{1 2 \cdots n} 
\left(e_1^{11} + e_1^{22}R_{1, 2 \cdots n} \right)^{\dagger}  
{\cal C}_{1 2 \cdots n} = F_{n \cdots 2, 1}^{t_1 \cdots t_n} \, . 
$$
Making use of the induction assumption we show 
\bea 
F_{1 2 \cdots n}^{\dagger} & = &
(F_{2 \cdots n} F_{1, 2 \cdots n})^{\dagger} 
= F_{1, 2 \cdots n}^{\dagger} F_{2 \cdots n}^{\dagger} \non \\ 
& = & {\cal C}_{1 \cdots n} 
F_{n \cdots 2, 1}^{t_1 \cdots t_n}{\cal C}_{1 \cdots n} \, 
{\cal C}_{1 \cdots n} F_{n \cdots 2}^{t_2 \cdots t_n} {\cal C}_{1 \cdots n} 
\non \\ 
& = &  {\cal C}_{1 \cdots n} F_{n \cdots 2 1 }^{t_1 \cdots t_n} 
 {\cal C}_{1 \cdots n} \, . \non 
\eea
\end{proof}

\begin{cor2} The inverse of total $F$ is given as follows. 
\be 
F_{1 \cdots n} F_{n \cdots 2 1}^{t_1 \cdots t_n} 
= \widehat{\delta}^{-1}_{1 \cdots n}
\label{eq:FF}
\ee
Equivalently, we have
\be 
F_{1 \cdots n} ^{-1} = 
F_{n \cdots 2 1}^{t_1 \cdots t_n} \, {\widehat \delta}_{1 \cdots n}
\label{eq:F^{-1}}
\ee
\end{cor2}
\begin{proof} 
It follows from (\ref{eq:Fd}) that 
${\bar F}^{\dagger} = F_{n \cdots 1}^{t_1 \cdots t_n} 
\widehat{\delta}_{1 2 \cdots n}$.  
From (\ref{eq:F-bar-dag}) we have (\ref{eq:FF}). 
\end{proof}

%%%%%%%%%%%%%%%%%%%%%%%%%%%%%%%%%%%%%%%%%%%

 \setcounter{equation}{0} 
 \renewcommand{\theequation}{E.\arabic{equation}}
\setcounter{df}{0} 
\section{Lemmas for Diagonalizing the $C$ operator }

\begin{lem2} 
Let $X^{+}$ be the generator of the quantum group $U_q(sl_2)$ and 
$X_i^{+}$ the matrix representation of $X^{+}$ 
acting on the $i$th site. We have 
\be 
{\widetilde \Delta}_{1 \cdots n}(X^{+})  
= \left(X_n^{+} + e_n^{11} 
 {\widetilde {A^{+}} }^{\dagger}_{1 \cdots n-1}(\xi_n)
{\widetilde \Delta}_{1 \cdots n-1}(X^{+}) 
 + e_n^{22} 
{\widetilde \Delta}_{1 \cdots n-1}(X^{+}) 
 {\widetilde {D^{+}}}_{1 \cdots n-1}(\xi_n)
\right) \widehat{\delta}^{1 2 \cdots n}_n \, .  
\ee
\label{lem:C1}
\end{lem2}

\begin{lem2} 
The diagonalized form of $\Delta^{(n-1)}(X^{+})$ is expressed in terms 
of local operators $X_i^{+}$ as follows. 
\be 
{\widetilde \Delta}_{1 \cdots n}(X^{+}) = \sum_{i=1}^{n} 
X_i^{+} \widehat{\delta}^{1 \cdots n}_i 
\ee
\label{lem:C2}
\end{lem2}

\begin{lem2} Operator $C^{+}$ in the $F$-basis is expressed 
in terms of operator $D^{+}$ and $X_i^{+}$ as follows. 
\be 
{\widetilde {C^{+}}}_{1 \cdots n}(\lambda) = \sum_{i=1}^{n}  
({\widetilde {D^{+}}}_{1 \cdots i-1, i+1 \cdots n}(\lambda)
 - q^{-1} {\widetilde {D^{+}} }_{1 \cdots n}(\lambda) ) \, X_i^{+} 
\widehat{\delta}_i^{1 \cdots n} \, .   
\ee
\label{lem:C3} 
\end{lem2}

\section*{Acknowledgment} 
The authors would like to thank Prof. S. Miyashita for 
encouragement and keen interest in this work. 
One of the authors (C.M.) would like to thank Dr. K. Shigechi 
for his introduction to the mathematical physics of integrable models.  
This work is partially supported by 
Grant-in-Aid for Scientific Research (C) No. 20540365.

%%%%%%%%%%%%%%%%%%%%%%%%%%%%%%%%%%%%%%%%%

\end{document}